\documentclass[10pt,pdftex,a4paper]{article}

\usepackage{mathtools}
\usepackage{amssymb}
\usepackage[utf8x]{inputenc}
\usepackage{amsmath}
\usepackage{amsfonts}
\usepackage{amsthm}
\usepackage{stmaryrd}
\usepackage{mathabx}
\usepackage{rotating}
\usepackage{ppprog}       %
\usepackage{shorthand}    %
\usepackage{veralgs}      %
\usepackage{caption}
\usepackage[OT1]{fontenc}
\usepackage{paralist}
\usepackage{notation}
\usepackage{xcolor}
\usepackage{mathtools}
\usepackage{framed}
\usepackage{accents}
\usepackage[a4paper, total={5.5in, 8in}]{geometry}

\colorlet{shadecolor}{yellow!20}
\usepackage{tikz}
\usetikzlibrary{arrows,backgrounds,calc,decorations.pathmorphing,fit,shapes.geometric,matrix,shapes,patterns,positioning}
\usepackage{transactiondiagram}
\usepackage{codeconfig}

 \newtheorem{theorem}{Theorem}
 \newtheorem{proposition}{Proposition}
 \newtheorem{lemma}{Lemma}
 \newtheorem{definition}{Definition}
 \newtheorem{corollary}{Corollary}

\usepackage{enumerate} %

\begin{document}

\title{Last-use Opacity: A Strong Safety Property for Transactional Memory~with Early Release Support}

\author{Konrad Siek,
Pawe\l{} T. Wojciechowski\\
Institute of Computing Science\\
Pozna\'n University of Technology
}

\maketitle
\begin{abstract}

\emph{Transaction Memory} (TM) is a concurrency control abstraction that allows
the programmer to specify blocks of code to be executed atomically as
transactions. However, since transactional code can contain just about any
operation attention must be paid to the state of shared variables at any given
time. E.g., contrary to a database transaction, if a TM transaction reads a
stale value it may execute dangerous operations, like attempt to divide by
zero, access an illegal memory address, or enter an infinite loop. Thus
serializability is insufficient, and stronger safety properties are required in
TM, which regulate what values can be read, even by transactions that abort.
Hence, a number of TM safety properties were developed, including opacity, and
TMS1 and TMS2. However, such strong properties preclude using early release as
a technique for optimizing TM, because they virtually forbid reading from live
transactions. On the other hand, properties that do allow early release are
either not strong enough to prevent any of the problems mentioned above
(recoverability), or add additional conditions on transactions with early
release that limit their applicability (elastic opacity, live opacity, virtual
world consistency). This paper introduces last-use opacity, a new TM safety
property that is meant to be a compromise between strong properties like
opacity and serializability. The property eliminates all but a small class of
inconsistent views and poses no stringent conditions on transactions.  For
illustration, we present a last-use opaque TM algorithm and show that it
satisfies the new safety property.

\end{abstract}

\section{Introduction}

Writing concurrent programs using the low-level synchronization primitives is
notoriously difficult and error-prone. Over the past decade, there has been a
growing interest in alternatives to lock-based synchronization by turning to the
idea of software \emph{transactional memory} (TM) \cite{HME93,ST95}.
Basically, TM transplants the transaction abstraction from database systems and
uses it to hide the details of synchronization. In particular, TM uses
speculative execution to ensure that transactions in danger of reading
inconsistent state abort and retry. This is a fairly universal solution and
means that the programmer must only specify where transactions begin and end,
and TM manages the execution so that the transactional code executes correctly
and efficiently. Thus, the programmer avoids having to solve the problem of
synchronization herself, and can rely on any one of a plethora of TM systems
(e.g., \cite{DSS06,HF03,HLMS03,HMJH04,NWA+08,RG05}).

Since TM allows transactional code to be mixed with non-transactional code and
to contain virtually any operation, rather than just reads and writes like in
its database predecessors, greater attention must be paid to the state of
shared variables at any given time. For instance, if a database transaction
reads a stale value, it must simply abort and retry, and no harm is done.
Whereas, if a TM transaction reads a stale value it may execute an
unanticipated dangerous operation, like dividing by zero, accessing an illegal
memory address, or entering an infinite loop. 
Thus, TM systems must restrict the ability of transactions to view
inconsistent state.

To that end, the safety property called opacity \cite{GK08,GK10} was
introduced, which includes the condition that transactions do not read values
written by other live (not completed) transactions alongside serializability
\cite{Pap79} and real-time order conditions.
Opacity became the gold standard of TM safety properties, and most TM systems
found in the literature are, in fact, opaque.
However opacity precludes \emph{early release}, an important programming
technique, where two transactions technically conflict but nevertheless both
commit correctly, and still produce a history that is intuitively correct.
Systems employing early release (e.g.
\cite{HLMS03,RRHW09,FGG09,BMT10,SW14-hlpp}) show that this yields a significant
and worthwhile performance benefit.
This is particularly (but not exclusively) true with pessimistic concurrency
control, where early release is vital to increased parallelism between
transactions, and therefore essential for achieving high efficiency in
applications with high contention.

Since opacity is a very restrictive property, a number of more relaxed
properties were introduced that tweaked opacity's various aspects to achieve a
more practical property. These properties include virtual world consistency
(VWC) \cite{IMR08}, transactional memory specification (TMS1 and TMS2)
\cite{DGLM13}, elastic opacity \cite{FGG09}, and live opacity \cite{DFK14}.
The first contribution of this paper is to examine these properties and
determine whether or not they allow the use of early release in TM, and, if so,
what compromises they make with respect to consistency, and what additional
assumptions they require.
We then consider the applicability of these properties to TM systems that rely
on early release. 
In addition to TM properties, we similarly examine common database
consistency conditions: serializability \cite{Pap79}, recoverability
\cite{Had88}, avoiding cascading aborts (ACA) \cite{BHG87} , strictness
\cite{BHG87}, and rigorousness \cite{BGRS91}.

The second contribution of this paper is to introduce a new TM safety property
called last-use opacity that allows early release without requiring stringent
assumptions but nevertheless eliminates inconsistent views or restricts them to
a manageable minimum. 
We give a formal definition, discuss example last-use opaque histories, and
compare the new property with existing TM properties, specifically showing that
it is stronger than serializability but weaker than opacity. 
We also describe the guarantees given by last-use opacity and consider the
applicability of the property in system models that either allow, deny, or
restrict the explicit programmatic abort operation.
Whereas, last-use opacity eliminates inconsistent views in system models
that forbid explicitly aborting transactions or restricts this to particular
scenarios, we show that allowing free use of explicit aborts can lead to
inconsistent views in last-use opaque histories. Thus, we also introduce a
stronger variant of the property called $\beta$--last-use opacity that
precludes them.

Finally, we give SVA \cite{SW15-atomicrmi}, a TM concurrency control algorithm
with early release and demonstrate that it satisfies last-use opacity.

The paper is structured as follows. We present the definitions of basic terms
in \rsec{sec:preliminaries}. We follow by an examination of the TM property
space in \rsec{sec:properties}. Next, we define and discuss last-use opacity
in \rsec{sec:lopacity}.  
Then, we present \SVA{} and demonstrate its correctness in
\rsec{sec:lopaque-sva}.
Finally, we present the related work in \rsec{sec:rw} and
conclude in \rsec{sec:conclusions}.
We also include an appendix containing additional proofs.
      
\section{Preliminaries}
\label{sec:preliminaries}

Before discussing properties and their relation to early release, let us
provide definitions of the relevant ancillary concepts.  

{
Let $\processes = \{\proc_1, \proc_2, ..., \proc_n\}$ be a set of processes.
Then, let program $\prog$ be defined as a set of subprograms $\prog = \{
\subprog_1, \subprog_2, ..., \subprog_n \}$ 
such that for each process $\proc_k$ in $\processes$
there is exactly one corresponding subprogram $\subprog_k$ in $\prog$ and \emph{vice versa}.   
Each subprogram $\subprog_k \in \prog$ is a finite sequence of statements in
some language $\lang$.
The definition of $\lang$ can be whatsoever, 
as long as it
provides constructs to execute 
operations on shared variables 
in accordance with the interface and assumptions described further in this section.
In particular, $\lang$
can allow local computations whose effects
are not visible outside a single processes.

Given program $\prog$ and a set of processes $\processes$, we denote an
execution of $\prog$ by $\processes$ as $\exec{\prog}{\processes}$.
An execution entails each process $\proc_k \in \processes$ evaluating some
prefix of subprogram $\subprog_k \in \prog$.
The evaluation of each statement by a process is deterministic
and follows the semantics of $\lang$. 
$\exec{\prog}{\processes}$ is concurrent, 
i.e. while the statements in subprogram $\subprog_k$ are evaluated sequentially by
a single process, the evaluation of statements by different processes can be
arbitrarily interleaved.
We call $\exec{\prog}{\processes}$ a \emph{complete} execution if each process
$\proc_k$ in $\processes$ evaluates all of the statements in $\subprog_k$.
Otherwise, we call $\exec{\prog}{\processes}$ a \emph{partial} execution.
}

\paragraph{Variables}
Let $\objects$ be a set of \emph{shared variables} (or \emph{variables}, in short).
Each variable, denoted as $\objx, \objy, \objz$ etc.,
supports the following \emph{operations}, denoted $o$, that allow to retrieve or modify its state:
\begin{enumerate}[a) ]
     \item \emph{write} operation $\wop{\obj}{\val}$ 
     that sets the
     state of $\obj$ to value $\val$; the operation's \emph{return value} is
     the constant $\ok$,
     \item \emph{read} operation $\rop{\obj}{}$ whose \emph{return value} 
     is the current state of $\obj$.
\end{enumerate}
In order to execute some operation $o$ on variable $\obj$, process $\proc_k$
issues an \emph{invocation event} %
denoted
 $\inv{}{k}{o}$, and receives a
\emph{response event} %
denoted $\res{}{k}{\valu}$, where
$\valu$ is the return value of $o$.  %
The pair of these events is called a
\emph{complete operation execution} and it is denoted 
$o^k \rightarrow \valu$, 
whereas an invocation event $\inv{}{k}{o}$
without the corresponding response event is called a 
\emph{pending operation execution}. 
Specifically, a complete execution of a read operation by process $\proc_k$ is denoted
$\pfrop{}{k}{\obj}{\val}$ and a complete execution of a write operation is denoted
$\pfwop{}{k}{\obj}{\val}{\ok}$.
We refer to complete and pending 
operation executions as \emph{operation executions}, denoted by $op$.

Each event is atomic and instantaneous, but the execution of the entire
operation composed of two events is not. 

\paragraph{Transactions}
\emph{Transactional memory (TM)} is a programming paradigm that uses transactions to
control concurrent execution of operations on shared variables by parallel
processes. 
A \emph{transaction} $\tr_i \in \transactions$ is some piece of code executed
by process $\proc_k$,
as part of subprogram $\subprog_k$.
Hence, we say that $\proc_k$ executes $\tr_i$. 
Process $\proc_k$ can execute
local computations as well as operations on shared variables as part of the
transaction. In particular, the processes can execute the following operations
as part of transaction $\tr_i$: 
\begin{enumerate}[a) ]
     \item $\init_i$ which initializes transaction $\tr_i$, and
     whose return value is the constant $\ok_i$,
     \item $\twop{i}{\obj}{\val}$ and $\trop{i}{\obj}{}$ 
     which respectively write a value $\val$ to variable $\obj$ 
     and read $\obj$ within transaction $\tr_i$, and return either 
     the operation's return value or the constant $\ab_i$,
     \item $\tryC_i$ which attempts to
     commit $\tr_i$ and returns either the constant $\co_i$ or the constant
     $\ab_i$.
\end{enumerate}
There is also another operation allowed in some TM systems and not in
others, and we wish to discuss it separately. Namely, some TMs allow for a
transaction to programmatically roll back by executing the operation:
\begin{enumerate}[a) ]
    \setcounter{enumi}{3}
    \item $\tryA_i$ which aborts $\tr_i$ and returns $\ab_i$.
\end{enumerate}
The constant $\ab_i$ indicates that transaction $\tr_i$ has been aborted, as
opposed to the constant $\co_i$ which signifies a successful commitment of the
transaction.

By analogy to processes executing operations on %
 variables, if process
$\proc_k$ executes some operation as part of transaction $\tr_i$ it issues an
invocation event of the form $\inv{i}{k}{\init_i}$,
$\inv{i}{k}{o}$ %
for some $\obj$, or $\inv{i}{k}{\tryC_i}$, (or
possibly $\inv{i}{k}{\tryA_i}$) and receives a response of the form
$\res{i}{k}{\valu_i}$, where $\valu_i$ is a value, or 
the constant $\ok_i$, $\co_i$, or $\ab_i$.  
The superscript always denotes which process executes the operation, and the
subscript denotes of which transaction the operation is a part.
We denote operation executions by process $\proc_k$ within transaction
$\tr_i$ as:
\begin{enumerate}[a) ]
    \item $\init_i^k \rightarrow \ok_i$,
    \item $\pfrop{i}{k}{\obj}{\val}\,$ or $\,\pfrop{i}{k}{\obj}{\ab_i}$,
    \item $\pfwop{i}{k}{\obj}{\val}{\ok_i}\,$ or $\,\pfwop{i}{k}{\obj}{\val}{\ab_i}$,
    \item $\tryC_i^k \rightarrow \co_i\,$ or $\,\tryC_i^k \rightarrow \ab_i$.
    \item $\tryA_i^k \rightarrow \ab_i$.
\end{enumerate}

TM assumes that processes execute operations on shared variables only as part
of a transaction. Furthermore, we assume that any
transaction $\tr_i$ is executed by exactly one process $\proc_k$ and that each
process executes transactions sequentially. 

Even though transactions are subprograms evaluated by processes, it is 
convenient to talk about them as separate and independent entities. 
Thus, rather than saying $\proc_k$ executes some operation as part of 
transaction $\tr_i$, we will simply say that $\tr_i$ executes (or performs) 
some operation. 
Hence we will also forgo the distinction of processes in transactional 
operation executions, and write simply:
    $\init_i \rightarrow \ok_i$,
    $\frop{i}{\obj}{\val}$, %
    $\fwop{i}{\obj}{\val}{\ok_i}$, %
    $\tryC_i \rightarrow \co_i$, etc. %
By analogy, we also drop the superscript indicating processes in the notation
of invocation and response events, unless the distinction is needed.

\paragraph{Sequential Specification}

Given variable $\obj$, let \emph{sequential specification} of $\obj$, denoted
$\mathit{Seq}(\obj)$, be a prefix-closed set of sequences containing %
invocation events and response events which specify the semantics of
shared variables.
(A set $Q$ of sequences is
\emph{prefix-closed} if, whenever a sequence $S$ is in $Q$, every prefix of $S$ is
also in $Q$.)
Intuitively, a sequential specification enumerates all possible correct
sequences of operations that can be performed on a variable in a sequential
execution. 
Specifically, 
given $D$, the domain of variable $\obj$, and $\val_0 \in D$, an
initial state of $\obj$, we denote by $\mathit{Seq}(\obj)$ 
the sequential specification of $\obj$ s.t., $\mathit{Seq}(\obj)$ is a set of
sequences of the form $[\alpha_1 \rightarrow \val_1, \alpha_2 \rightarrow
\val_2, ..., \alpha_m \rightarrow \val_m]$, where each $\alpha_j \rightarrow
\val_j$ ($j=1..m$) is either:
\begin{enumerate}[a) ]
    \item $\fwop{i}{\objx}{\val_j}{\ok_i}$, where $\val_j \in D$,
    or 
    \item $\frop{i}{\objx}{\val_j}$, and either the most recent
    preceding write operation is $\fwop{l}{\objx}{\val_j}{\ok_l}$ ($l<i$) or there are no preceding
    writes and $\val_j = \val_0$.
\end{enumerate}

From this point on we assume that the domain $D$  of all transactional
variables is the set of natural numbers $\mathbb{N}_0$ and that the initial
value $\val_0$ of each variable is $0$.

{
Even though we describe the interface and sequential specification of variables
to represent the behavior of registers, we do so out of convenience and our
conclusions can be trivially extended to other types of objects e.g. compare
and swap objects or stacks.
}

{

\paragraph{Histories}
A TM \emph{history} $\hist$ is a sequence of invocation and response events
issued by the execution of
 transactions $\transactions_\hist = \{ \tr_1, \tr_2, ..., \tr_t \} $.
The occurrence and order of events in $\hist$ is dictated by a given (possibly
partial) execution of some program $\prog$ by processes $\processes$. 
We denote by $\hist \models \exec{\prog}{\processes}$ that history $\hist$ is
produced by $\exec{\prog}{\processes}$.
Note, that different interleavings of processes in $\exec{\prog}{\processes}$
can produce different histories.
A \emph{subhistory} of a history H is a subsequence of H. 
}

The sequence of events in a history $\hist_j$ can be denoted as $\hist_j = [
e_1, e_2, ..., e_m ]$. 
For instance, some history $\hist_1$ below is a history of a run of some
program that executes transactions $\tr_1$ and $\tr_2$:
\begin{equation*}
\begin{split}
\hist_1 = [~&
              \inv{1}{}{\init_1}, \res{1}{}{\ok_1},
              \inv{2}{}{\init_2}, \res{2}{}{\ok_2},  \\
            & \inv{1}{}{\twop{1}{\obj}{\val}},
              \inv{2}{}{\trop{2}{\obj}{}}, 
              \res{1}{}{\ok_1}, 
              \res{2}{}{\val},                       \\
            & \inv{1}{}{\tryC_1},\res{1}{}{\co_1},
              \inv{2}{}{\tryC_2},\res{2}{}{\co_2}    ~].
\end{split}
\end{equation*}

Given any history $\hist$, let $\hist|\tr_i$ be the longest 
subhistory of $\hist$ consisting only of invocations and responses executed 
by transaction $\tr_i$.
For example,  $\hist_1|\tr_2$ is defined as:
\begin{equation*}
\begin{split}
\hist_1|\tr_2 = [~ &
                     \inv{2}{}{\init_2}, \res{2}{}{\ok_2},  
                     \inv{2}{}{\trop{2}{\obj}{}}, \res{2}{}{\val},                       
                    \inv{2}{}{\tryC_2},\res{2}{}{\co_2}    ~].
\end{split}
\end{equation*}
We say transaction $\tr_i$ \emph{is in} $\hist$, which we denote $\tr_i \in \hist$, if $\hist|\tr_i \neq \varnothing$.
Let $\hist|\proc_k$ be the longest subhistory of $\hist$ consisting only of
invocations and responses executed by process $\proc_k$.
Let $\hist|\obj$ be the longest subhistory of $\hist$ consisting only of
invocations and responses executed on variable $\obj$, but only those that 
form complete operation executions.

Given complete operation execution $op$ that consists of an invocation
event $e'$ and a response event $e''$, we say $op$ \emph{is in} $\hist$ 
($\op \in \hist$) if $e' \in \hist$ and $e'' \in \hist$.
Given a pending operation execution $op$ consisting of an invocation $e'$, 
we say $op$ \emph{is in} $\hist$ ($\op \in \hist$) if $e' \in \hist$ and 
there is no other operation execution $op'$ consisting of an invocation 
event $e'$ and a response event $e''$ s.t. $\op' \in \hist$.      

Given two complete operation executions $\op'$ and $\op''$ in some history 
$\hist$, where $\op'$ contains the response event $\mathit{res}'$ and 
$\op''$ contains the invocation event $\mathit{inv}''$, 
we say $\op'$ \emph{precedes} $\op''$ in $\hist$ if
$\mathit{res}'$ precedes $\mathit{inv}''$ in $\hist$.

A history whose all operation executions are complete is a \emph{complete}
history.

Most of the time it will be convenient to denote any two adjoining events in a
history that represent the invocation and response of a complete execution of
an operation as that operation execution, using the syntax $e \rightarrow e'$. 
Then, an alternative representation of $\hist_1|\tr_2$ is denoted as follows: 
$$\hist_1|\tr_2 = [~
        \init_2 \rightarrow \ok_2,~
        \frop{2}{\obj}{\val},~
        \tryC_2 \rightarrow \co_2~].$$

History $\hist$ is \emph{well-formed} if, for every transaction $\tr_i$ in
$\hist$, $\hist|\tr_i$ is an alternating sequence of invocations and responses
s.t., 
\begin{enumerate}[a) ] 
    \item $\hist|\tr_i$ starts with an invocation $\inv{i}{}{\init_i}$, 
    \item no events in  $\hist|\tr_i$ follow $\res{i}{}{\co_i}$ or $\res{i}{}{\ab_i}$,
    \item no invocation event in $\hist|\tr_i$ follows $\inv{i}{}{\tryC_i}$ or $\inv{i}{}{\tryA_i}$,
    \item for any two transactions $\tr_i$ and $\tr_j$ s.t., $\tr_i$ and
    $\tr_j$ are executed by the same process $\proc_k$, the last event of
    $\hist|\tr_i$ precedes the first event of $\hist|\tr_j$ in $\hist$ or \emph{vice
    versa}.
\end{enumerate}
In the remainder of the paper we assume that all histories are well-formed.

\paragraph{History Completion}
Given history $\hist$ and transaction $\tr_i$, $\tr_i$ is \emph{committed} if
$\hist|\tr_i$ contains operation execution $\tryC_i \rightarrow \co_i$.
Transaction $\tr_i$ is \emph{aborted} if $\hist|\tr_i$ contains response
$\res{i}{}{\ab_i}$ to any invocation. Transaction $\tr_i$ is
\emph{commit-pending} if $\hist|\tr_i$ contains invocation $\tryC_i$ but it
does not contain $\res{i}{}{\ab_i}$ nor $\res{i}{}{\co_i}$.  Finally,
$\tr_i$ is \emph{live} if it is neither committed, aborted, nor
commit-pending.

Given two histories $\hist' = [e_1', e_2', ..., e_m']$ and $\hist'' = [e_1'',
e_2'', ..., e_m'']$, we define their concatenation as $\hist' \cdot \hist'' =
[e_1', e_2', ..., e_m', e_1'', e_2'', ..., e_m'']$.
We say $P$ is a prefix of $\hist$ if $\hist = P \cdot \hist'$. 
Then, let a \emph{completion} $\mathit{Compl}(\hist)$ of history $\hist$ be any
complete history s.t., $\hist$ is a prefix of
$\mathit{Compl}(\hist)$ and for every transaction $\tr_i \in \hist$ subhistory
$\mathit{Compl}(\hist)|\tr_i$ equals one of the following:
\begin{enumerate}[a) ]
    \item $\hist|\tr_i$, if $\tr_i$ finished committing or aborting,
    \item $\hist|\tr_i \cdot [\res{i}{}{\co_i}]$, if $\tr_i$ is live and contains a
            pending $\tryC_i$,
    \item $\hist|\tr_i \cdot [\res{i}{}{\ab_i}]$, if $\tr_i$ is live and contains some 
            pending operation,
    \item $\hist|\tr_i \cdot [\tryC_i \rightarrow \ab_i]$, if $\tr_i$  is live and
            contains no pending operations.
\end{enumerate}
Note that, if all transactions in $\hist$ are committed or aborted then
$\mathit{Compl}(\hist)$ and $\hist$ are identical.

Two histories $\hist'$ and $\hist''$ are \emph{equivalent} (denoted $\hist'
\equiv \hist''$) if for every $\tr_i \in \transactions$ it is true that
$\hist'|\tr_i = \hist''|\tr_i$.
When we write $\hist'$ is equivalent to $\hist''$ we mean that $\hist'$ and $\hist''$ are equivalent.

\paragraph{Sequential and Legal Histories} 

A \emph{real-time order} $\prec_\hist$ is an order over history $\hist$ s.t., given two
transactions $\tr_i, \tr_j \in \hist$, if the last event in $\hist|\tr_i$
precedes in $\hist$ the first event of $\hist|\tr_j$, then 
$\tr_i$ \emph{precedes} $\tr_j$ in $\hist$, denoted $\tr_i \prec_\hist
\tr_j$.
We then say that two transactions  $\tr_i, \tr_j \in \hist$ are \emph{concurrent} if neither
$\tr_i \prec_\hist \tr_j$ nor $\tr_j \prec_\hist \tr_i$.
We say that history $\hist'$ \emph{preserves the real-time order} of $\hist$ if
$\prec_\hist \subseteq \prec_{\hist'}$. 
A \emph{sequential history} $S$ is a history, s.t. no two transactions in $S$
are concurrent in $S$. 
Some sequential history $S$ is a \emph{sequential extension} of $\hist$ if $S$
is equivalent to $\hist$ and $S$ preserves the real time order of $\hist$.

We analogously define a real-time order $\prec_\hist$ of operation executions over history $H$.

Let $S'$ be a sequential history that only contains committed transactions,
with the possible exception of the last transaction, which can be aborted.
We say that sequential history $S'$ is \emph{legal} if for every $\objx \in
\objects$, $S'|\objx \in \mathit{Seq}(\objx)$.

Using the definitions above allows us to formulate the central concept that
defines consistency in opacity: \emph{transaction legality}.
Intuitively, we can say a transaction is legal in a sequential history if it
only reads values of variables that were written by committed transactions or by itself.
More formally, given a sequential history $S$ and a transaction $\tr_i \in S$,
we then say that transaction $\tr_i$ is \emph{legal in} $S$ if $\vis{S}{\tr_i}$ is
legal, where $\vis{S}{\tr_i}$ is the longest subhistory $S'$ of $S$ s.t., for
every transaction $\tr_j \in S'$, either $i=j$ or $\tr_j$ is committed in $S'$
and $\tr_j\prec_S \tr_i$.

\paragraph{Unique Writes}

History $\hist$ has \emph{unique writes} if, given transactions $\tr_i$ and
$\tr_j$ (where $i\neq j$ or $i=j$), for any two write operation executions
$\fwop{i}{\obj}{\val'}{\ok_i}$ and $\fwop{j}{\obj}{\val''}{\ok_j}$ it is true
that $\val' \neq \val''$ and neither $\val' = \val_0$ nor $\val'' = \val_0$.

For the remainder of the paper we focus exclusively on histories with unique writes.
This assumption does not reduce generality, in that any history without unique
writes trivially can be transformed into a history with unique writes (for
instance, by appending a timestamp to each written value).

\paragraph{Accesses}

Given a history $\hist$ and a transaction $\tr_i$ in $\hist$, we say that $\tr_i$
\emph{reads} variable $\obj$ in $\hist$ if there exists an invocation
$\inv{i}{}{\trop{i}{\obj}{}}$ in $\hist|\tr_i$.
By analogy, we say that $\tr_i$ \emph{writes} to $\obj$ in $\hist$ if there
exists an invocation $\inv{i}{}{\twop{i}{\obj}{\val}}$ in $\hist|\tr_i$.
If $\tr_i$ reads $\obj$ or writes to $\obj$ in $\hist$, we say $\tr_i$
\emph{accesses} $\obj$ in $\hist$.
In addition, 
let $\tr_i$'s \emph{read set} be a set that contains every variable $\obj$,
s.t.  $\tr_i$ reads $\obj$. By analogy, $\tr_i$'s \emph{write set} contains
every $\obj$, s.t. $\tr_i$ writes to $\obj$. A transaction's \emph{access set},
denoted $\accesses{\tr_i}$, is the union of its read set and its write set.

Given a history $\hist$ and a pair of transactions $\tr_i, \tr_j \in \hist$, we
say $\tr_i$ and $\tr_j$ \emph{conflict} on variable $\obj$ in $\hist$ if
$\tr_i$ and $\tr_j$ are concurrent, both $\tr_i$ and $\tr_j$ access $\obj$, and
one or both of $\tr_i$ and $\tr_j$ write to $\obj$.

Given a history $\hist$ (with unique writes) and a pair of transactions $\tr_i,
\tr_j \in \hist$, we say $\tr_i$ \emph{reads from} $\tr_j$ if there is some
variable $\obj$, for which 
there is a complete operation execution $\fwop{j}{\obj}{\val}{\ok_j}$ in $\hist|\tr_j$
and another complete operation execution $\frop{i}{\obj}{\valu}$ in
$\hist|\tr_i$, s.t. $\val = \valu$.

Given any transaction $\tr_i$ in some history $\hist$ (with unique writes) any
operation execution on a variable $\obj$ within $\hist|\tr_i$ is either
\emph{local} or \emph{non-local}. Read operation execution
$\frop{i}{\obj}{\val}$ in $\hist|\tr_i$ is local if it is preceded in
$\hist|\tr_i$ by a write operation execution on $\obj$, and it is non-local
otherwise. Write operation execution $\fwop{i}{\obj}{\val}{\ok_i}$ in
$\hist|\tr_i$ is local if it is followed in $\hist|\tr_i$ by an invocation of a
write operation on $\obj$, and non-local otherwise.

\paragraph{Safety Properties}

A \emph{property} $\property$ is a condition that stipulates correct behavior.
In relation to histories, a given history satisfies $\property$
if the condition is met for that history.
In relation to programs, program $\prog$ satisfies $\property$ if all histories
produced by $\prog$ satisfy $\property$. 

Safety properties \cite{Lam77} are properties which guarantee that
"something [bad] will not happen." In the case of TM this means that,
transactions will not observe concurrency of other transactions.
Property $\property$ is a safety property if it meets the following definition
(adapted from \cite{AHKR13}):
\begin{definition} \label{def:safety-property}
    A property $\property$ is a \emph{safety property} if, given the set
    $\mathbb{H}_\property$ of all histories that satisfy $\property$:
    \begin{enumerate}[a) ]
        \item\emph{Prefix-closure}: every prefix $H'$ of a history $\hist \in
            \mathbb{H}_\property$ is also in $\mathbb{H}_\property$,
        \item\emph{Limit-closure}: for any infinite sequence of finite histories
            $\hist_0, \hist_1, ...$, s.t. for every $\hist_h \in
            \mathbb{H}_\property$ and $\hist_h$ is a prefix of $\hist_{h+1}$,
            the infinite history that is the \emph{limit} of the sequence is also in
            $\mathbb{H}_\property$.
    \end{enumerate}
\end{definition}
For distinction, in the remainder of the paper we refer to properties that are
not safety properties as \emph{consistency conditions}.
  
\section{Early Release}
\label{sec:properties}

In this section we discuss whether existing safety properties and 
consistency conditions allow for early release (extending our work in
\cite{SW14-wttm}) and to what extent. 
The aim of the analysis is to find properties that describe the guarantees of
TM systems with early release that can be applied in practice. 
That is, we seek a safety property that allows early release but reduces or
eliminates undesired behaviors.

Early release pertains to a situation where conflicting transactions execute
partially in parallel while accessing the same variable. The implied intent is
for all such transactions to access these variables without losing consistency
and thus for them all to finally commit.
We define the concept of early release as follows:

\begin{definition}[Early Release] \label{def:early-release}
    Given history $\hist$
    (with unique writes), transaction $\tr_i \in \hist$
    \emph{releases variable} $\obj$ \emph{early} in $\hist$ iff there is some
    prefix $P$ of $\hist$, such that $\tr_i$ is live in $P$ and there exists
    some transaction $\tr_j \in P$ such that there is a 
    complete non-local read operation execution $\op_j = \frop{j}{\obj}{\val}$ in $P|\tr_j$ 
    and a 
    write operation execution $\op_i = \fwop{i}{\obj}{\val}{\ok_i}$ in $P|\tr_i$ 
    such that $\op_i$ precedes $\op_j$ in $P$.
\end{definition}

We begin our analysis by defining its key questions. The first and the most
obvious is whether a particular property supports early release at all. This is
defined as follows:

\begin{definition}[Early Release Support] \label{def:early-release-support}
    \label{def:release-support}
    Property $\property$ supports early release iff
    given some history $\hist$ that satisfies $\property$
    there exists some transaction $\tr_i \in \hist$, s.t.
    $\tr_i$ releases some variable $\obj$ early in $\hist$.
\end{definition}

If a property allows early release, it allows a significant performance boost (e.g.
\cite{RRHW09,SW14-hlpp}) as transactions are executed with a higher degree of
parallelism.
However, early release can give rise to some unwanted or unintuitive scenarios
with respect to consistency. The most egregious of these is \emph{overwriting},
where one transaction releases some variable early, but proceeds to modify it
afterward. In that case, any transaction that started executing operations on
the released variable will observe an intermediate value with respect to the
execution of the other transaction, ie., \emph{view inconsistent state}.

\begin{figure}
\begin{center}
\begin{tikzpicture}
     \draw
           (0,2)        node[tid]       {$\tr_i$}
                        node[aop]       {$\init_i$} %
                        node[dot]       {} 

      -- ++(1.25,0)     node[aop]       {$\twop{i}{\obj}{1}$}
                        node[dot]       {}

      -- ++(2.0,0)      node[aop]       {$\twop{i}{\obj}{2}$}
                        node[dot] (wi)  {}

      -- ++(1.25,0)     node[aop]       {$\tryC_i\!\to\!\co_i$}
                        node[dot]       {}         
                        ;

     \draw
           (1.5,1)      node[tid]       {$\tr_j$}
                        node[aop]       {$\init_j$} %
                        node[dot]       {} 

      -- ++(1.25,0)     node[aop]       {$\trop{j}{\obj}{1}$}
                        node[dot]       {}

      -- ++(1.25,0)     node[aop]       {$\twop{j}{\obj}{3}\!\to\!\ab_j$}
                        node[dot] (wj)  {}
                        ;

     \draw[hb] (wi) \squiggle (wj);

     \draw
           (7,1)        node[tid]       {$\tr_{j'}$}
                        node[aop]       {$\init_{j'}$} %
                        node[dot] (sk)  {} 

      -- ++(1.25,0)     node[aop]       {$\trop{j'}{\obj}{2}$}
                        node[dot]       {}

      -- ++(1.25,0)     node[aop]       {$\twop{j'}{\obj}{4}$}
                        node[dot]       {}

      -- ++(1.25,0)     node[aop]       {$\tryC_{j'}\!\to\!\co_{j'}$}
                        node[dot]       {}
                        ;

     \draw[retry] (wj) -- (sk);   

\end{tikzpicture}
\end{center}
\caption{\label{fig:serializable-history} \label{fig:overwriting-history}
History with early release and overwriting.
The diagram depicts some history $\hist$ presented as operations executed by
transactions on a time axis. Every line depicts the operations executed by a
particular transaction, e.g., the line marked $\tr_i$ depicts subhistory
$\hist|\tr_i$.
The symbol
\protect\tikz{
    \protect\draw[] (0,0) -- ++(0.25,0) node[dot] {} -- ++(0.25,0);
} 
denotes a complete operation execution (an invocation event immediately
followed by a response event).
For brevity, whenever the response event of some operation
execution is $\ok_i$ we omit it, eg., we write $\twop{i}{\obj}{1}$ rather than
$\twop{i}{\obj}{1}\!\to\!{\ok_i}$.
We also shorten the representation of complete read operation
executions, so that eg. $\frop{j}{\obj}{1}$ is represented as
$\trop{j}{\obj}{1}$.
The arrow 
\protect\tikz{
    \protect\draw[hb] (0,0.2) .. controls +(270:.25) and +(90:0.25) .. (0.5,0.0);
} 
is used to emphasize a happens before relation, and
 \protect\tikz{
    \protect\draw[retry] (0,0) -- (1,0.0);
} 
denotes that the preceding transaction aborts (here, $\tr_j$) and a new
transaction ($\tr_{j'}$) is spawned.
}
\end{figure}
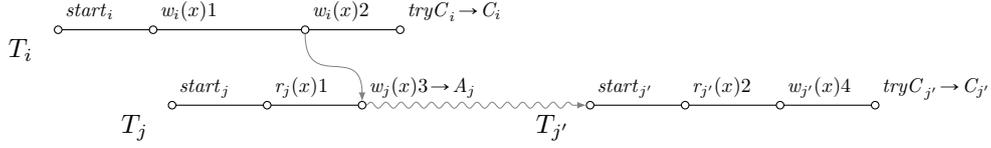

An example of overwriting is shown in \rfig{fig:overwriting-history}, where
transaction $\tr_i$ releases variable $\obj$ early but continues to write to
$\obj$ afterward. As a consequence, $\tr_j$ first reads the value of $\obj$
that is later modified. When $\tr_j$ detects it is in conflict while executing
a write operation it is aborted. This is a way for the TM to attempt to
mitigate the consequences of viewing inconsistent state. The transaction is
then restarted as a new transaction $\tr_{j'}$.

However, as argued in \cite{GK10}, simply aborting a transaction that views
inconsistent state is not enough, since the transaction can potentially act
in an unpredictable way on the basis of using an inconsistent value to perform
local operations. 
For instance, if the value is used in pointer arithmetic it is possible for the
transaction to 
access an unexpected memory location and crash the process.
Alternatively, if the transaction uses the value within a loop condition, it
can enter an infinite loop and become parasitic.

Thus, in our analysis of existing properties we ask the question whether, apart from
allowing early release, the properties also forbid overwriting. In the light of the
potential dangerous behaviors that can be caused by it, we consider properties
that allow overwriting to be too weak to be practical.

\begin{definition}[Overwriting Support] \label{def:overwriting-support}
    Property $\property$ supports overwriting iff 
    $\property$ supports early release, and 
    given some history $\hist$ (with early release) that satisfies $\property$,
    for some pair of transactions $\tr_i, \tr_j \in \hist$ 
    s.t., 
    \begin{enumerate}[a) ]
    \item $\tr_i$ releases some variable $\obj$ early,
    \item $\hist|\tr_i$ contains two write operation executions:
        $\fwop{i}{\obj}{\val}{\ok_i}$ and
        $\fwop{i}{\obj}{\val'}{\ok_i}$, s.t. the former precedes the latter in
        $\hist|\tr_i$,
    \item $\hist|\tr_j$ contains a read operation execution
        $\frop{j}{\obj}{\val}$ that precedes $\fwop{i}{\obj}{\val'}{\ok_i}$ in
        $\hist.$
    \end{enumerate}
\end{definition}

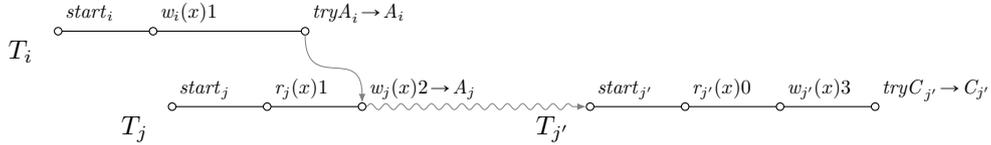
\begin{figure}
\begin{center}
\begin{tikzpicture}
     \draw
           (0,2)        node[tid]       {$\tr_i$}
                        node[aop]       {$\init_i$} %
                        node[dot]       {} 

      -- ++(1.25,0)     node[aop]       {$\twop{i}{\obj}{1}$}
                        node[dot]       {}

      -- ++(2.00,0)     node[aop]       {$\tryA_i\!\to\!\ab_i$}
                        node[dot] (ab)  {}         
                        ;

     \draw
           (1.5,1)      node[tid]       {$\tr_j$}
                        node[aop]       {$\init_j$} %
                        node[dot]       {} 

      -- ++(1.25,0)     node[aop]       {$\trop{j}{\obj}{1}$}
                        node[dot]       {}

      -- ++(1.25,0)     node[aop]       {$\twop{j}{\obj}{2}\!\to\!\ab_j$}
                        node[dot] (wj)  {}
                        ;

     \draw[hb] (ab) \squiggle (wj);

     \draw
           (7,1)        node[tid]       {$\tr_{j'}$}
                        node[aop]       {$\init_{j'}$} %
                        node[dot] (sk)  {} 

      -- ++(1.25,0)     node[aop]       {$\trop{j'}{\obj}{0}$}
                        node[dot]       {}

      -- ++(1.25,0)     node[aop]       {$\twop{j'}{\obj}{3}$}
                        node[dot]       {}

      -- ++(1.25,0)     node[aop]       {$\tryC_{j'}\!\to\!\co_{j'}$}
                        node[dot]       {}
                        ;

     \draw[retry] (wj) -- (sk);   

\end{tikzpicture}
\end{center}
\caption{\label{fig:cascading-abort-history}
History with early release and cascading abort.
}
\end{figure}

In addition, we look at whether or not a particular property forbids a transaction
that releases some variable early to abort. This is a precaution taken by many
properties to prevent \emph{cascading aborts}, another type of scenario leading to
inconsistent views. 
An example of this is shown in \rfig{fig:cascading-abort-history}.
In such a case a transaction, here $\tr_i$, releases a variable early and
subsequently aborts. This can cause another transaction $\tr_j$ that executed
operations on that variable in the meantime to observe inconsistent state.  In
order to maintain consistency, a TM will then typically force $\tr_j$ to abort
and restart as a result. 

However, while the condition that no transaction that releases early can abort,
solves the problem of cascading aborts, it significantly limits the usefulness
of any TM that satisfies it, since TM systems typically cannot predict whether
any particular transaction eventually commits or aborts.
In particular, there are important applications for
TM, where a transaction can arbitrarily and uncontrollably abort at any time.
Such applications include distributed TM and hardware TM, where aborts can be
caused by outside stimuli, such as machine crashes.

An exception to this may be found in systems making special provisions to
ensure that irrevocable transactions eventually commit (see e.g.,
\cite{WSA08}).  In such systems, early release transactions could be ensured
never to abort. However, case in point, these take drastic measures to ensure
that, e.g., at most a single irrevocable transaction is present in the system
at one time.
Therefore, the requirement may be difficult to enforce.
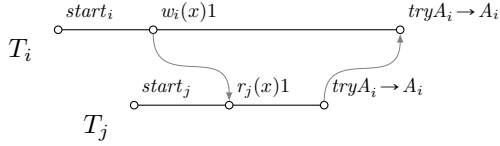
\begin{figure}
\begin{center}
\begin{tikzpicture}
     \draw
           (0,2)        node[tid]       {$\tr_i$}
                        node[aop]       {$\init_i$} %
                        node[dot]       {} 

      -- ++(1.25,0)     node[aop]       {$\twop{i}{\obj}{1}$}
                        node[dot] (wi)  {}

      -- ++(3.25,0)     node[aop]       {$\tryA_i\!\to\!\ab_i$}
                        node[dot] (ai)  {}         
                        ;

     \draw
           (1,1)        node[tid]       {$\tr_j$}
                        node[aop]       {$\init_j$} %
                        node[dot]       {} 

      -- ++(1.25,0)     node[aop]       {$\trop{j}{\obj}{1}$}
                        node[dot] (rj)  {}

      -- ++(1.25,0)     node[aop]       {$\tryA_i\!\to\!\ab_i$}
                        node[dot] (aj)  {}
                        ;

     \draw[hb] (wi) \squiggle (rj);
     \draw[hb] (aj) \squiggleup (ai);

\end{tikzpicture}
\end{center}
\caption{\label{fig:nonaborting-precluded-history}
History with an aborting early release transaction.
}
\end{figure}

Finally, the requirement that transactions which released early must not abort
precludes some scenarios that are intuitively correct. 
For instance, take the example in \rfig{fig:nonaborting-precluded-history}.
Here, $\tr_i$ writes $1$ to $\obj$ and releases it early. $\tr_j$ reads 
$1$
from $\obj$  and then aborts by executing the $\tryA$
operation, which also causes $\tr_i$ to abort.
Since $\tr_j$ reads from $\tr_i$ while the latter is live, $\tr_i$ releases
early in this history.
Then, if there is a requirement that transactions which release early not
abort, then this history is an incorrect one.
However, since $\tr_j$ aborted on its own accord, there are no transactions
that would be affected by $\tr_i$ aborting later on. Hence, intuitively, the
history is actually correct.
Thus, we consider the requirement that transactions which release early must
not abort to be overstrict.

Hence, we seek properties that allow aborts in transactions that release early.

\begin{definition}[Aborting Early Release Support] \label{def:aborting-support}
    Property $\property$ supports aborting early release iff 
    $\property$ supports early release, and 
    given some history $\hist$ that satisfies $\property$,
    for some transaction $\tr_i \in \hist$ that releases some variable $\obj$ early, 
    $\hist|\tr_i$ contains $\ab_i$.
\end{definition}

The properties under consideration are the typical TM safety properties:
serializability, opacity, transactional memory
specification, virtual world consistency, and elastic opacity.
Furthermore, we examine some of the family of live properties from
\cite{DFK14}, since this recent work introduces a number of relaxed versions of
TM safety properties with the view of accommodating early release.
Finally, we consider some strong database consistency conditions that pertain to
transactional processing: recoverability, avoiding cascading aborts,
strictness, and rigorousness.

\subsection{Serializability}
The first property we consider is serializability,
which can be regarded as a baseline TM
safety property. It is defined in \cite{Pap79} in three variants: conflict
serializability, view serializability, and final-state serializability. We
follow a more general version of serializability defined in \cite{Wei89} (as
\emph{global atomicity}), which we adjust to account for non-atomicity of
commits in our model. 

\begin{definition} [Serializability] \label{def:serializability}
    History $\hist$ is serializable iff there exists some sequential history
    $S$ equivalent to a completion $\compl{\hist}$ such that any committed
    transaction in $S$ is legal in $S$.
\end{definition}

This definition does not preclude early release, as long as illegal
transactions are aborted. Serializability also 
permits overwriting and cascading aborts.

\begin{theorem} \label{thm:serializability-early-release}
    Serializability supports early release.
\end{theorem}

\begin{proof} %
    Let $\hist$ be a transactional history as shown in
    \rfig{fig:serializable-history}.
    Note that since all transactions in $\hist$ are committed or aborted then
    $\hist = \compl{\hist}$.  Then, let there be a sequential history
    $S = \hist|\tr_i \cdot \hist|\tr_j \cdot \hist|\tr_{j'}$. Note that $S
    \equiv \hist$.
    Trivially, all the committed transactions in $S$, i.e. $\tr_i$ and $\tr_j'$, are
    legal in $S$, so $\hist$ is serializable.
    Since, by \rdef{def:early-release}, $\tr_i$ releases early in $\hist$, then,
    by \rdef{def:release-support}, serializability supports early release.
\end{proof}

\begin{theorem} \label{thm:serializability-overwriting}
    Serializability supports overwriting.
\end{theorem}

\begin{proof} %
    Let $\hist$ be a serializable history as in the proof of
    \rthm{thm:serializability-early-release} above.
    Transaction $\tr_i$ writes $1$ to $\obj$ in $\hist$ 
    prior to $\tr_j$
    reading $1$ from $\obj$, and  subsequently $\tr_i$ writes $2$ to $\obj$.
    Thus, according to \rdef{def:overwriting-support}, serializability supports
    overwriting.
\end{proof}

\begin{theorem} \label{thm:serializability-aborting}
    Serializability supports aborting early release. 
\end{theorem}

\begin{proof} %
    Let $\hist$ be a history such as the one in
    \rfig{fig:cascading-abort-history}.
    Since all transactions in $\hist$ are committed or aborted then $\hist =
    \compl{\hist}$. Then, let $S$ be a sequential history equivalent to
    $\hist$ such that $S = \hist|\tr_i \cdot \hist|\tr_j \cdot \hist|\tr_{j'}$.
    $S$ contains only one committed transaction $\tr_{j'}$, which is trivially
    legal in $S$.
    Thus $\hist$ is serializable.
    In addition, transaction $\tr_i$ in $S$ both releases $\obj$ early
    (\rdef{def:early-release})    and contains an abort ($\ab_i \in
    \hist|\tr_i$).  Thus, by \rdef{def:aborting-support}, serializability
    supports aborting early release.
\end{proof}

\subsection{Opacity}
\label{sec:opacity}
Opacity \cite{GK08,GK10} can be considered the standard TM safety property that
guarantees serializability and preservation of real-time order, and prevents
reading from live transactions. It is defined by the following two
definitions. The first definition specifies \emph{final state opacity} that
ensures the appropriate guarantees for a complete transactional history. The
second definition uses final state opacity to define a safety property
that is prefix closed. Both definitions follow those in \cite{GK10}.

\begin{definition} [Final state opacity] \label{def:final-state-opacity} \label{def:fs-opacity}
    A finite TM history $\hist$ is final-state opaque if, and only if, there
    exists a sequential history $S$ equivalent to any completion of $\hist$
    s.t., 
    \begin{enumerate}[(a)] 
        \item $S$ preserves the real-time order of $\hist$,
        \item every transaction $\tr_i$
        in $S$ is legal in $S$. 
    \end{enumerate}
\end{definition}

\begin{definition} [Opacity] \label{def:opacity}
    A TM history $\hist$ is opaque if, and only if, every finite prefix of
    $\hist$ is final-state opaque.
\end{definition}

\begin{theorem} \label{thm:opacity-early-release}
    Opacity does not support early release.
\end{theorem}

\begin{proof} %
By contradiction let us assume that opacity supports early release. Then,
from \rdef{def:early-release-support}, there
exists some history $\hist$ (with unique writes), s.t. $\hist$ is opaque and
there exists some transaction $\tr_i \in \hist$ that releases some variable
$\obj$ early in $\hist$.

From \rdef{def:early-release}, this implies that there exists some prefix $P$
of $\hist$ s.t. 
\begin{enumerate}[a) ]
    \item there is an operation execution $\op_i = \fwop{i}{\obj}{\val}{\ok_i}$
        and $\op_i \in P|\tr_i$,
    \item there exists a transaction $\tr_j \in P$ ($i\neq j$) and an operation
        execution $\op_j = \frop{j}{\obj}{\val}$, s.t. $\op_j \in P|\tr_j$ and
        $\op_i$ precedes $\op_j$ in $P$,
    \item $\tr_i$ is live in $P$.
\end{enumerate}
Let $P_c$ be any completion of $P$. Since $\tr_i$ is live in $P$, by definition
of completion, it is necessarily aborted in $P_c$ (ie. $A_i \in P_c|\tr_i$).
Given any sequential history $S$ equivalent to $P_c$, since $\tr_i$ is aborted
in $P_c$ and $\vis{S}{\tr_j}$ only contains operations of committed
transactions, then $P_c|\tr_i \nsubseteq \vis{S}{\tr_j}$.
This means that $\op_j \in \vis{S}{\tr_j}$ but $\op_i \not\in \vis{S}{\tr_j}$,
so $\vis{S}{\tr_j} \nsubseteq \mathit{Seq}(\obj)$ and therefore $\vis{S}{
\tr_j}$ is not legal.

On the other hand, \rdef{def:opacity} implies that any prefix $P$ of $\hist$ is
final state opaque, which, by \rdef{def:final-state-opacity}, implies that
there exists some completion $P_c$ of $P$ for which there exists an equivalent
sequential history $S$ s.t., any $\tr_j$ in $S$ is legal in $S$.
Since any $\tr_j$ is legal then for any $\tr_j$, $\vis{S}{\tr_j}$ is legal.
This is a contradiction with the paragraph above.
Thus, there cannot exist a history like $\hist$ that is both opaque and contains a
transaction that releases some variable early.
\end{proof}

Since both \rdef{def:overwriting-support} and \rdef{def:aborting-support}
require early release support, then:

\begin{corollary} \label{thm:opacity-overwriting}
    Opacity does not support overwriting.
\end{corollary}

\begin{corollary} \label{thm:opacity-aborting}
    Opacity does not support aborting early release.
\end{corollary}

\subsection{TMS1 and TMS2}

In \cite{DGLM13} the authors argue that some scenarios, such as sharing
variables between transactional and non-transactional code, require additional
safety properties. Thus, they propose and rigorously define two consistency
conditions for TM: \emph{transactional memory specification 1 (TMS1)} and
\emph{transactional memory specification 2 (TMS2)}.

TMS1 follows a set of design principles including a requirement for observing
consistent behavior that can be justified by some serialization. Among others,
TMS1 also requires that partial effects of transactions are hidden from other
transactions. These principles are reflected in the definition of the TMS1
automaton, and we paraphrase the relevant parts of the condition for the
correctness of an operation's response in the following definitions (see the
definitions of \textit{extConsPrefix} and \textit{validResp} for TMS1 in
\cite{DGLM13}).

Given a history $\hist$ and some response event $r$ in $H$, let $\hist\upto r$ denote
a subhistory of $\hist$ s.t. for every operation execution $\op \in \hist$, $\op
\in \hist\upto r$ iff $\op \prec_H r$ and $op$ is complete. This represents all
operations executed ,,thus far,'' when considering the legality of $r$.

Let $\transactions^d_\hist$ be the set of all transactions in $\hist$ s.t. $\tr_k \in
\transactions^d_\hist$ iff $\tr_k \in \hist$ and 
$\inv{k}{}{\tryC_k} \in \hist|\tr_k$.
Given response event $r$, let $\transactions^d_\hist\upto r$ be the set of all
transactions in $\hist$ s.t. $\tr_k \in \transactions^d_\hist\upto r$ if $\tr_k \in
\transactions^d_\hist$ and 
$\inv{k}{}{\tryC_k} \prec_{\hist} r$.
These sets represent transactions which committed or aborted (but are not live)
and the set of all such transactions that did so before response event $r$.

Given some history $\hist$, let $\transactions_\hist'$ by any subset of transactions
in $\hist$.  Let $\sigma$ be a sequence of transactions.
Let $\mathit{ser}(\transactions_\hist', \prec_{\hist})$
be a set of all sequences of transactions s.t. $\sigma \in
\mathit{ser}(\transactions_\hist', \prec_{\hist})$ if 
$\sigma$ contains every element of $\transactions_\hist'$ exactly once and
for any $\tr_i, \tr_j \in \transactions_\hist'$, if $\tr_i \prec_{\hist} \tr_j$ then
$\tr_i$ precedes $\tr_j$ in $\sigma$.

Given a history H and some response event r in H, 
let $\mathit{ops}(\sigma, r)$ be a sequence of operations s.t. if $\sigma = [
\tr_1, \tr_2, ..., \tr_n]$ then $\mathit{ops}(\sigma, r) = \hist\upto r|\tr_1
\cdot \hist\upto r|\tr_2 \cdot ... \cdot \hist\upto r|T_n$. 
This represents the sequential history prior to response event $r$ that
respects a specific order of transactions defined by $\sigma$.

The most relevant condition in TMS1 with respect to early release checks the
validity of individual response operations. A prequisite for checking validity
is to check whether a response event can be justified by some \emph{externally
consistent prefix}. This prefix consists of operations from all transactions
that precede the response event and whose effects are visible to other
transactions.
Specifically, if a transaction precedes another transaction in the real time order, then it must be
both committed and included in the prefix, or both not committed and excluded
from the prefix.
However, if a transaction does not precede another transaction, it can be in
the prefix regardless of whether it committed or aborted.

\begin{definition} [Extended Consistent Prefix] \label{def:extcp}
    Given a history $\hist$ and a response event $r$,
    let the set of transactions $\transactions^r_\hist$ be any subset of all transactions in $\hist$ s.t.
    for any $\tr_i, \tr_j \in \transactions^r_\hist$, 
    if $\tr_i \prec_\hist \tr_j$ 
    then $\tr_i$ is in $\transactions^r_\hist$ 
         iff $\res{i}{}{\co_i} \in \hist\upto r|\tr_i$.
\end{definition}

TMS1 specifies that each response to an operation invocation in a safe history must be
\emph{valid}. Intuitively, a valid response event is one for which there exists
a sequential prefix that is both legal and meets the conditions of an
externally consistent prefix. More precisely, the following condition must be met.

\begin{definition} [Valid Response] \label{def:valid-response}
    Given a transaction $\tr_i$ in $\hist$,
    we say the response $r$ to some operation invocation $e$ is valid 
    if there exists set $\transactions^r_\hist \subseteq \transactions_d\upto r$
        and sequence $\sigma \in \mathit{ser}(\transactions^r_\hist, \prec_{\hist})$
        s.t. $\transactions^r_\hist$ satisfies \rdef{def:extcp} and 
        $\mathit{ops}(\sigma \cdot \tr_i, r) \cdot [e \rightarrow r]$ is legal.
\end{definition}

\begin{theorem} \label{thm:tms1-early-release}
    TMS1 does not support early release. 
\end{theorem}

\begin{proof} %
    Assume by contradiction that TMS1 supports early release.
    Then by \rdef{def:release-support},
    there exists some TMS1 history $\hist$ s.t. $\tr_i,\tr_j \in \hist$   
    and there is a prefix $P$ of $\hist$ s.t. 
    $\op_i = \fwop{i}{\obj}{\val}{\ok_i} \in P|\tr_i$,
    $\op_j = \frop{j}{\obj}{\val} \in P|\tr_j$, and $\tr_i$ is live in $\hist$.
    This implies that 
    $\inv{i}{}{\tryC_i} \not\in P\upto \res{j}{}{\val}|\tr_i$.
    This means that $\tr_i \not\in \transactions_d$ and therefore not in any
    $\transactions' \subseteq \transactions_d$ or, by extension, any $\sigma
    \in \mathit{ser}(\transactions', \prec_{\hist})$. Therefore, there is
    no $\op_i$ in $\mathit{ops}(\sigma,  \res{j}{}{\val})$, so, assuming unique writes, $\op_j$
    is not preceded by a write of $\val$ to $\obj$ in $\mathit{ops}(\sigma
    \cdot \tr_j,  \res{j}{}{\val}) \cdot [
    \frop{j}{\obj}{\val}
    ]$. Therefore,  $\mathit{ops}(\sigma
    \cdot \tr_j,  \res{j}{}{\val}) \cdot [
    \frop{j}{\obj}{\val}
    ]$ is not legal, which contradicts
    \rdef{def:valid-response}.
\end{proof}

Since both \rdef{def:overwriting-support} and \rdef{def:aborting-support}
require early release support, then:

\begin{corollary} \label{thm:tms1-overwriting}
    TMS1 does not support overwriting.
\end{corollary}

\begin{corollary} \label{thm:tms1-aborting}
    TMS1 does not support aborting early release.
\end{corollary}

TMS2 is a stricter, but more intuitive version of TMS1.
Since the authors show in \cite{DGLM13} that TMS2 is strictly stronger than
TMS1 (TMS2 implements TMS1), the conclusions above equally apply to TMS2.
Hence, from \rthm{thm:tms1-early-release}:

\begin{corollary}
    TMS2 does not support early release. 
\end{corollary}

\begin{corollary} \label{thm:tms2-overwriting}
    TMS2 does not support overwriting.
\end{corollary}

\begin{corollary} \label{thm:tms2-aborting}
    TMS2 does not support aborting early release.
\end{corollary}

\subsection{Virtual World Consistency}

The requirements of opacity, while very important in the context of TM's
ability to execute any operation transactionally, can \mbox{often} be
excessively stringent. On the other hand serializability is considered too weak
for many TM applications. Thus, a weaker TM consistency condition called \emph{virtual world consistency} (\emph{VWC})
was introduced in \cite{IMR08}.
The definition of VWC depends on \emph{causal past}. The causal past $C(\hist,
\tr_i)$ of some transaction $\tr_i$ in some history $\hist$ is the set that
contains $\tr_i$ and all aborted or committed transactions that precede $\tr_i$
in $\hist$.  
A causal past $C(\hist, \tr_i)$ is legal, if for every $\tr_j \in
C(\hist, \tr_i)$, s.t. $i \neq j$, $\tr_j$ is committed in $\hist$.

\begin{definition} [Virtual World Consistency] \label{def:vwc-short}
    History $H$ is virtual world consistent iff all committed transactions are
    serializable and preserve real-time order, and for each aborted transaction
    there exists a linear extension of its causal past that is legal.
\end{definition}

This property allows a limited support for early release as follows.

\begin{theorem} \label{thm:vwc-early-release}
    VWC supports early release. 
\end{theorem}

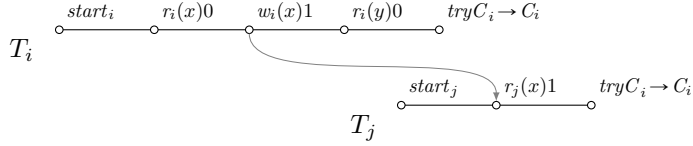
\begin{figure}
\begin{center}
\begin{tikzpicture}
     \draw
           (0,2)        node[tid]       {$\tr_i$}
                        node[aop]       {$\init_i$} %
                        node[dot]       {} 

      -- ++(1.25,0)     node[aop]       {$\trop{i}{\obj}{0}$}
                        node[dot] (wi0) {}

      -- ++(1.25,0)     node[aop]       {$\twop{i}{\obj}{1}$}
                        node[dot] (wi1) {}

      -- ++(1.25,0)     node[aop]       {$\trop{i}{\objy}{0}$}
                        node[dot]       {}

      -- ++(1.25,0)     node[aop]       {$\tryC_i\!\to\!\co_i$}
                        node[dot]       {}         
                        ;

     \draw
           (4.5,1)      node[tid]       {$\tr_j$}
                        node[aop]       {$\init_j$} %
                        node[dot]       {} 

      -- ++(1.25,0)     node[aop]       {$\trop{j}{\obj}{1}$}
                        node[dot] (rj1) {}

      -- ++(1.25,0)     node[aop]       {$\tryC_i\!\to\!\co_i$}
                        node[dot]       {} 
                        ;

     \draw[hb] (wi1) \squiggle (rj1);
\end{tikzpicture}
\end{center}
\caption{\label{fig:vwc-history} VWC history with early
    release.}
\end{figure}

\begin{proof} %
    Let $\hist$ be a transactional history as shown in \rfig{fig:vwc-history}.
    Here, $\tr_i$ performs two operations on $\obj$ and one on $\objy$, while
    $\tr_j$ reads $\obj$.  
    The linear extension of $\hist$ is
    $S = \hist|\tr_i \cdot \hist|\tr_j$ 
    wherein both transactions are trivially legal. Thus $\hist$ is VWC.
    Since, by \rdef{def:early-release}, $\tr_i$ releases early in $\hist$, then,
    by \rdef{def:release-support}, VWC supports early release.
\end{proof}

\begin{theorem} \label{thm:vwc-lu-release}
    VWC does not support overwriting.
\end{theorem}

\begin{proof} %
    Since VWC requires that aborting transactions view a legal causal past,
    then if a transaction reading $\obj$ is aborted, it must read a legal
    (i.e. "final") value of $\obj$. 
    Thus, let us consider some history $\hist$ where some $\tr_i$ releases
    $\obj$ early, and some $\tr_j$ reads $\obj$ from $\tr_i$. 
    \begin{compactenum}[a)]
        \item If $\tr_i$ writes to $\obj$ after releasing it, and $\tr_j$
        commits, then $\tr_j$ is not legal, and therefore $\hist$ does not
        satisfy VWC.
        \item If $\tr_i$ writes to $\obj$ after releasing it, and $\tr_j$
        aborts, then the causal past of $\tr_j$ contains $\tr_i$, and $\tr_j$
        reads an illegal (stale) value of $\obj$ from $\tr_i$, so $\hist$ does
        not satisfy VWC.
    \end{compactenum}
    Therefore, any history $\hist$ containing $\tr_i$, such that $\tr_i$
    releases $\obj$ early and modifies it after release does not satisfy VWC.
    Hence, by \rdef{def:overwriting-support}, VWC does not support overwriting.
\end{proof}

While VWC supports early release, there are severe
limitations to this capability. That is, VWC does not allow a transaction that
released early to subsequently abort for any reason.

\begin{theorem} \label{thm:vwc-no-abort}
    VWC does not support aborting early release
\end{theorem}

\begin{proof}
    Given a history $\hist$ that satisfies VWC and a transaction $\tr_i \in
    \hist$ that releases variable $\obj$ in $\hist$, let us assume for the sake
    of contradiction that
    $\tr_i$ eventually aborts.  By
    \rdef{def:early-release}, there is some $\tr_j$ in $\hist$ that reads from
    $\tr_i$. If $\tr_i$ eventually aborts, then $\tr_j$ reads from an aborted
    transaction.
    \begin{compactenum}[a)]
        \item If $\tr_j$ eventually aborts, then its causal past contains two
        aborted transactions ($\tr_i$ and $\tr_j$) and is, therefore, illegal.
        Hence $\hist$ does not satisfy VWC, which is a contradiction.
        \item If $\tr_j$ eventually commits, then the sequential witness
        history is also illegal. Hence $\hist$ does not satisfy VWC, which is a
        contradiction.
    \end{compactenum}
    Therefore, if $\tr_i$ eventually aborts, $\hist$ does not satisfy VWC,
    which is a contradiction. Thus, since a VWC history cannot contain an
    abortable transaction that releases a variable early.
    Hence, by \rdef{def:aborting-support}, VWC does not support aborting early
    release. 
\end{proof}

VWC does not allow for transactions that release early to abort, which we
consider to be an impractical assumption in some TM systems and an overstrict
requirement in general.

\subsection{Live Opacity}
\label{sec:live-opacity}

\emph{Live opacity} was introduced in \cite{DFK14} as part of a set of
consistency conditions and safety properties that were meant to regulate the
ability of transactions to read from live transactions. The work analyzes a
number of properties and for each one presents a commit oriented variant that
forbids early release and a live variant that allows it. Here, we concentrate
on live opacity, since it best fits alongside the other properties presented
here, however our conclusions will apply to the remainder of live properties.

Let $\hist|(\tr_i, r)$ be the longest subsequence of $\hist|\tr_i$ containing
only read operation executions (possibly pending), with the exclusion of the
last read operation if its response event is $\ab_i$. Let $\hist|(\tr_i, gr)$
be a subsequence of $\hist|(\tr_i, r)$ that contains only non-local operation
executions.
Let $\tr_i^r$ be a transaction that invokes the same transactional operations
as those invoked in $\hist|(\tr_i, r) \cdot [\inv{i}{}{\tryC_i}]$ if
$\hist|(\tr_i, r) \neq \varnothing$, or $\varnothing$ otherwise.
Let $\tr_i^{\mathit{gr}}$ be a transaction that invokes the same transactional
operations as those invoked in $[\init_i\to\ok_i] \cdot \hist|(\tr_i,
\mathit{gr}) \cdot [\tryC_i \to \co_i]$ if $\hist|(\tr_i, \mathit{gr}) \neq
\varnothing$, or $\varnothing$ otherwise.

Given a history $\hist$, a transaction $\tr_i \in \hist$, and a complete local
operation execution $\op = \frop{i}{\obj}{\val}$, we say the latter's response
event $\res{i}{}{\val}$ is \emph{legal} if the last preceding write operation
in $\hist|\tr_i$ writes $\val$ to $\obj$. 
We say sequential history $S$ justifies the serializability of history $\hist$
when there exists a history $\hist'$ that is a subsequence of $\hist$ s.t.
$\hist$' contains invocation and response events issued and received by
transactions committed in $\hist$, and $S$ is a legal history equivalent to
$\hist'$. 

\begin{definition}[Live Opacity]\label{def:live-opacity}
    A history $\hist$ is live opaque iff, there
    exists a sequential history $S$ that preserves the real time order of
    $\hist$ and justifies that $\hist$ is serializable
    and all of the following hold:

    \begin{enumerate}[a) ]
        \item We can extend history $S$ to get a sequential history $S'$ such that:
            \begin{itemize}[--]
                \item for each transaction $\tr_i \in \hist$ s.t. $\tr_i
                \not\in S$, $\tr_i^{\mathit{gr}} \in S'$,
                \item if $<$ is a partial order induced by the real time order
                of $S$ in such a way that for each transaction $\tr_i \in
                \hist$ s.t. $\tr_i \not\in S$ we replace each instance of
                $\tr_i$ in the set of pairs of the real time order of $\hist$
                with transaction $\tr_i^{\mathit{gr}}$, then $S'$ respects $<$,
                \item S' is legal.
            \end{itemize}
        \item For each transaction $\tr_i \in \hist$ s.t. $\tr_i \not\in S$ and
        for each operation $\op$ in $\tr_i^{r}$ that is not in $\tr_i^{\mathit{gr}}$,
        the response for $\op$ is legal.
    \end{enumerate}
\end{definition}

\begin{figure}
\begin{center}
\begin{tikzpicture}
     \draw
           (0,2)        node[tid]       {$\tr_i$}
                        node[aop]       {$\init_i$} %
                        node[dot]       {} 

      -- ++(1.25,0)     node[aop]       {$\twop{i}{\obj}{1}$}
                        node[dot] (wi1) {}

      -- ++(1.25,0)     node[aop]       {$\tryC_i\!\to\!\co_i$}
                        node[dot]       {}         
                        ;

     \draw
           (1.00,1)     node[tid]       {$\tr_j$}
                        node[aop]       {$\init_j$} %
                        node[dot]       {} 

      -- ++(1.25,0)     node[aop]       {$\trop{j}{\obj}{1}$}
                        node[dot] (rj1) {}

      -- ++(1.25,0)     node[aop]       {$\tryC_i\!\to\!\co_i$}
                        node[dot]       {} 
                        ;

     \draw[hb] (wi1) \squiggle (rj1);
\end{tikzpicture}
\end{center}
 \caption{\label{fig:live-opaque-history} Live opaque history with early
          release.}
\end{figure}
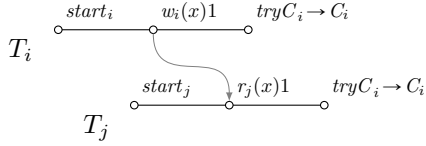

\begin{theorem} \label{thm:live-opacity-early-release}
    Live opacity supports early release. 
\end{theorem}

\begin{proof} %
    Let history $\hist$ be that represented in \rfig{fig:live-opaque-history}.
    Since there is a transaction $\tr_i \in \hist$ that writes $1$ to $\obj$
    and a transaction $\tr_j$ that reads $1$ from $\obj$ before $\tr_i$
    commits, then there is a prefix $P$ of $\hist$ that meets
    \rdef{def:early-release}. Therefore $\tr_i$ releases $\obj$ early in
    $\hist$.

    Let $S$ be a sequential history s.t. $S = \hist|\tr_i \cdot \hist|\tr_j$.
    Since the real-time order of $H$ is $\varnothing$, then, trivially, $S$
    preserves the real-time order of $\hist$.
    Since $\vis{S}{\tr_i}$ contains only $\hist|\tr_i$ and therefore only a
    single write operation execution and no reads, then it is legal and $\tr_i$
    in $S$ is legal in $S$.
    Furthermore, $\vis{S}{\tr_j}$ is such that $\vis{S}{\tr_j} = \hist|\tr_i
    \cdot \hist|\tr_j$ and contains a read operation $\frop{j}{\obj}{1}$
    preceded by the only write operation
    $\fwop{i}{\obj}{1}{\ok_i}$, so
    $\vis{S}{\tr_j}$ is legal, and, consequently, $\tr_j$ in $S$ is legal in
    $S$.
    Thus, all transactions in $S$ are legal in $S$, so $\hist$ is serializable.
    
    Let $S'$ be a sequential history that extends $S$ in accordance to
    \rdef{def:live-opacity}. Since there are no transactions in $S'$ that are
    not in $S$, then $S'=S$. Thus, since every transaction in $S$ is legal in
    $S$, then every transaction in $S'$ is legal in $S'$. Trivially, $S'$ also
    preserves the real time order of $S$. Therefore, the condition
    \rdef{def:live-opacity}a is met.
    Since there are no local read operations in $S$, then condition
    \rdef{def:live-opacity}b is trivially met as well.
    Therefore, $\hist$ is live opaque.

    Since $\hist$ is both live opaque and contains a transaction that releases
    early, then the theorem holds.
\end{proof}

\begin{theorem} \label{thm:live-opacity-overwriting}
    Live opacity does not support overwriting.
\end{theorem}

\begin{proof} %
    For the sake of contradiction, let us assume that there is a history (with
    unique writes) $\hist$ that is live opaque and, from
    \rdef{def:overwriting-support}, contains some transaction $\tr_i$ that
    writes value $\val$ to some variable $\obj$ and releases $\obj$ early and
    subsequently executes another write operation writing $\val'$ to $\obj$
    where the second write follows a read operation executed by transaction
    $\tr_j$ reading $\val$ from $\obj$.

    Since $\hist$ is live opaque there exists a sequential history $S$ that
    justifies the serializability of $\hist$.
    There cannot exist a sequential history $S$ where $\tr_j$ reads from $\obj$
    between two writes to $\obj$ executed by $\tr_i$, because there cannot
    exist a legal $\vis{S}{\tr_j}$, so $\tr_j$ would not
    be legal in $S$. Therefore, $\tr_j$ must be aborted in $\hist$ and
    therefore $\tr_j$ is not in any sequential history $S$ that justifies the
    serializability of $\hist$.
    
    Since $\tr_j$ is in $\hist$ but not in $S$, then given any sequential
    extension $S'$ of $S$ in accordance to \rdef{def:live-opacity} $\tr_j$ is
    replaced in $S'$ by $\tr_j^{gr}$ which reads $\val$ from $\obj$ and finally
    commits.
    However, since the write operation execution writing $\val$ to $\obj$ in 
    $\tr_i$ is followed in $S'|\tr_i$ by another write operation execution that writes
    $\val'$ to $\obj$, then there cannot exist a $\vis{S'}{\tr_j^{\mathit{gr}}}$ that is
    legal. Thus $\tr_j^{\mathit{gr}}$ in $S'$ cannot be legal in $S'$, which contradicts
    \rdef{def:live-opacity}a. Thus, $\hist$ is not live opaque, which is a
    contradiction.

    Therefore $\hist$ cannot simultaneously be live opaque and contain a
    transaction with early release and overwriting.
\end{proof}

\begin{theorem} \label{thm:live-opacity-aborting}
    Live opacity does not support aborting early release.
\end{theorem}

\begin{proof} %
    For the sake of contradiction, let us assume that there is a history (with
    unique writes) $\hist$ that is live opaque and, from
    \rdef{def:aborting-support}, contains some transaction $\tr_i$ that
    writes value $\val$ to some variable $\obj$ and releases $\obj$  and
    subsequently aborts in $\hist$.

    Let $S$ be any sequential history that justifies the serializability of
    $\hist$, and let $S'$ be any sequential extension $S'$ of $S$ in accordance
    to \rdef{def:live-opacity}.
    Since $\tr_i$ aborts in $\hist$, then it is not in $S$, and therefore it is
    replaced in $S'$ by $\tr_i^{\mathit{gr}}$.
    Since, by construction, $\tr_i^{\mathit{gr}}$ does not contain any write operation executions, there
    is no write operation execution  writing $\val$ to $\obj$ in $S'$.
    Since $\tr_i$ released $\obj$ early in $\hist$ there is a transaction
    $\tr_j$ in $\hist$ that executes a read operation reading $\val$ from
    $\obj$ and the same read operation is in $S'$. 
    But since there  is no write operation execution writing $\val$ to $\obj$
    in $S'$, no transaction containing a read operation execution reading
    $\val$ from $\obj$ can be legal in $S'$.
    Thus, $\hist$ is not live opaque, which is a contradiction.

    Therefore $\hist$ cannot be simultaneously live opaque and contain a
    transaction with early release that aborts.
\end{proof}

Like with VWC, live opacity does not allow transactions that release early to
abort, which we consider too strict a condition.

\subsection{Elastic Opacity}
\label{sec:elastic-opacity}

\emph{Elastic opacity} is a safety property based on opacity, that was introduced to
describe the safety guarantees of elastic transactions \cite{FGG09}. The
property allows to relax the atomicity %
 requirement of transactions to allow
each of them to execute as a series of smaller transactions. 
An \emph{elastic transaction} $\tr_i$ is split into a sequence of subhistories called a
\emph{cut} denoted $\cut{\hist}{i}$, where each subhistory represents a "subtransaction." In brief, a cut that contains more than one operation execution is
\emph{well-formed} if all subhistories are longer than one operation execution,
all the write operations are in the same subhistory, and the first operation execution on any variable in every subhistory 
is not a write operation, except possibly in the first subhistory. 
A well-formed cut of some transaction $\tr_i$ is consistent in some history
$\hist$, if given any two operation executions $\op_i$ and $\op_i'$ on $\obj$
in any subhistories of the cut, no transaction $\tr_j$ ($i\neq j$) executes a
write operation $\op_j$ on $\obj$ s.t.  $\op_i \prec_\hist \op_j \prec_\hist
\op_i'$.  In addition,  given any two operation executions $\op_i$ and $\op_i'$
on $\obj, \objy$ respectively, no two transactions $\tr_k, \tr_l$ ($l\neq i$,
$k\neq i$) execute writes $\op_k$ on $\obj$ and $\op_l$ on $\objy$, s.t.  $\op_i
\prec_\hist \op_k \prec_\hist \op_i'$ and $\op_i \prec_\hist \op_l \prec_\hist
\op_i'$.
A \emph{cutting function} $f_\cutf$ takes a history $\hist$ as an argument and
produces a new history $\hist_f$ where for each transaction $\tr_i \in \hist$
declared as elastic,
$\tr_i$ is replaced in $\hist_f$ with the transactions resulting from the cut $\cut{\hist}{i}$
of $\tr_i$. 
If some transaction is committed (aborted) in $\hist$, then all transactions
resulting from its cut are committed (aborted) in $f_\cutf(\hist)$.
Then, elastic opacity is defined as follows:

\begin{definition} [Elastic Opacity] \label{def:elastic-opacity}
    History $\hist$ is elastic opaque iff there exists a cutting function
    $f_\cutf$ that replaces each elastic transaction $\tr_i$ in $\hist$
    with its consistent cut $\cut{\hist}{i}$, such that history $f_\cutf(\hist)$ is opaque.
\end{definition}

\begin{figure}
\begin{center}
\begin{tikzpicture}
     \draw
           (0,2)        node[tid]       {$\tr_i$}
                        node[aop]       {$\init_i$} %
                        node[dot]       {} 

      -- ++(1.25,0)     node[aop]       {$\trop{i}{\objy}{0}$}
                        node[dot]       {}

      -- ++(1.25,0)     node[aop]       {$\twop{i}{\obj}{1}$}
                        node[dot] (wi1) {}

      -- ++(5.50,0)     node[aop]       {$\trop{i}{\obj}{1}$}
                        node[dot] (ri1) {}

      -- ++(1.25,0)     node[aop]       {$\trop{i}{\objy}{0}$}
                        node[dot]       {}

      -- ++(1.25,0)     node[aop]       {$\tryC_i\!\to\!\co_i$}
                        node[dot]       {}         
                        ;

     \draw
           (4.00,1)     node[tid]       {$\tr_j$}
                        node[aop]       {$\init_j$} %
                        node[dot]       {} 

      -- ++(1.25,0)     node[aop]       {$\trop{j}{\obj}{1}$}
                        node[dot] (rj1) {}

      -- ++(1.25,0)     node[aop]       {$\tryC_j\!\to\!\co_j$}
                        node[dot]       {} 
                        ;

     \draw[hb] (wi1) \squiggle (rj1);
     \draw[hb] (rj1) \squiggleup (ri1);
\end{tikzpicture}
\end{center}
 \caption{\label{fig:elastic-opaque-history} Elastic opaque history with early
          release.}
\end{figure}
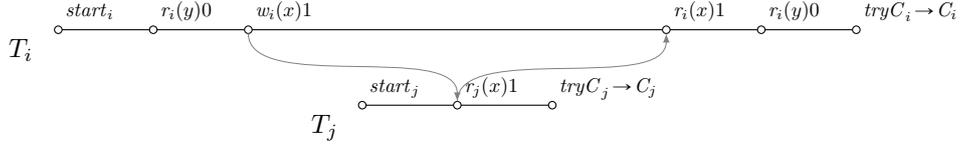
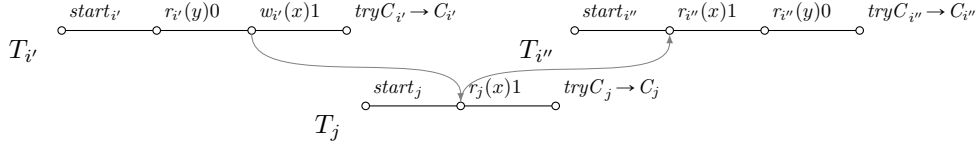
\begin{figure}
\begin{center}
\begin{tikzpicture}
     \draw
           (0,2)        node[tid]       {$\tr_{i'}$}
                        node[aop]       {$\init_{i'}$} %
                        node[dot]       {} 

      -- ++(1.25,0)     node[aop]       {$\trop{i'}{\objy}{0}$}
                        node[dot]       {}

      -- ++(1.25,0)     node[aop]       {$\twop{i'}{\obj}{1}$}
                        node[dot] (wi1) {}
                        
      -- ++(1.25,0)     node[aop]       {$\tryC_{i'}\!\to\!\co_{i'}$}
                        node[dot]       {} 
                        ;

     \draw
           (6.75,2)        node[tid]       {$\tr_{i''}$}
                        node[aop]       {$\init_{i''}$} %
                        node[dot]       {}
                         
      -- ++(1.25,0)     node[aop]       {$\trop{i''}{\obj}{1}$}
                        node[dot] (ri1) {}

      -- ++(1.25,0)     node[aop]       {$\trop{i''}{\objy}{0}$}
                        node[dot]       {}

      -- ++(1.25,0)     node[aop]       {$\tryC_{i''}\!\to\!\co_{i''}$}
                        node[dot]       {}         
                        ;

     \draw
           (4.0,1)      node[tid]       {$\tr_j$}
                        node[aop]       {$\init_j$} %
                        node[dot]       {} 

      -- ++(1.25,0)     node[aop]       {$\trop{j}{\obj}{1}$}
                        node[dot] (rj1) {}

      -- ++(1.25,0)     node[aop]       {$\tryC_j\!\to\!\co_j$}
                        node[dot]       {} 
                        ;

     \draw[hb] (wi1) \squiggle (rj1);
     \draw[hb] (rj1) \squiggleup (ri1);
\end{tikzpicture}
\end{center}
\caption{\label{fig:elastic-opaque-cut} History after applying a cutting function.}
\end{figure}

\begin{theorem} \label{thm:elastic-opacity-early-release}
    Elastic opacity supports early release.
\end{theorem}

\begin{proof} %
    Let $\hist$ be a transactional history with unique writes as shown in 
    \rfig{fig:elastic-opaque-history}. Let $\tr_i$ be an elastic transaction. 
    Let $\cut{\hist}{i}$ be a cut of subhistory $\hist|\tr_i$, such that: 
    \begin{equation*}
    \begin{split}
    \cut{\hist}{i} =  \{ 
     & [\init_{i'}\to\ok_{i'},\frop{i'}{\objy}{0},
        \fwop{i'}{\objx}{1}{\ok_{i'}},\tryC_{i'}\to\co_{i'}], \\
     & [\init_{i''}\to\ok_{i''},\frop{i''}{\objx}{1},
        \frop{i''}{\objy}{0},\tryC_{i''}\to\co_{i''}]\}.
    \end{split}
    \end{equation*}
    All subhistories of $\cut{\hist}{i}$ are longer than one operation, all the writes are
    in the first subhistory, and no subhistory starts with a write, so $\cut{\hist}{i}$ is 
    well-formed.
    Since there are no write operations outside of $\tr_i$, then it follows
    that $\cut{\hist}{i}$ is a consistent cut in $\hist$. Let $f_\cutf$ be any cutting
    function such that it cuts $\tr_i$ according to $\cut{\hist}{i}$, in which case
    $f_\cutf(\hist)$ is defined as in \rfig{fig:elastic-opaque-cut}. 
    Let $S$ be a sequential history s.t. $S = f_\cutf(\hist)|\tr_{i'} \cdot
    f_\cutf(\hist)|\tr_{j} \cdot f_\cutf(\hist)|\tr_{i''}$.
    Since $\tr_{i'}$ precedes $\tr_{i''}$ in $S$ as well as in $f_\cutf(\hist)$,
    and all other transactions are not real time ordered, $S$ preserves the
    real time order of $f_\cutf(\hist)$.
    Trivially, each transaction in $S$ is legal in $S$.
    Thus, $f_\cutf(\hist)$ is opaque by \rdef{def:opacity}, and in effect $\hist$
    is elastic opaque by \rdef{def:elastic-opacity}.
    Since in $\hist$ transaction $\tr_j$ reads $\obj$ from $\tr_i$ while
    $\tr_i$ is live, then, by \rdef{def:early-release}, $\tr_i$ releases $\obj$
    early in $\obj$. Hence, since $\hist$ is elastic opaque, elastic opacity
    supports early release, by \rdef{def:release-support}. 
\end{proof}

\begin{theorem} \label{thm:elastic-opacity-overwriting}
    Elastic opacity does not support overwriting.
\end{theorem}

\begin{proof} %
    For the sake of contradiction, let us assume that there is an elastic
    opaque history $\hist$ s.t. transaction $\tr_i$ writes value $\val$ to  some variable
    $\obj$ and releases it early in $\hist$. Furthermore, let us assume that there is overwriting, so after some
    transaction  $\tr_j$ reads $\val$ from $\obj$, $\tr_i$ writes $\valu$ to
    $\obj$.
    Since only elastic transactions can release early in elastic opaque
    histories, and $\tr_i$ releases early, $\tr_i$ is necessarily elastic.
    Thus, in any $f_\cutf(\hist)$ $\tr_i$ is replaced by a cut $\cut{\hist}{i}$. 
    
    The two writes on $\obj$ in $\tr_i$ are either
    \begin{inparaenum}[a) ]
        \item in two different subhistories in $\cut{i}{\hist}$, or
        \item in the same subhistory in $\cut{i}{\hist}$.
    \end{inparaenum}
    Since the definition of a consistent cut requires all writes on a single
    variable are within one subhistory of the cut, then in case (a),
    $\cut{i}{\hist}$ is inconsistent. Since by \rdef{def:elastic-opacity}
    elastic opaque histories are created using consistent cuts, then $\hist$ is
    not elastic opaque, which is a contradiction.

    In the case of (b), let us say that both writes are in a subhistory that is
    converted into transaction $\tr_i'$ in $f_\cutf(\hist)$. 
    Since $\tr_i$ releases $\obj$ early, then by \rdef{def:early-release},
    there is a transaction $\tr_j'$ in $f_\cutf(\hist)$ which executes a read
    on $\obj$ reading the value written by $\tr_i'$ in $f_\cutf(\hist)$. 
    Since we assume overwriting, the read operation on $\obj$ in $\tr_j'$ reads
    the value written by the first of the two writes in $\tr_i'$ and does
    so before the other write on $\obj$ is performed within $\cut{i}{\hist}$.
    Then, in any sequential history $S$ equivalent to $f_\cutf(\hist)$ either
    $\tr_j' \prec_S \tr_i'$ or $\tr_i' \prec_S \tr_j'$. In the former case
    $\tr_j'$ in $S$ is not legal in $S$, since the read on $\obj$ that yields
    value $\val$ will not be preceded by any operation that writes $\val$ to
    $\obj$ in any possible $\vis{S,\tr_j'}$.
    In the latter case $\tr_j'$ in $S$ is also not legal in $S$, since there
    will be a write operation writing $\valu$ to $\obj$ between the read on
    $\obj$ that yields value $\val$ and any operation that writes $\val$ to
    $\obj$ in $\vis{S,\tr_j'}$.
    Since $\tr_j'$ in $S$ is not legal in any $S$ equivalent to
    $f_\cutf(\hist)$, then, by \rdef{def:final-state-opacity}, $f_\cutf(\hist)$
    is not final-state opaque, and hence, by \rdef{def:opacity}, not opaque. In
    effect, by \rdef{def:elastic-opacity}, $\hist$ is not opaque, which is a
    contradiction.

    Thus, there cannot be an elastic opaque history $\hist$ with overwriting.
\end{proof}

\begin{theorem} \label{thm:elastic-opacity-overwriting}
    Elastic opacity does not support early release aborting.
\end{theorem}

\begin{proof} %
    For the sake of contradiction, let us assume that there is an elastic
    opaque history $\hist$ s.t. transaction $\tr_i$ releases some variable
    $\obj$ early in $\hist$ and aborts. 
    Since $\tr_i$ releases early then it
    writes $\val$ to $\obj$, and there is another $\tr_j$ that executes a read
    on $\obj$ that returns $\val$ before $\tr_i$ aborts.
    Since only elastic transactions can release early in elastic opaque
    histories, and $\tr_i$ releases early, $\tr_i$ is necessarily elastic.
    If $\tr_i$ aborts in $\hist$, then all of the transactions resulting from
    its cut $\cut{\hist}{i}$ in $f_\cutf(\hist)$ also abort (by construction of
    $f_\cutf(\hist)$). 
    Therefore, for any sequential history $S$ equivalent to $f_\cutf(\hist)$, there
    is no subhistory $\hist' \in \cut{\hist}{i}$ s.t. $\hist' \subseteq \vis{S}{\tr_j}$,
    and in effect the read operation in $\tr_j$ on $\obj$ reading $\val$ is not
    preceded by a write operation writing $\val$ to $\obj$. Therefore,
    $\vis{S}{\tr_j}$ is illegal, so $\tr_j$ in $S$ is not legal in $S$, and
    thus, by \rdef{def:opacity} $f_\cutf(\hist)$ is not opaque.  
    Since $f_\cutf(\hist)$ is not opaque, then by \rdef{def:elastic-opacity},
    $\hist$ is not elastic opaque, which is a contradiction.
\end{proof}

Elastic opacity supports early release, but, since it does not guarantee
serializability (as shown in \cite{FGG09}), we consider it to be a relatively
weak property. This is contrary to our premise of finding a property that
allows early release and provides stronger guarantees than serializability.
Elastic transactions were proposed as an alternative to traditional transactions for implementing search structures, but we submit
that the restrictions placed on the composition of
elastic transactions 
and the need for transactions with early release to
be non-aborting 
put an unnecessary burden on general-purpose TM. In particular, for a cut to be
well-formed, it is necessary that all writes are executed in the same
subtransaction, and that no subtransaction starts with a write, which severely
limits how early release can be used and precludes scenarios that are
nevertheless intuitively correct.
In addition, elastic opacity requires that transactions which release early do
not subsequently abort.

\subsection{Database Properties}
\label{sec:properties-database}

We follow the discussion of TM safety properties with a brief foray into
database properties that deal with transaction consistency. Given that TM
properties tend not to be very helpful when describing the behavior of early
release, these consistency properties may be used to supplement that.

\emph{Recoverability} is a database property defined as below (following
\cite{Had88}):
\begin{definition}[Recoverability]\label{def:recoverability-short}
    History $\hist$ is \emph{recoverable} iff for any $\tr_i, \tr_j \in H$,
    s.t. $i\neq j$ and  $\tr_j$ reads from $\tr_i$, $\tr_i$ commits in $H$ before
    $\tr_j$ commits.
\end{definition}

Recoverability does not make requirements about values read by transactions, so
it necessarily supports early release, overwriting, and aborting early release.
It also allows histories that are not even serializable.
As such, it is too weak for application in TM.
Recoverability can be combined with serializability to restrict the order on
commits and aborts in serializable histories. The resulting consistency
condition is therefore stronger than serializability. However, it still allows
unrestricted early release, overwriting, and aborting early release, and thus is not
suitable for TM.

\emph{Avoiding cascading aborts} (ACA) \cite{BHG87} is a database property
defined as:
\begin{definition}[Avoiding Cascading Aborts]\label{def:aca-short}
    History $\hist$ Avoids Cascading Aborts iff for any $\tr_i, \tr_j \in
    \hist$ s.t. $i\neq j$ and $\tr_j$ reads from $\tr_i$, $\tr_i$ commits before the read.
\end{definition}

As with recoverability, ACA restricts reading from live transactions.
Therefore, ACA clearly removes all the scenarios encompassed by
\rdef{def:early-release}.
Since this is the only provision of ACA, the property forbids early release,
without giving any additional guarantees. 
Hence, it also does not support overwriting nor aborting early release.

\emph{Strictness} \cite{BHG87} is a database property defined as:

\begin{definition}[Strictness] \label{def:strict-short}
    History $\hist$ is \emph{strict} iff for any $\tr_i, \tr_j \in H$ ($i\neq j$) 
    and given any
    operation execution $\op_i = \pfrop{i}{}{\obj}{v}$ or $\pfwop{i}{}{\obj}{v'}{\ok_i}$ in $\hist|\tr_i$, and
    any operation execution $\op_j = \pfwop{j}{}{\obj}{v}{\ok_j}$ in $\hist|\tr_j$, if $\op_i$
    follows $\op_j$, then $\tr_j$ commits or aborts before $\op_i$.
\end{definition}

The definition unequivocally states that a transaction cannot read from 
another transaction, until the latter is committed or aborted.
Thus, strictness precludes early release altogether. 
Hence, it also does not support overwriting nor aborting early release.

\emph{Rigorousness} is defined (following \cite{BGRS91}) as:

\begin{definition}[Rigorousness] \label{def:rigorous-short}
    History $\hist$ is \emph{rigorous} iff it is strict and for any $\tr_i, \tr_j \in
    \hist$ ($i\neq j$) such that $\tr_i$ writes to variable $\obj$, i.e., 
    $\op_i = \fwop{i}{\obj}{v}{\ok_i} \in \hist|\tr_i$ 
    after $\tr_j$ reads $\obj$, then $\tr_j$
    commits or aborts before $\op_i$.
\end{definition}

Since in \cite{AH14} the authors demonstrate that rigorous histories are
opaque, and since we show in \rthm{thm:opacity-early-release} that opaque
histories do not support early release, then neither does rigorousness. 
Hence, it also does not support overwriting nor aborting early release.

\subsection{Discussion}
\label{sec:properties-discussion}

The survey of properties shows that, while there are many safety properties for
TM with a wide range of guarantees they provide, with respect to early
release they fall into three basic groups.

The first group consists of properties that allow early release but do
not prevent overwriting: serializability and recoverability.
These properties do not control what can be seen by aborting transactions. As
argued in \cite{GK10}, this is insufficient for TM in general, because operating
on inconsistent state may lead to uncontrollable errors, whose consequences include crashing the
process.

The second group consists of properties that preclude the dangerous situations
allowed by the first group. This group includes opacity, TMS1, TMS2, ACA,
strictness, and rigorousness. The properties in this group forbid early release
altogether and obviously are not suited for TM systems that employ that mechanism.

The third group allows early release and precludes overwriting but also precludes aborting in transactions that release early. It includes live opacity, elastic opacity, and VWC.
These properties seem to provide a reasonable middle ground between allowing
early release and eliminating inconsistent views.
However, these properties effectively require that transactions that
release early become irrevocable. That is, once a live transaction is read
from, it can never abort.
The need to deal with irrevocable transactions is detrimental, because
irrevocable transactions introduce additional complexity to a TM (see e.g.,
\cite{WSA08}).
In addition, in applications like distributed computing, transaction aborts may
be induced by external stimuli, so it can be completely impossible to prevent
transactions from aborting \cite{SW13}.
Finally, the requirement to have transactions that release early eventually
commit unnecessarily precludes some intuitively correct histories (see
\rfig{fig:nonaborting-precluded-history}).

In summary, properties from the first group are not adequate for \emph{any} TM
and those from the second group do not allow any form of early release. 
The third group imposes an overstrict restriction that transactions which
release early be irrevocable. 
None of the properties provide a satisfactory, strong safety property that
could be used for a TM with early release, where aborts cannot be arbitrarily
restricted.
Therefore, a property expressing the guarantees of such systems is lacking.
Hence, we introduce a property in \rsec{sec:lopacity} to fill this niche.

\section{Last-use Opacity} \label{sec:lopacity}

\begin{figure}
\begin{minipage}[t]{.08\linewidth}
$\subprog_1$:
\end{minipage}
\begin{minipage}[t]{.4\linewidth}
\vspace{-\baselineskip}
\begin{lstlisting}
transaction { // spawns as $T_1$
    x = 1;
    if (y > 0)
        x = x + y;
    y = x + 1;
}
\end{lstlisting}
\end{minipage}
\begin{minipage}[t]{.08\linewidth}
$\subprog_2$:
\end{minipage}
\begin{minipage}[t]{.4\linewidth}
\vspace{-\baselineskip}
\begin{lstlisting}
transaction { // spawns as $T_2$
    y = y + 1;
}
\end{lstlisting}
\end{minipage}
\caption{\label{fig:last-use-code}Transactional program with \last{} write.}
\end{figure}

We present \emph{last-use opacity}, a new TM safety property that provides
strong consistency guarantees and allows early release without compromising on
the ability of transactions to abort.
The property is based on the preliminary work in \cite{SW14-disc,SW14-wttm}.

The idea of last-use opacity hinges on identifying the \emph{\last{} write}
operation execution on a given variable in individual transactions. 
Informally, a \last{} write on some variable is such, that the 
transaction which executed it will not subsequently execute another write
operation on the same variable in any \emph{possible} extension of the history.
What is possible is determined by the program that is being evaluated to create
that history.
Knowing the program, it is possible to infer (to an extent) what operations a
particular transaction will execute.
Hence, knowing the program, we can determine whether a particular operation on
some variable is the last possible such operation on that variable within a
given transaction.
Thus, we can determine whether a given operation is the \last{} write operation
in a transaction.

Take, for instance, the program in \rfig{fig:last-use-code}, where subprogram
$\subprog_1$ spawns transaction $\tr_1$, and $\subprog_2$ spawns $\tr_2$.
Let us assume that initially {\tt x} and {\tt y} are set to {\tt 0}.
Depending on the semantics of the TM, as these subprograms interweave during
the execution, a number of histories can be produced. We can divide all of
among them into two cases.
In the first case $\tr_2$ writes {\tt 1} to {\tt y} in line 2 of $\subprog_2$
and this value is then read by
$\tr_1$ in line 3 of $\subprog_1$. As a consequence, $\tr_1$ will execute the
write operation in line 4.
The second case assumes that $\tr_1$ reads {\tt 0} in line 3 of $\subprog_1$
(e.g., because $\tr_2$ executed line 2 much later). In this case, $\tr_1$ will
not execute the write operation in line 4.
We can see, however, that in either of the above cases, once $\tr_1$ executes
the write to {\tt x} on line 4, then no further writes to {\tt x} will follow
in $\tr_1$ in any conceivable history. Thus, the write operation execution
generated by line 4 of $\subprog_1$ is the \last{} write on {\tt x} in $\tr_1$.
On the other hand, the write operation execution generated by line 2 of
$\subprog_1$ is never the \last{} write on {\tt x} in $\tr_1$, because there
exists a conceivable history where another write operation execution will
appear (i.e., once line 4 is evaluated). This is true even in the second of
the cases because line 4 can be executed \emph{in potentia}, even if it is not
executed \emph{de facto}.

Note that once any transaction $\tr_i$ completes executing its \last{} write on some
variable $\obj$, it is certain that no further modifications to that variable are
intended by the programmer as part of $\tr_i$.
This means, from the perspective of $\tr_i$ (and assuming no other transaction modifies $\obj$), that the
state of $\obj$ would be the same at the time of the \last{} write as
if the transaction attempted to commit. Hence, with respect to $\obj$,
we can treat $\tr_i$ as if it had attempted to commit.

Last use opacity uses the concept of a \last{} write to dictate one transaction
can read from another transaction. 
We give a formal definition in \rsec{sec:definition}, but, in short, given any
two transactions, $\tr_i$ and $\tr_j$, last-use opacity allows $\tr_i$ to read
variable $\obj$ from $\tr_j$ if the latter is either committed or
commit-pending, or, if $\tr_j$ is live and it already executed its \last{}
write on $\obj$.
This has the benefit of allowing early release while excluding overwriting
completely.
However, last-use opacity does allow cascading aborts to occur.
We discuss their implications in \rsec{sec:implications}, as well as ways of
mitigating them.  That section also describes the guarantees given by last-use
opacity.

\subsection{Definition}
\label{sec:definition}

{
First, we define the concept of a \last{} write to some variable by a
particular transaction. We do this by first defining a \last{} write operation
invocation, and then extend the definition to complete operation executions.
}

Given program $\prog$ and a set of processes $\processes$
executing $\prog$, since different interleavings of $\processes$ cause an
execution $\exec{\prog}{\processes}$ to produce different histories, then let
$\evalhist{\prog}{\processes}$ be the set of all possible histories that can be
produced by $\exec{\prog}{\processes}$, i.e., $\evalhist{\prog}{\processes}$ is
the largest possible set s.t. $\evalhist{\prog}{\processes} = \{ \hist \mid
\hist \models \exec{\prog}{\processes} \}$.

\begin{definition}[\Last{} Write Invocation] \label{def:last-write-inv}
Given a program $\prog$, a set of processes $\processes$ executing $\prog$ and
a history $\hist$ s.t. $\hist \models \exec{\prog}{\processes}$, i.e.  $\hist \in
\evalhist{\prog}{\processes}$, an invocation $\inv{i}{}{\wop{\obj}{\val}}$ is
the \emph{\last{} write invocation} on some variable $\obj$ by transaction $\tr_i$
in $\hist$, if for any history $\hist' \in \evalhist{\prog}{\processes}$ for
which $\hist$ is a prefix (i.e., $\hist' = \hist \cdot R$) there is no
operation invocation $\inv{i}{}{\wop{\obj}{\valu}}$ s.t.
$\inv{i}{}{\wop{\obj}{\val}}$ precedes $\inv{i}{}{\wop{\obj}{\valu}}$ in
$\hist'|\tr_i$.
\end{definition}

\begin{definition}[\Last{} Write] \label{def:last-write-op}
Given a program $\prog$, a set of processes $\processes$ executing $\prog$ and
a history $\hist$ s.t. $\hist \models \exec{\prog}{\processes}$, an operation
execution is the \emph{\last{} write} on some variable $\obj$ by transaction
$\tr_i$ in $\hist$ if it comprises of an invocation and a response other than
$\ab_i$, and the invocation is the \emph{\last{} write invocation} on $\obj$ by
$\tr_i$ in $\hist$.
\end{definition}

The \emph{\last{} read invocation} and the \emph{\last{} read} 
are defined analogously.

If a transaction executes its \last{} write on some variable, we say that
the transaction \emph{decided on} $\obj$.

\begin{definition} [Transaction Decided on $\obj$] \label{def:decided} 
Given a program $\prog$, a set of processes $\processes$ and a history $\hist$
s.t. $\hist \models \exec{\prog}{\processes}$, we say transaction $\tr_i \in
\hist$ \emph{decided on} variable $\obj$ in $\hist$ iff $\hist|\tr_i$ contains
a complete write operation execution $\fwop{i}{\obj}{\val}{\ok_i}$ that is the
\last{} write on $\obj$.
\end{definition}

{
Given some history $\hist$, let $\cpetrans{\hist}$ be a set of transactions
s.t. $\tr_i \in \cpetrans{\hist}$ iff there is some variable $\obj$
s.t. $\tr_i$ decided on $\obj$ in $\hist$.

Given any $\tr_i \in \hist$, a \emph{decided transaction subhistory}, denoted 
$\hist\cpe\tr_i$, is the longest subsequence of $\hist|\tr_i$ s.t.:
\begin{enumerate}[a) ]
    \item $\hist\cpe\tr_i$ contains $\init_i\to\valu$, and
    \item for any variable $\obj$, if 
    $\tr_i$ decided on $\obj$ in $\hist$, 
    then
    $\hist\cpe\tr_i$ contains $(\hist|\tr_i)|\obj$.
\end{enumerate}
In addition, a \emph{decided transaction subhistory completion}, denoted
$\hist\cpeC\tr_i$, is a sequence s.t. $\hist\cpeC\tr_i =
\hist\cpe\tr_i \cdot [\tryC_i\to\co_i]$.

Given a sequential history $S$ s.t. $S\equiv\hist$,
$\luvis{S}{\tr_i}$ is the longest subhistory of $S$, s.t. for each $\tr_j
\in S$:
\begin{enumerate}[a) ]
    \item if $i=j$ or $\tr_j$ is committed in $S$ and $\tr_j \prec_S \tr_i$,
        then $S|\tr_j \subseteq \luvis{S}{\tr_i}$,
    \item  if  $\tr_j$ is not committed in $S$ but $\tr_j \in \cpetrans{\hist}$
        and $\tr_j \prec_S \tr_i$, and it is not true that $\tr_j \prec_\hist
        \tr_i$, then either $S\cpeC\tr_j \subseteq \luvis{S}{\tr_i}$ or not.       
\end{enumerate}

Given a sequential history $S$ and a transaction $\tr_i \in S$, we then say
that transaction $\tr_i$ is \emph{last-use legal in} $S$ if $\luvis{S}{\tr_i}$
is legal.
Note that if $S$ is legal, then it is also last-use legal (see appendix for proof).

} %

\begin{definition} [Final-state Last-use Opacity] \label{def:final-state-lopacity} \label{def:fs-lopacity}
    A finite history $\hist$ is \emph{final-state last-use opaque} if, and only if,
    there exists a sequential history $S$ equivalent to any completion of
    $\hist$ s.t., 
    \begin{enumerate}[a) ] 
        \item $S$ preserves the real-time order of $\hist$,
        \item every transaction in $S$ that is committed in $S$ is
        legal in $S$,
        \item every transaction in $S$ that is not committed in $S$ is
        last-use legal in $S$. 
    \end{enumerate}
\end{definition}

\begin{definition} [Last-use Opacity] \label{def:lopacity}
    A history $\hist$ is \emph{last-use opaque} if, and only if, every finite prefix of
    $\hist$ is final-state last-use opaque.
\end{definition}

{

\begin{theorem} \label{thm:lopacity-safety-property}
    Last-use opacity is a safety property.
\end{theorem}

\begin{proof}
        By \rdef{def:lopacity}, \lopacity{} is trivially prefix-closed.
        
        Given $\hist_L$ that is an infinite limit
        of any sequence of finite histories $\hist_0, \hist_1, ...$, s.t every
        $\hist_h$ in the sequence is \lopaque{} and every $\hist_h$ is a prefix
        of $\hist_{h+1}$, since each prefix $\hist_h$ of $\hist_L$ is
        \lopaque{}, then, by extension, every prefix $\hist_h$ of $\hist_L$ is
        also final-state last-use opaque, so, by \rdef{def:lopacity}, $\hist_L$ is
        \lopaque{}. Hence, \lopacity{} is limit-closed.

        Since \lopacity{} is both prefix-closed and limit-closed, then, by
        \rdef{def:safety-property}, it is a safety property.
\end{proof}

} %

\subsection{Examples}
\label{sec:examples}

In order to aid understanding of the property we present examples of last-use
opaque histories in
\rfig{fig:example-early-release}--\ref{fig:example-abort-before-commit}.
These are contrasted by examples of histories that are not last-use opaque in
\rfig{fig:example-early-release-not-last}--\ref{fig:example-overwriting}. We
discuss the examples below and prove them in the appendix.

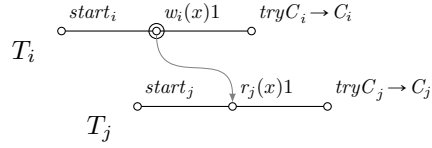
\begin{figure}
\begin{center}
\begin{tikzpicture}
     \draw
           (0,2)        node[tid]       {$\tr_i$}
                        node[aop]       {$\init_i$} %
                        node[dot]       {} 

      -- ++(1.25,0)     node[aop]       {$\twop{i}{\obj}{1}$}
                        node[dot] (wi)  {}
                        node[cir]       {}

      -- ++(1.25,0)     node[aop]       {$\tryC_i\!\to\!\co_i$}
                        node[dot]       {}         
                        ;

     \draw
           (1,1)        node[tid]       {$\tr_{j}$}
                        node[aop]       {$\init_{j}$} %
                        node[dot]       {} 

      -- ++(1.25,0)     node[aop]       {$\trop{j}{\obj}{1}$}
                        node[dot] (rj)  {}

      -- ++(1.25,0)     node[aop]       {$\tryC_j\!\to\!\co_j$}
                        node[dot]       {}
                        ;
     
     \draw[hb] (wi) \squiggle (rj);
\end{tikzpicture}
\end{center}
\caption{\label{fig:example-early-release}
Example satisfying last-use opacity: early release.
We mark a \last{} write operation execution in some history in the diagram as 
\protect\tikz{
    \protect\draw[] (0,0) -- ++(0.25,0) node[dot] {} node[cir] {} -- ++(0.25,0);
}.
Note that an operation can be the ultimate operation execution in some
transaction, but still not fit the definition of a \last{} operation execution.
}
\end{figure}

\begin{figure}
\begin{center}
\begin{tikzpicture}
     \draw
           (0,2)        node[tid]       {$\tr_i$}
                        node[aop]       {$\init_i$} %
                        node[dot]       {} 

      -- ++(1.25,0)     node[aop]       {$\twop{i}{\obj}{1}$}
                        node[dot] (wi)  {}
                        node[cir]       {}

      -- ++(1.25,0)     node[aop]       {$\tryC_i\!\to\!\co_i$}
                        node[dot]       {}         
                        ;

     \draw
           (0.75,1)     node[tid]       {$\tr_{j}$}
                        node[aop]       {$\init_{j}$} %
                        node[dot]       {} 

      -- ++(1.25,0)     node[aop]       {$\trop{j}{\obj}{1}$}
                        node[dot] (rj)  {}

      -- ++(1.25,0)     node[aop]       {$\tryA_j\!\to\!\ab_j$}
                        node[dot]       {}
                        ;
     
     \draw[hb] (wi) \squiggle (rj);
\end{tikzpicture}
\end{center}
\caption{\label{fig:example-commit-abort}
Example satisfying last-use opacity: early release to an aborting transaction.
}
\end{figure}
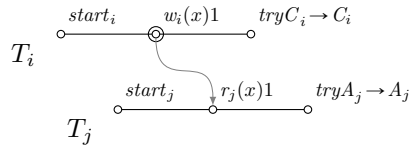

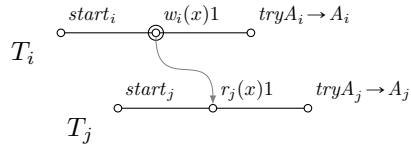
\begin{figure}
\begin{center}
\begin{tikzpicture}
     \draw
           (0,2)        node[tid]       {$\tr_i$}
                        node[aop]       {$\init_i$} %
                        node[dot]       {} 

      -- ++(1.25,0)     node[aop]       {$\twop{i}{\obj}{1}$}
                        node[dot] (wi)  {}
                        node[cir]       {}

      -- ++(1.25,0)     node[aop]       {$\tryA_i\!\to\!\ab_i$}
                        node[dot]       {}         
                        ;

     \draw
           (0.75,1)     node[tid]       {$\tr_{j}$}
                        node[aop]       {$\init_{j}$} %
                        node[dot]       {} 

      -- ++(1.25,0)     node[aop]       {$\trop{j}{\obj}{1}$}
                        node[dot] (rj)  {}

      -- ++(1.25,0)     node[aop]       {$\tryA_j\!\to\!\ab_j$}
                        node[dot]       {}
                        ;
     
     \draw[hb] (wi) \squiggle (rj);
\end{tikzpicture}
\end{center}
\caption{\label{fig:example-abort-abort}
Example satisfying last-use opacity: early release between two aborting
transactions.
}
\end{figure}

\begin{figure}
\begin{center}
\begin{tikzpicture}
     \draw
           (0,2)        node[tid]       {$\tr_i$}
                        node[aop]       {$\init_i$} %
                        node[dot]       {} 

      -- ++(1.25,0)     node[aop]       {$\twop{i}{\obj}{1}$}
                        node[dot] (wi)  {}
                        node[cir]       {}

      -- ++(2.5,0)      node[aop]       {$\tryC_i\!\to\!\co_i$}
                        node[dot]       {}         
                        ;

     \draw
           (0.75,1)     node[tid]       {$\tr_{j}$}
                        node[aop]       {$\init_{j}$} %
                        node[dot]       {} 

      -- ++(1.25,0)     node[aop]       {$\trop{j}{\obj}{1}$}
                        node[dot] (rj)  {}

      -- ++(1.25,0)     node[aop]       {$\tryA_j\!\to\!\ab_j$}
                        node[dot]       {}
                        ;
     
     \draw[hb] (wi) \squiggle (rj);
\end{tikzpicture}
\end{center}
\caption{\label{fig:example-abort-before-commit}
Example satisfying last-use opacity: early release to a prematurely aborting
transaction.
}
\end{figure}
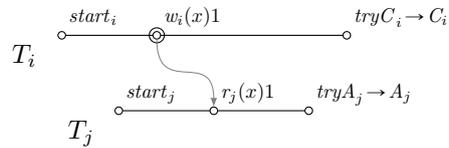

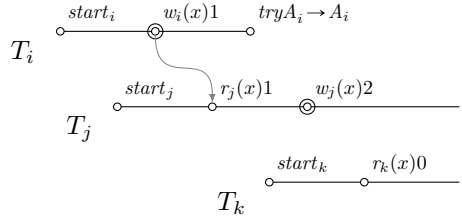
\begin{figure}
\begin{center}
\begin{tikzpicture}
     \draw
           (0,2)        node[tid]       {$\tr_i$}
                        node[aop]       {$\init_i$} %
                        node[dot]       {} 

      -- ++(1.25,0)     node[aop]       {$\twop{i}{\obj}{1}$}
                        node[dot] (wi)  {}
                        node[cir]       {}

      -- ++(1.25,0)      node[aop]       {$\tryA_i\!\to\!\ab_i$}
                        node[dot]       {}         
                        ;

     \draw
           (0.75,1)     node[tid]       {$\tr_{j}$}
                        node[aop]       {$\init_{j}$} %
                        node[dot]       {} 

      -- ++(1.25,0)     node[aop]       {$\trop{j}{\obj}{1}$}
                        node[dot] (rj)  {}

      -- ++(1.25,0)     node[aop]       {$\twop{j}{\obj}{2}$}
                        node[dot]       {}
                        node[cir]       {}                       

      -- ++(2,0)     node            {}       
                        ;

      \draw
           (2.75,0)   node[tid]       {$\tr_{k}$}
                        node[aop]       {$\init_{k}$} %
                        node[dot]       {} 

      -- ++(1.25,0)        node[aop]       {$\trop{k}{\obj}{0}$}
                        node[dot]       {}

      -- ++(1.25,0)     node            {}

                        ;                       
     
     \draw[hb] (wi) \squiggle (rj);
\end{tikzpicture}
\end{center}
\caption{\label{fig:example-freedom}
Example satisfying last-use opacity: freedom to read from or ignore an aborted transaction.
}
\end{figure}

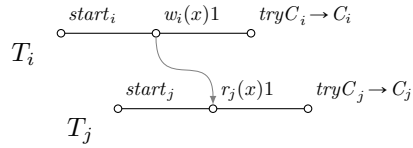
\begin{figure}
\begin{center}
\begin{tikzpicture}
     \draw
           (0,2)        node[tid]       {$\tr_i$}
                        node[aop]       {$\init_i$} %
                        node[dot]       {} 

      -- ++(1.25,0)     node[aop]       {$\twop{i}{\obj}{1}$}
                        node[dot] (wi)  {}

      -- ++(1.25,0)     node[aop]       {$\tryC_i\!\to\!\co_i$}
                        node[dot]       {}         
                        ;

     \draw
           (0.75,1)     node[tid]       {$\tr_{j}$}
                        node[aop]       {$\init_{j}$} %
                        node[dot]       {} 

      -- ++(1.25,0)     node[aop]       {$\trop{j}{\obj}{1}$}
                        node[dot] (rj)  {}

      -- ++(1.25,0)     node[aop]       {$\tryC_j\!\to\!\co_j$}
                        node[dot]       {}
                        ;
     
     \draw[hb] (wi) \squiggle (rj);
\end{tikzpicture}
\end{center}
\caption{\label{fig:example-early-release-not-last}Example breaking last-use
opacity: early release before \last{} write operation execution.}
\end{figure}

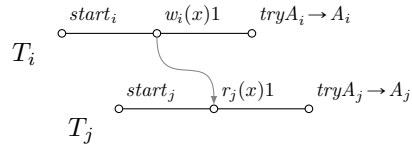
\begin{figure}
\begin{center}
\begin{tikzpicture}
     \draw
           (0,2)        node[tid]       {$\tr_i$}
                        node[aop]       {$\init_i$} %
                        node[dot]       {} 

      -- ++(1.25,0)     node[aop]       {$\twop{i}{\obj}{1}$}
                        node[dot] (wi)  {}

      -- ++(1.25,0)     node[aop]       {$\tryA_i\!\to\!\ab_i$}
                        node[dot]       {}         
                        ;

     \draw
           (0.75,1)     node[tid]       {$\tr_{j}$}
                        node[aop]       {$\init_{j}$} %
                        node[dot]       {} 

      -- ++(1.25,0)     node[aop]       {$\trop{j}{\obj}{1}$}
                        node[dot] (rj)  {}

      -- ++(1.25,0)     node[aop]       {$\tryA_j\!\to\!\ab_j$}
                        node[dot]       {}
                        ;
     
     \draw[hb] (wi) \squiggle (rj);
\end{tikzpicture}
\end{center}
\caption{\label{fig:example-abort-abort-not-last}
Example breaking last-use opacity: early release between two aborting
transactions before \last{} write operation execution.
}
\end{figure}

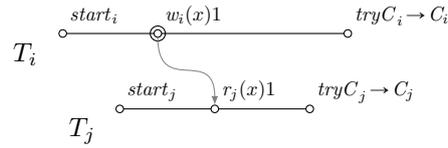
\begin{figure}
\begin{center}
\begin{tikzpicture}
     \draw
           (0,2)        node[tid]       {$\tr_i$}
                        node[aop]       {$\init_i$} %
                        node[dot]       {} 

      -- ++(1.25,0)     node[aop]       {$\twop{i}{\obj}{1}$}
                        node[dot] (wi)  {}
                        node[cir]       {}

      -- ++(2.5,0)      node[aop]       {$\tryC_i\!\to\!\co_i$}
                        node[dot]       {}         
                        ;

     \draw
           (0.75,1)     node[tid]       {$\tr_{j}$}
                        node[aop]       {$\init_{j}$} %
                        node[dot]       {} 

      -- ++(1.25,0)     node[aop]       {$\trop{j}{\obj}{1}$}
                        node[dot] (rj)  {}

      -- ++(1.25,0)     node[aop]       {$\tryC_j\!\to\!\co_j$}
                        node[dot]       {}
                        ;
     
     \draw[hb] (wi) \squiggle (rj);
\end{tikzpicture}
\end{center}
\caption{\label{fig:example-commit-before-commit}
Example breaking last-use opacity: commit order not respected.
}
\end{figure}

\begin{figure}
\begin{center}
\begin{tikzpicture}
     \draw
           (0,2)        node[tid]       {$\tr_i$}
                        node[aop]       {$\init_i$} %
                        node[dot]       {} 

      -- ++(1.25,0)     node[aop]       {$\twop{i}{\obj}{1}$}
                        node[dot] (wi)  {}

      -- ++(1.25,0)     node[aop]       {$\twop{i}{\obj}{2}$}
                        node[dot]       {}
                        node[cir]       {}

      -- ++(1.25,0)     node[aop]       {$\tryC_i\!\to\!\co_i$}
                        node[dot]       {}         
                        ;

     \draw
           (0.75,1)     node[tid]       {$\tr_{j}$}
                        node[aop]       {$\init_{j}$} %
                        node[dot]       {} 

      -- ++(1.25,0)     node[aop]       {$\trop{j}{\obj}{1}$}
                        node[dot] (rj)  {}

      -- ++(2.5,0)      node[aop]       {$\tryC_j\!\to\!\co_j$}
                        node[dot]       {}
                        ;
     
     \draw[hb] (wi) \squiggle (rj);
\end{tikzpicture}
\end{center}
\caption{\label{fig:example-overwriting}Example breaking last-use opacity:
early release with overwriting.}
\end{figure}
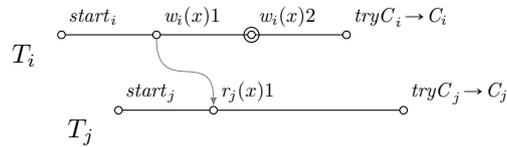

\begin{figure}
\begin{center}
\begin{tikzpicture}
     \draw
           (0,2)        node[tid]       {$\tr_i$}
                        node[aop]       {$\init_i$} %
                        node[dot]       {} 

      -- ++(1.25,0)     node[aop]       {$\twop{i}{\objx}{1}$}
                        node[dot] (wi)  {}
                        node[cir]       {}

      -- ++(2.25,0)     node[aop]       {$\trop{i}{\objy}{1}$}
                        node[dot] (ri)  {}

      -- ++(1.25,0)     node[]{} %
                        ;

     \draw
           (0,1)        node[tid]       {$\tr_{j}$}
                        node[aop]       {$\init_{j}$} %
                        node[dot]       {} 

      -- ++(1.75,0)      node[aop]       {$\trop{j}{\objx}{1}$}
                        node[dot] (rj)  {}

      -- ++(1.25,0)     node[aop]       {$\twop{j}{\objy}{1}$}
                        node[dot] (wj)  {}  
                        node[cir]       {}                     

     -- ++(1.75,0)      node[]{} %
                        ;
     
     \draw[hb] (wi) \squiggle (rj);
     \draw[hb] (wj) \squiggleup (ri);
\end{tikzpicture}
\end{center}
\caption{\label{fig:example-cycle}Example breaking last-use opacity:
dependency cycle.}
\end{figure}
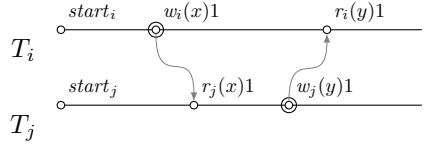

The example in \rfig{fig:example-early-release} shows $\tr_i$ executing a
write on $\obj$ once and releasing $\obj$ early to $\tr_j$. We assume
that the program generating the history is such, that the write operation
executed by $\tr_i$ is the \last{} write operation execution on $\obj$.
The history is intuitively correct, since both transactions commit, and $\tr_j$
reads a value written by $\tr_i$.
On the formal side, since both transactions are committed in this history,
the equivalent sequential history would consist of all the events in $\tr_i$
followed by the events in $\tr_j$ and both transactions would be legal, since
$\tr_i$ writes a legal value to $\obj$ and $\tr_j$ reads the last value written
by $\tr_i$ to $\obj$. Thus, the history is final-state last-use opaque.

Since last-use opacity requires prefix closedness, then all prefixes of the
history in \rfig{fig:example-early-release} also need to be final-state
last-use opaque. We present only two of the interesting prefixes, since the
remainder are either similar or trivial.
The first interesting prefix is created by removing the commit operation
execution from $\tr_j$, which means $\tr_j$ is aborted in any completion of the
history. We show such a completion in \rfig{fig:example-commit-abort}.
Still, $\tr_i$ writes a legal value to $\obj$ and $\tr_j$ reads the
last value written by $\tr_i$ to $\obj$, so that prefix is also final-state
last-use opaque.
Another interesting prefix is created by removing the commit operation
executions from both transactions. Then, in the completion of the history both
transactions are aborted, as in \rfig{fig:example-abort-abort}. Then, in an
equivalent sequential history $\tr_j$ would read a value written by an aborted
transaction.  
In order to show legality of a committed transaction, we use the subhistory
denoted $\visf$, which does not contain any transactions that were not
committed in the history from which it was derived. Thus, if $\tr_j$ were
committed, it would not be legal, since its $\visf$ would not contain a
write operation execution writing the value the transaction actually read.
However, since $\tr_j$ is aborted, the definition of final-state last-use opacity only requires that
$\luvisf$ rather than $\visf$ be legal, and $\luvisf$ can contain operation
executions on particular variables from an aborted transaction under the
condition that the transaction already executed its \last{} write on the variables
in question. Since, in the example $\tr_i$ executed its \last{} write on $\obj$,
then this write will be included in $\luvisf$ for $\tr_j$, so $\tr_j$ will be
last-use legal. In consequence the prefix is also final-state last-use opaque.
Indeed, all prefixes of example \rfig{fig:example-early-release} are
final-state last-use opaque, so the example is last-use opaque, and, by
extension, so are the examples in \rfig{fig:example-commit-abort} and
\rfig{fig:example-abort-abort}.

Contrast the example in \rfig{fig:example-early-release} with the one in
\rfig{fig:example-early-release-not-last}. The histories presented in both are
identical, with the exception that the write operation in
\rfig{fig:example-early-release} is considered to be the \last{} operation
execution, while in \rfig{fig:example-early-release-not-last} it is not. The
difference would stem from differences in the programs that produced these
histories. For instance, the program producing the history in
\rfig{fig:example-early-release-not-last} could conditionally execute another
operation on $\obj$, so, even though that condition was not met in this
history, the potential of another write on $\obj$ means that the existing write
cannot be considered a \last{} write operation execution.
The consequence of this is that while the example itself is final-state
last-use opaque, one of its prefixes is not, so the history is not last-use
opaque. The offending prefix is created by removing commit operations in both
transactions, so both transactions would abort in any completion, as in
\rfig{fig:example-abort-abort-not-last}.
Here, since $\tr_i$ does not execute the \last{} write operation on $\obj$, then
the write operation would not be included in $\luvisf$ for $\tr_j$, so the
value read by $\tr_j$ could not be justified. Thus, $\tr_j$ is not legal in
that history, and, therefore, the history in \rfig{fig:example-abort-abort-not-last} is not
final-state last-use opaque (so also not last-use opaque). 
\rfig{fig:example-abort-abort-not-last} represents the completion of a prefix of the history in \rfig{fig:example-early-release-not-last}, so
\rfig{fig:example-abort-abort-not-last} not being final-state last-use opaque, means that \rfig{fig:example-early-release-not-last} is not last-use opaque.

The examples in \rfig{fig:example-abort-before-commit} and
\rfig{fig:example-commit-before-commit}, show that recoverability is required, i.e.,
transactions must commit in order.  Last-use opacity
of the example in \rfig{fig:example-abort-before-commit} is analogous to the
one in \rfig{fig:example-commit-abort}, since their equivalent sequential
histories are identical, as are the sequential histories equivalent to their
prefixes.
Furthermore, intuitively, if $\tr_j$ reads a value of a variable
released early by $\tr_i$ and aborts before $\tr_i$ commits, this is
correct behavior.
On the other hand, the history in \rfig{fig:example-commit-before-commit} is
not last-use opaque, even though it is final-state last-use opaque (by analogy
to \rfig{fig:example-early-release}). However, a prefix of the history where the
commit operation execution is removed from $\tr_i$ is not final-state last-use
opaque. This is because a completion will require that $\tr_i$ be aborted, the
operations executed by $\tr_i$ are not going to be included in any $\visf$.
Since $\tr_j$ is committed, then its $\visf$ must be legal, but it is not,
because the read operation reading $1$ will not be preceded by any writes in
$\visf$. Since the prefix contains an illegal transaction, then it is not
final-state last-use opaque, and thus, the history in
\rfig{fig:example-commit-before-commit} is not last-use opaque.

The example in \rfig{fig:example-freedom} shows that a transaction is allowed
to read from a transaction that eventually aborts, or ignore that transaction,
because of the freedom left within the definition of $\luvisf$. 
I.e.  transactions $\tr_j$ is concurrent to $\tr_i$, but $\tr_k$ follows
$\tr_i$ in real time.  $\tr_i$ executes a \last{} write on $\objx$, so $\tr_j$
is allowed to include the write operation on in its $\luvisf$. Since $\tr_j$
sees the value written to $\objx$ by that write, $\tr_j$ includes the write in
$\luvisf$.
On the other hand, $\tr_k$ cannot include $\tr_i$'s write in $\luvis$, since it
aborted before $\tr_k$ even started, so the write should not be visible to
$\tr_k$. On the other hand $\tr_k$ is allowed to include $\tr_j$ in its
$\luvisf$. $\tr_k$ should not do so, however, since it ignores $\tr_j$ as well
as $\tr_i$ (which makes sense as $\tr_j$ is doomed to abort). Hence $\tr_k$
reads the value of $\obj$ to be 0. If $\tr_j$ is included in $\tr_k$'s
$\luvisf$, reading 0 would be incorrect. Hence, the definition of $\luvisf$
allows $\tr_j$ to be arbitrarily excluded. In effect all three transactions are
correct (so long as $\tr_j$ does not eventually commit).

\rfig{fig:example-overwriting} shows an example of overwriting, which
is not last-use opaque, since there is no equivalent sequential history where
the write operation in $\tr_i$ writing $1$ to $\obj$ would precede the read
operation in $\tr_j$ reading $1$ from $\obj$ without the other write operation
writing $2$ to $\obj$ also preceding the read. Thus, in all cases $\tr_j$ is
not legal, and the history is neither final-state last-use opaque, nor last-use opaque.

Finally, \rfig{fig:example-cycle} shows an example of a cyclic dependency,
where $\tr_j$ reads $\objx$ from $\tr_i$, and subsequently $\tr_i$ reads
$\objy$ from $\tr_j$. Both writes in the history are \last{} writes. This
example has unfinished transactions, which are thus aborted in any possible
completion of this history. There are two possible sequential histories
equivalent to that completion: one where $\tr_i$ precedes $\tr_j$ and one where
$\tr_j$ precedes $\tr_i$. In the former case, $\luvisf$ of $\tr_i$ does not
contain any operations from $\tr_j$, because $\tr_j$ follows $\tr_i$. Thus,
there is no write operation on $\objy$ preceding a read on $\objy$ returning
$1$ in $\tr_i$'s $\luvisf$, which does not conform to the sequential
specification, so $\tr_i$'s $\luvisf$ is not legal. Hence, $\tr_i$ is not legal
in that scenario.  The former case is analogous: $\tr_j$'s $\luvisf$ will not
contain a write operation from $\tr_i$, because $\tr_i$ follows $\tr_j$.
Therefore $\tr_j$'s $\luvisf$ contains a read on $\objx$ that returns $1$,
which is not preceded by any write on $\objx$, which causes the sequence not to
conform to the sequential specification and renders the transaction not legal.
Since either case contains a transaction that is not legal, then that history
is not final-state last-use opaque, and therefore not last-use opaque.

\subsection{Guarantees}
\label{sec:implications}

Last-use opacity gives the programmer the following guarantees:

\paragraph{Serializability}
If a transaction commits, then the value it reads can be explained by
operations executed by preceding or concurrent transactions. This guarantees
that a transaction that views inconsistent state will not commit.

\begin{lemma} 
\label{lemma:lop-serializability}
    Every \lopaque{} history is serializable.
\end{lemma}

We provide the proof for \rlemma{lemma:comp-serializability}. %

\paragraph{Real-time Order}
Successive transactions will not be rearranged to fit serializability, so a
correct history will agree with an external clock, or an external order of
events.

\begin{lemma} %
    \label{lemma:lop-rt-order}
    Every \lopaque{} history preserves real-time order.
\end{lemma}

\begin{proof}
    Trivially from \rdef{def:lopacity} and \rdef{def:fs-lopacity}a.
\end{proof}

\paragraph{Recoverability}
If one transaction reads from another transaction, the former will commit only
after the latter commits. This guarantees that transactions commit in order. 

\begin{lemma} %
\label{lemma:lop-recoverability}
    Every \lopaque{} history is recoverable.
\end{lemma}

We provide the proof in \rappx{sec:property-comparison}.

\paragraph{Precluding Overwriting}
If transaction $\tr_i$ reads the value of some variable written by transaction
$\tr_j$, then $\tr_j$ will never subsequently modify that variable.

\begin{lemma}%
\label{lemma:lop-no-overwriting}
    Last-use opacity does not support overwriting.
\end{lemma}

\begin{proof}
    For the sake of contradiction let us assume that there exists $\hist$ that
    is a \lopaque{} history with overwriting, i.e. (from
    \rdef{def:overwriting-support}) there are transaction $\tr_i$ and $\tr_j$
    s.t.:
    \begin{enumerate}[a) ]
    \item $\tr_i$ releases some variable $\obj$ early,
    \item $\hist|\tr_i$ contains %
        $\fwop{i}{\obj}{\val}{\ok_i}$ and
        $\fwop{i}{\obj}{\val'}{\ok_i}$, s.t. the former precedes the latter in
        $\hist|\tr_i$,
    \item $\hist|\tr_j$ contains %
        $\frop{j}{\obj}{\val}$ that precedes $\fwop{i}{\obj}{\val'}{\ok_i}$ in
        $\hist.$
    \end{enumerate}
    Since $\hist$ is opaque, then there is a completion $C = \compl{\hist}$ and
    a sequential history $S$ s.t. $S \equiv \hist$, $S$ preserves the real-time
    order of $\hist$, and both $\tr_i$ and $\tr_j$ in $S$ are legal in $S$.
    In $S$, either $\tr_i \prec_S \tr_j$ or $\tr_j \prec_S \tr_i$.
    In either case, any $\vis{S}{\tr_j}$ or $\luvis{S}{\tr_j}$ by their
    definitions will contain either the sequence of both
    $\fwop{i}{\obj}{\val}{\ok_i}$ and $\fwop{i}{\obj}{\val'}{\ok_i}$ or neither
    of those write operation executions. %
    In either case, $\frop{j}{\obj}{\val}$ will not be directly preceded by
    $\fwop{i}{\obj}{\val}{\ok_i}$  among operations on $\obj$ in either
    $\vis{S}{\tr_j}$ or $\luvis{S}{\tr_j}$. Therefore, $\tr_j$ in $S$
    cannot be legal in $S$, which is a contradiction.    
\end{proof}

\paragraph{Aborting Early Release}
A transaction can release some variable early and subsequently abort.

\begin{lemma} %
\label{lemma:lop-aborting-early-release}
    Last-use opacity{} supports aborting early release.
\end{lemma}

\begin{proof} 
    Let $\hist$ be the history depicted in \rfig{fig:example-abort-abort}.
    Here, $\tr_i$ releases $\obj$ early to $\tr_j$ and subsequently aborts,
    which satisfies \rdef{def:aborting-support}.
    Since $\tr_i$ and $\tr_j$ are both aborted in $\hist$, $\hist$ has a
    completion $C = \compl{\hist} = \hist$. Let $S$ be a sequential history
    s.t. $S=\hist|\tr_i \cdot \hist|\tr_j$. $S$ vacuously preserves the
    real-time order of $\hist$ and trivially $S \equiv \hist$. 
    Transaction $\tr_i$ in $S$ is last-use legal in $S$, because
    $\luvis{S}{\tr_i} = \hist\cpeC\tr_i$ whose operations on $\obj$ are limited
    to a single write operation execution is within the sequential
    specification of $\obj$.
    Transaction  $\tr_j$ in $S$ is also last-use legal in $S$, since
    $\luvis{S}{\tr_j} = \hist\cpeC\tr_i \cdot \hist\cpeC\tr_j$ whose operations
    on $\obj$ consist of $\fwop{i}{\obj}{\val}{\ok_i}$ followed by
    $\frop{j}{\obj}{\val}$ is also within the sequential specification of
    $\obj$.
    Since both $\tr_i$ and $\tr_j$ in $S$ are last-use legal in $S$, $\hist$ is
    final-state last-use opaque.
    All prefixes of $\hist$ are trivially also final-state last-use opaque
    (since either their completion is the same as $\hist$'s, they contain only
    a single write operation execution on $\obj$, or contain no operation
    executions on variables), so $\hist$ is last-use opaque.
\end{proof}

\paragraph{Exclusive Access}
Any transaction has effectively exclusive
access to any variable it accesses, at minimum, from the first to the final
modification it performs, regardless of whether it eventually commits or aborts.

\begin{lemma} %
\label{lemma:lop-exclusive-access}
    Any transaction in any \lopaque{} history has exclusive access to any
    variable between first and last access to.
\end{lemma}

\begin{proof}
    From \rlemma{lemma:lop-serializability} and
    \rlemma{lemma:lop-no-overwriting}.
\end{proof}

\begin{figure}
\hspace{.5cm}
\begin{minipage}[t]{.49\linewidth}
\begin{lstlisting}
// invariant: $\mathtt{x}\geq\mathtt{0}$
transaction {
    x = y - 1; 
    if (x < 0) 
        abort();
}
\end{lstlisting}
\end{minipage}
\begin{minipage}[t]{.49\linewidth}
\begin{lstlisting}
// invariant: $\mathtt{x}\geq\mathtt{0}$
transaction {
    *(_array + x);  
}
\end{lstlisting}
\end{minipage}

\begin{minipage}[t]{.49\linewidth}
\captionof{figure}{\label{fig:abort-example}Abort example.}
\end{minipage}
\begin{minipage}[t]{.49\linewidth}
\captionof{figure}{\label{fig:abort-example-c}Memory error example.}
\end{minipage}
\end{figure}

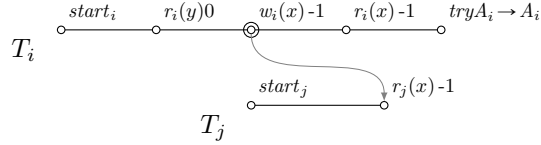
\begin{figure}
\begin{center}
\begin{tikzpicture}
     \draw
           (0,2)        node[tid]       {$\tr_i$}
                        node[aop]       {$\init_i$} %
                        node[dot]       {} 

      -- ++(1.25,0)     node[aop]       {$\trop{i}{\objy}{0}$}
                        node[dot]       {}

      -- ++(1.25,0)     node[aop]       {$\twop{i}{\objx}{\,\shortminus1}$}
                        node[dot] (wi)  {}
                        node[cir]       {}

      -- ++(1.25,0)     node[aop]       {$\trop{i}{\objx}{\,\shortminus1}$}
                        node[dot]       {}

      -- ++(1.25,0)     node[aop]       {$\tryA_i\!\to\!\ab_i$}
                        node[dot]       {}         
                        ;

     \draw
           (2.5,1)      node[tid]       {$\tr_{j}$}
                        node[aop]       {$\init_{j}$} %
                        node[dot]       {} 

      -- ++(1.75,0)     node[aop]       {$\trop{j}{\obj}{\,\shortminus1}$}
                        node[dot] (rj)  {}
                        ;
     
     \draw[hb] (wi) \squiggle (rj);
\end{tikzpicture}
\end{center}
\caption{\label{fig:abort-history}Last-use opaque history with inconsistent view.}
\end{figure}

However, last-use opacity does not preclude transactions from aborting after
releasing a variable early. As a consequence there may be instances of
cascading aborts, which have varying implications on consistency depending
on whether the TM model allows transactions to abort programmatically.
We distinguish three cases of models and discuss them below.

\paragraph{Only Forced Aborts}
Let us assume that transactions cannot arbitrarily abort, but only do so
as a result of receiving an abort response to invoking a read or write
operation, or while attempting to commit.
In other words, there is no $\tryA$ operation in the transactional API.
In that case, since overwriting is not allowed, the transaction never reveals
intermediate values of variables to other transactions. This means, that if a
transaction released a variable early, then the programmer did not intend to
change the value of that variable. So, if the transaction eventually committed,
the value of the variable would have been the same.
So, if the transaction is eventually forced to abort rather than committing, the
value of any variable released early would be the same 
regardless of whether the transaction committed or aborted. 
Therefore, we can consider the inconsistent state to be safe.
In other words, if the variable caused an error to occur, the error would be
caused regardless of whether the transaction finally aborts or commits.
Thus, we can say that with this set of assumptions, the programmer is
guaranteed that none of the inconsistent views will cause unexpected behavior,
even if cascading aborts are possible.
Note that the use of this model is not uncommon (see eg. \cite{FGG09,AGHR13,AGHR14}).

\paragraph{Programmer-initiated Aborts}
Alternatively, let us assume that transactions can arbitrarily abort (in
addition to forced aborts as described above) by executing the operation
$\tryA$ as a result of some instruction in the program.
In that case it is possible to imagine programs that use the abort instruction
to cancel transaction due to the "business logic" of the program. Therefore a
programmer explicitly specifies that the value of a variable is different
depending on whether the transaction finally commits or not.
An example of such a program is given in \rfig{fig:abort-example}. Here, the
programmer enforced an invariant that the value of $x$ should never be less
than zero. If the invariant is not fulfilled, the transaction aborts.  However,
writing a value to $x$ that breaks the invariant is the \last{} write operation
execution for this program, so it is possible that another transaction reads
the value of $x$ before the transaction aborts.
If the transaction that reads $x$ is like the one in 
\rfig{fig:abort-example-c}, where $x$ is used to index an array via pointer
arithmetic, a memory error is possible.
Nevertheless, the history from \rfig{fig:abort-history} that corresponds to a
problematic execution of these two transactions is clearly allowed by last-use
opacity (assuming that the domain of $x$ is $\mathbb{Z}$). 
Thus, if the abort operation is available to the programmer the guarantee  
that inconsistent views will not lead to unexpected effects is lost.
Therefore it is up to the programmer to use aborts wisely or to prevent
inconsistent views from causing problems, by prechecking invariants at the
outset of a transaction, or maintaining invariants also within a transaction
(in a similar way as with monitor invariants).
Alternatively, a mechanism can be built into the TM that prevents specific
transactions at risk from reading variables that were released early, while
other transactions are allowed to do so.
However, if these workarounds are not satisfactory, we present a stronger
variant of last-use opacity in \rsec{sec:blu-opacity} that deals specifically
with this model and eliminates its inconsistent views.

\paragraph{Arbitrary Aborts}
We present a third alternative to aborts in transactions: a compromise
between only forced aborts and programmer-initiated aborts. This option assumes
that the $\tryA$ operation is not available to the programmer, so it cannot be
used to implement business logic. However, we allow the TM system to somehow
inject $\tryA$ operations in the code in response to external stimuli, such as
crashes or exceptions and use aborts as a fault tolerance mechanism.
However, since the programmer cannot use the operation, the programs must be
coded as in the forced aborts case, and therefore the same guarantees are
given.

\subsection{Last-use Opacity in Context}
\label{sec:lopacity-comparison}

We compare last-use opacity with other safety properties with respect to their
relative strength. Given two properties $\property_1$ and $\property_2$ and
the set of all histories that satisfy each property $\mathbb{\hist}_1$,
$\mathbb{\hist}_2$, respectively. $\property_1$ is stronger than $\property_2$ if
$\mathbb{\hist}_1 \subset \mathbb{\hist}_2$ (so $\property_2$ is weaker than
$\property_1$). If neither $\mathbb{\hist}_1 \subset \mathbb{\hist}_2$ nor
$\mathbb{\hist}_1 \supset \mathbb{\hist}_2$, then the properties are
incomparable.

We present the result of the comparison in \rfig{fig:properties}.  We describe
the comparison with opacity and serializability in particular below, and
provide proofs for the comparison of the remaining properties in the appendix.

\begin{figure}[t]
    \hspace{.09\textwidth}
    \begin{minipage}{.2\textwidth}
    {%
    $\begin{array}{lclclcl}
            \mathbb{H}_{tms2} &\subset&    \mathbb{H}_{lop}  \\ 
            \mathbb{H}_{op}   &\subset&    \mathbb{H}_{lop}  \\ 
            \mathbb{H}_{rig}  &\subset&    \mathbb{H}_{lop}  \\ 
            \mathbb{H}_{lvop} &\subset&    \mathbb{H}_{lop}  \\
            \mathbb{H}_{rec}  &\supset&    \mathbb{H}_{lop}  \\
            \mathbb{H}_{ser}  &\supset&    \mathbb{H}_{lop}  \\
    \end{array}$}
    \end{minipage}
    \hspace{.09\textwidth}
    \begin{minipage}{.49\textwidth}
    {%
     $\begin{array}{lclclcl}
            \mathbb{H}_{tms1} &\not\subset& \mathbb{H}_{lop} &\wedge& 
            \mathbb{H}_{tms1} &\not\supset& \mathbb{H}_{lop} \\
            \mathbb{H}_{eop}  &\not\subset& \mathbb{H}_{lop} &\wedge& 
            \mathbb{H}_{eop}  &\not\supset& \mathbb{H}_{lop} \\ 
            \mathbb{H}_{str}  &\not\subset& \mathbb{H}_{lop} &\wedge& 
            \mathbb{H}_{str}  &\not\supset& \mathbb{H}_{lop} \\ 
            \mathbb{H}_{vwc}  &\not\subset& \mathbb{H}_{lop} &\wedge& 
            \mathbb{H}_{vwc}  &\not\supset& \mathbb{H}_{lop} \\ 
            \mathbb{H}_{aca}  &\not\subset& \mathbb{H}_{lop} &\wedge& 
            \mathbb{H}_{aca}  &\not\supset& \mathbb{H}_{lop} \\      
            && \\    
    \end{array}$}
    \end{minipage}
    \caption{\label{fig:properties} Last-use opaque histories $\mathbb{H}_{lop}$ 
    in relation to:
    TMS2 and TMS1 histories $\mathbb{H}_{tms2}$, $\mathbb{H}_{tms1}$,
    opaque histories $\mathbb{H}_{op}$,
    elastic opaque histories $\mathbb{H}_{eop}$,
    rigorous histories $\mathbb{H}_{rig}$,
    strict histories $\mathbb{H}_{str}$,
    live opaque histories $\mathbb{H}_{lvop}$, 
    virtual world consistent histories $\mathbb{H}_{vwc}$,
    recoverable histories $\mathbb{H}_{rec}$, 
    histories avoiding cascading abort $\mathbb{H}_{aca}$, and
    serializable histories $\mathbb{H}_{ser}$.
}
\end{figure}

Opacity is strictly stronger than \lopacity{}.

\begin{lemma} \label{lemma:vis-legal-lvis-legal}
    For any history $S$  and transaction $\tr_i \in S$, if $\vis{S}{\tr_i}$ is
    legal, then $\luvis{S}{\tr_i}$ is legal.
\end{lemma}

\begin{proof}[sketch]
    By definition of $\vis{S}{\tr_i}$, if operation $\op \in \vis{S}{\tr_i}$,
    then $\op \in \vis{S}{\tr_i}$ only if $\op \in \hist|\tr_j$ and either
    $i=j$ or $\tr_j \prec_S \tr_i$ and $\tr_j$ is committed.
    By definition of $\luvis{S}{\tr_i}$, given transactions $\tr_i, \tr_j$ and
    operation $\op \in S|{\tr_j}$, if $i=j$ or $\tr_j \prec_S \tr_i$ and
    $\tr_j$ is committed, then $S|\tr_j \subseteq \luvis{S}{\tr_i}$. Therefore
    $\luvis{S}{\tr_i} \equiv \vis{S}{\tr_i}$. Since $\vis{S}{\tr_i}$ and
    $\luvis{S}{\tr_i}$ preserve the order of operations in $S$, then
    $\luvis{S}{\tr_i} = \vis{S}{\tr_i}$. 
    Hence, if  $\vis{S}{\tr_i}$ is legal, then $\luvis{S}{\tr_i}$ is legal.
\end{proof}

\begin{lemma} \label{lemma:all-fs-opaque-are-fs-lopaque}
    Any final-state last-use opaque history $\hist$ is final-state \lopaque{}.
\end{lemma}

\begin{proof}[sketch]
    From \rdef{def:fs-opacity}, for any final-state opaque history
    $\hist$, there is a sequential history $S \equiv \compl{\hist}$ s.t. $S$ preserves
    the real time order of $\hist$ and every transaction $\tr_i$ in $S$ is
    legal in $S$. Thus, for every transaction $\tr_i$ in $S$ $\vis{S}{\tr_i}$
    is legal.
    From the definition of completion, any $\tr_i$ is either committed or
    aborted in $\compl{\hist}$ and therefore likewise completed or aborted in
    $S$.
    If $\tr_i$ is committed in $S$, then it is legal in $S$, so
    $\vis{S}{\tr_i}$ is legal, and therefore $\tr_i$ is
    last-use legal in $S$.
    If $\tr_i$ is aborted in $S$, then it is legal in $S$, so
    $\vis{S}{\tr_i}$ is legal, and therefore, from
    \rlemma{lemma:vis-legal-lvis-legal}, $\luvis{S}{\tr_i}$ is also legal, so
    $\tr_i$ is last-use legal in $S$.
    Given that all transactions in $S$ are last-use legal in $S$, then,
    from \rdef{def:fs-lopacity}, $\hist$ is final-state last-use opaque. 
\end{proof}

\begin{lemma} \label{lemma:all-opaque-are-lopaque}
    Any opaque history $\hist$ is \lopaque{}.
\end{lemma}

\begin{proof}
    If $\hist$ is opaque, then, from \rdef{def:opacity}, any prefix $P$ of
    $\hist$ is final-state opaque. 
    Since any prefix $P$ of $\hist$ is final-state opaque, then, from
    \rlemma{lemma:all-fs-opaque-are-fs-lopaque}, any $P$ is also
    final-state last-use opaque. Then, by \rdef{def:lopacity} $\hist$ is
    last-use opaque.
\end{proof}

Last-use opacity is strictly stronger than serializability.

\begin{lemma} \label{lemma:comp-serializability}
    Any \lopaque{} history $\hist$ is serializable.
\end{lemma}

\begin{proof}
    For the sake of contradiction let us assume that $\hist$ is \lopaque{}
    and not serializable. 
    Since $\hist$ is \lopaque{}, then from \rdef{def:lopacity} $\hist$ is also
    final-state last-use opaque.  Then, from \rdef{def:fs-lopacity} there
    exists a completion $\hist_C=\compl{\hist}$ such that there is a sequential
    history $S$ s.t. $S \equiv \hist_C$, $S$ preserves the real-time order of
    $\hist_C$, and any committed transaction in $S$ is legal in $S$.
    However, since $\hist$ is not serializable, then from
    \rdef{def:serializability} there does not exist a completion
    $\hist_C=\compl{\hist}$ such that there is a sequential history $S$ s.t. $S
    \equiv \hist_C$, and any committed transaction in $S$ is legal in $S$. This
    contradicts the previous statement.
\end{proof}

\subsection{Last-use Opacity Variant for the Programmer Initiated Abort Model}
\label{sec:blu-opacity}

Even though last-use opacity prevents inconsistent views in the only forced
aborts and arbitrary aborts models, it does not prevent inconsistent views in
the programmer initiated aborts model. 
Hence, we present a variant of last-use opacity called $\beta$--last-use
opacity (\bluopacity{}) that extends the definition of a \last{} write
operation to take $\tryA$ operations into account, as if it was an operation
that modifies a given variable.

\begin{definition}[\Blast{} Write Invocation] \label{def:blast-write-inv} Given
a program $\prog$, a set of processes $\processes$ executing $\prog$ and a
history $\hist$ s.t. $\hist \models \exec{\prog}{\processes}$, i.e.  $\hist \in
\evalhist{\prog}{\processes}$, an invocation $\inv{i}{}{\wop{\obj}{\val}}$ is
the \emph{\last{} write invocation} on some variable $\obj$ by transaction
$\tr_i$ in $\hist$, if for any history $\hist' \in
\evalhist{\prog}{\processes}$ for which $\hist$ is a prefix (i.e., $\hist' =
\hist \cdot R$) there is no operation invocation $\mathit{inv}'$ s.t.
$\inv{i}{}{\wop{\obj}{\val}}$ precedes $\mathit{inv}'$ in
$\hist'|\tr_i$ where 
\begin{inparaenum}[(a) ]
    \item $\mathit{inv}' = \inv{i}{}{\wop{\obj}{\valu}}$ or
    \item $\mathit{inv}' = \inv{i}{}{\tryA}$.
\end{inparaenum}
\end{definition}

The remainder of the definitions of \bluopacity{} are formed by analogy to
their counterparts in last-use opacity. We only summarize them here and give
their full versions in the appendix.

The definition of a \blast{} write operation execution is analogous to that of
\last{} write operation execution \rdef{def:last-write-op}.
The \blast{} write is used instead of the \last{} write to define a transaction
\emph{\bdecided{} on $\obj$} in analogy to \rdef{def:decided}.
Then, that definition is used to define  $\cpetrans{\hist}_\beta$, $\hist\cpe^\beta\tr_j$, and
$\hist\cpeC^\beta\tr_j$ by analogy to $\cpetrans{\hist}$, $\hist\cpe\tr_j$ and $\hist\cpeC\tr_j$.
Next, those definitions are used to define $\bluvisf$ by analogy to $\luvisf$.
Finally, we say a transaction $\tr_i$ is \emph{$\beta$--last-use legal} in some
sequential history $S$ if $\bluvis{S}{\tr_i}$ is legal. 
This allows us to
define \bluopacity{} as follows.

\begin{definition} [Final-state $\beta$--Last-use Opacity] \label{def:final-state-blu-opacity} \label{def:fs-blu-opacity}
    A finite history $\hist$ is \emph{final-state $\beta$--last-use opaque} if, and only if,
    there exists a sequential history $S$ equivalent to any completion of
    $\hist$ s.t., 
    \begin{enumerate}[a) ] 
        \item $S$ preserves the real-time order of $\hist$,
        \item every transaction in $S$ that is committed in $S$ is
        legal in $S$,
        \item every transaction in $S$ that is not committed in $S$ is
        $\beta$--last-use legal in $S$. 
    \end{enumerate}
\end{definition}

\begin{definition} [$\beta$--Last-use Opacity] \label{def:blu-opacity}
    A history $\hist$ is \emph{$\beta$--last-use opaque} if, and only if, every finite prefix of
    $\hist$ is final-state $\beta$--last-use opaque.
\end{definition}

In this variant of last use opacity a transaction is not allowed to
release a variable early if it is possible that the transaction may execute
a voluntary abort.
In effect, \bluopacity{} precludes inconsistent views in the programmer
initiated abort model.

The \bluopacity{} property is trivially equivalent to last-use opacity in the
only forced abort model (because there are no $\tryA$ operations in that
model), but it is stronger than \lopacity{} in the arbitrary abort model to the
point of being overstrict.

The \bluopacity{} property is strictly stronger than last-use opacity in the
arbitrary abort model, but it is too strong to be applicable to systems with
early release.
In the first place, even though the histories that are excluded by \bluopacity{} 
contain inconsistent views, these are harmless, because as we argue
in \rsec{sec:implications}, transactions always release variables with ``final''
values. These final values cannot be reverted by a programmer-initiated abort,
so if the programmer sets up a \last{} write to a variable in a transaction,
the value that was written was expected to both remain unchanged and be
committed. Hence, it is acceptable for these values to be read by other
transactions, even before the original transaction commits.

Secondly, the arbitrary abort model 
specifies
that the TM
system can inject a $\tryA$ operation into the transactional code to respond to
some outside stimuli, such as crashes. Such events are unpredictable, so it may
be possible for any transaction to abort at any time. Hence, it is necessary to
assume that a $\tryA$ operation can be produced as the next operation
invocation in any transaction at any time.
In effect, as the definition of \bluopacity{} does not allow a transaction to
release a variable early if a $\tryA$ is possible in the future, \bluopacity{}
actually prevents early release altogether in the arbitrary abort model. 

In summary, \bluopacity{} is a useful variant of last-use opacity to exclude
inconsistent views in the programmer initiated abort model (if workarounds
suggested in \rsec{sec:implications} are insufficient solutions). However
\bluopacity{} is too strict for TMs operating in the arbitrary abort model,
where it prevents early release altogether. For that reason, last-use opacity
remains our focus.

\section{Supremum Versioning Algorithm}
\label{sec:lopaque-sva}

In this section we discuss the Supremum Versioning Algorithm (\SVA{}), a
pessimistic concurrency control algorithm with early release and rollback
support, and demonstrate that it satisfies \lopacity{}. 
\SVA{} in its current form was introduced in \cite{SW15-atomicrmi,SW14-hlpp},
although the presentation here gives a more complete description. It builds on
our rollback-free variant in \cite{WRS04,Woj05b,Woj07} as well as an earlier
version in \cite{SW13}.
Even though the current implementation of \SVA{} as part of Atomic RMI is
distributed, the concurrency control algorithm itself was created for
multiprocessor TMs and is applicable to both distributed as well as
multiprocessor systems. 

The main aspect of \SVA{} is the ability to release variables early.
The early release mechanism in \SVA{} is based on \emph{a priori} knowledge
about the maximum number of accesses that a transaction can perform on
particular variables. 
We explored various methods of obtaining satisfactory upper bounds, including
static analysis \cite{SW12} and static typing \cite{Woj05b}. 
A transaction that knows it performed exactly as many operations on some
variable as the upper bound allows may then release that variable. \SVA{} does
not require the upper bounds to be precise, and can handle situation when they
are either too great (some variables are not released early) or too low
(transactions are aborted).
Since \SVA{} does not distinguish between reads and writes when releasing, but
instead treats all accesses uniformly, we will extend the definition of \last{}
write and \last{} read operation executions to a \last{} access execution which
is such a read or write operation, that is not followed by any other \last{}
read or \last{} write operation.

\begin{figure}
\begin{minipage}[t]{.52\linewidth}
\begin{lstlisting}
procedure start(Transaction $\tr_i$) {     $\label{sva:start}$
  for ($\obj$ : $\accesses{\tr_i}$ sorted by $\prec_{\glocks}$) $\label{sva:start-lock}$
    lock $\glocks(\obj)$ $\label{sva:lock-x}$
  for ($\obj$ : $\accesses{\tr_i}$ in parallel) { $\label{sva:start-main}$
    $\gv{\obj}$ ← $\gv{\obj}$ + 1 $\label{sva:start-set-gv}$
    $\pv{i}{\obj}$ ← $\gv{\obj}$ $\label{sva:start-set-pv}$ $\label{sva:start-main-end}$
  }
  for ($\obj$ : $\accesses{\tr_i}$ sorted by $\prec_{\glocks}$)
    unlock $\glocks(\obj)$ $\label{sva:start-end-unlock}$ $\label{sva:unlock-x}$
}

procedure access(Transaction $\tr_i$, Var $\obj$) { $\label{sva:access}$
  wait until $\pv{i}{\obj}$ - 1 = $\lv{\obj}$ $\label{sva:call-access-cond}$
  checkpoint($\tr_i$, $\obj$)
  if ($\rv{i}{\obj}$ ≠ $\cv{\obj}$) $\label{sva:call-rollback-cond}$
    abort($\tr_i$) and exit $\label{sva:call-rollback}$
  execute read or write $\label{sva:call-access}$
  $\cc{i}{\obj}$ ← $\cc{i}{\obj}$ + 1$\label{sva:call-cc-set}$ $\label{sva:call-inc-cc}$
  if ($\cc{i}{\obj}$ = $\supr{i}{\obj}$) {  $\label{sva:early-release-cond}$
    $\cv{\obj}$ ← $\pv{i}{\obj}$ $\label{sva:early-release-cv}$
    $\lv{\obj}$ ← $\pv{i}{\obj}$ $\label{sva:early-release-lv}$ $\label{sva:early-release-end}$
  }
}

procedure abort(Transaction $\tr_i$) { $\label{sva:rollback}$
  for ($\obj$ : $\accesses{\tr_i}$ in parallel) {
    wait until $\pv{i}{\obj}$ - 1 = $\ltv{\obj}$ $\label{sva:abort-dismiss-access-cond}$
    dismiss($\tr_i$, $\obj$) $\label{sva:abort-dismiss}$
    restore($\tr_i$, $\obj$)
    delete $\stored{i}{\obj}$ $\label{sva:restore-clear}$
    $\ltv{\obj}$ ← $\pv{i}{\obj}$ $\label{sva:restore-ltv-set}$
  }
}
\end{lstlisting}
\end{minipage}
\begin{minipage}[t]{0.48\linewidth}
\begin{lstlisting}[firstnumber=34]
procedure commit(Transaction $\tr_i$) { $\label{sva:commit}$
  for ($\obj$ : $\accesses{\tr_i}$ in parallel) {
    wait until $\pv{i}{\obj}$ - 1 = $\ltv{\obj}$ $\label{sva:commit-dismiss-access-cond}$
    dismiss($\tr_i$, $\obj$)
  }
  if (∃$\obj$ in $\accesses{\tr_i}$ such that $\rv{i}{\obj}$ > $\cv{\obj}$) $\label{sva:after-release}$ $\label{sva:commit-rollback-cond}$    
    abort($\tr_i$) and exit $\label{sva:commit-rollback}$
  for ($\obj$ : $\accesses{\tr_i}$ in parallel) {
    delete $\stored{i}{\obj}$ $\label{sva:commit-clear}$
    $\ltv{\obj}$ ← $\pv{i}{\obj}$         $\label{sva:commit-ltv-set}$
  }
}

procedure checkpoint(Transaction $\tr_i$, Var $\obj$) { $\label{sva:checkpoint}$
  if ($\cc{i}{\obj}$ = 0) { $\label{sva:checkpoint-cond}$
    copy $\obj$ to $\stored{i}{\obj}$ $\label{sva:checkpoint-store}$
    $\rv{i}{\obj}$ ← $\cv{\obj}$  $\label{sva:checkpoint-rv-set}$
  }
}

procedure dismiss(Transaction $\tr_i$, Var $\obj$) { $\label{sva:dismiss}$
  if ($\cc{i}{\obj}$ = 0 and $\rv{i}{\obj}$ = $\cv{\obj}$) $\label{sva:release-cond}$
    $\cv{\obj}$ ← $\pv{i}{\obj}$ $\label{sva:release-cv-set}$
  if ($\pv{i}{\obj}$ - 1 = $\lv{\obj}$) $\label{sva:dismiss-cond}$
    $\lv{\obj}$ ← $\pv{i}{\obj}$ $\label{sva:dismiss-release}$ $\label{sva:dismiss-end}$
} 

procedure restore(Transaction $\tr_i$, Var $\obj$) { $\label{sva:restore}$
  if ($\cc{i}{\obj}$ ≠ 0 and $\rv{i}{\obj}$ < $\cv{\obj}$) { $\label{sva:restore-cond}$
    revert $\obj$ from $\stored{i}{\obj}$ $\label{sva:restore-revert}$
    $\cv{\obj}$ ← $\rv{i}{\obj}$ $\label{sva:restore-cv-set}$
  }
}
    \end{lstlisting}
    \end{minipage}
    \caption{\label{fig:sva-pseudocode} \SVA{} pseudocode.}
\end{figure}

\subsection{Counters}

The basic premise of versioning algorithms is that counters are associated with
transactions and used to allow or deny access by these transactions to shared
objects (rather than only for recovery). 
\SVA{} uses several version counters. The \emph{private version} counter
$\pv{i}{\obj}$ uniquely defines the version of transaction $\tr_i$ with respect
to variable $\obj$.  The \emph{global version} counter $\gv{\obj}$ shows how
many transactions that have $\obj$ in their access set have started. 
The \emph{local version} counter $\lv{\obj}$ shows which transaction can
currently and exclusively access variable $\obj$.
Specifically, the transaction that can access $\obj$ is such $\tr_i$ whose
$\pv{i}{\obj}$ is one greater than $\lv{\obj}$ (we refer to this as the
\emph{access condition}).
The \emph{local terminal version} counter $\ltv{\obj}$ shows which of the
transactions that have $\obj$ in their access set can currently commit or
abort.  That is, $\tr_i$ such that $\obj \in \accesses{\tr_i}$ can commit or
abort if its $\pv{i}{\obj}$ is one greater than $\ltv{\obj}$.

In addition, the \emph{current version} counter $\cv{\obj}$ defines what state
the variable is in and transactions that operate on $\obj$ will increment its
$\cv{\obj}$ to indicate a change of state. It will also revert the counter to
indicate a rollback (abort). 
A \emph{recovery version} counter $\rv{i}{\obj}$
indicates what state variable $\obj$ was in prior to transaction $\tr_i$'s
modifications. These counters are used together to detect whether a variable
accessed by the current transaction was rolled back by some other
transaction, requiring that the current transaction roll back as well.
These counters also determine which
transaction is responsible for reverting the state of the variable and to what
state it should be reverted.
In order be able to revert the state of variable \SVA{} also
uses a \emph{variable store} $\storedf$, where transactions store copies of
variables before modifying them.

In order to detect the \last{} use of a variable, \SVA{} requires that suprema
on accesses be given for each variable used by a transaction. These are given
for each variable $\obj$ in transaction $\tr_i$'s access set as
    $\supr{i}{\obj}$. If the supremum is unknown $\supr{i}{\obj} = \infty$.
We assume that if a
supremum on accessing some variable $\obj$ is zero, then $\obj$ is excluded
from the transaction's access set. %
Then, \emph{access counter} $\cc{i}{obj}$ is used to track the actual number of
accesses by $\tr_i$ on $\obj$ and to check when the supremum is reached, to
release a variable early.

Finally, \SVA{} uses a map of locks $\glocks$, containing one lock for each
variable to make a globally consistent snapshots.  Locks are acquired and
released in order $\prec_{\glocks}$ to prevent deadlocks.
Initially, all locks are unlocked, counters are set to zero, and the variable
store is empty.

\subsection{Transactions}

The pseudocode for \SVA{} is shown in \rfig{fig:sva-pseudocode}. 
The life cycle of every \SVA{} transaction begins with procedure
$\ppprocedure{start}$ (we also refer to this part as initialization). Following
that, a transaction may execute one or more accesses (reads or writes) to
shared variables (procedure $\ppprocedure{access}$). After any
$\ppprocedure{access}$ or right after $\ppprocedure{start}$ a transaction may
then either proceed to $\ppprocedure{commit}$ or 
$\ppprocedure{abort}$, both
of which end a transaction's life cycle. \SVA{} transactions are prevented from
committing until all preceding transactions which released their variables
early and with which the current transaction shares variables also commit.
Accesses to shared variables can be interleaved with various
transaction-local operations, including accesses to non-shared structures or
variables. However, those are only visible to the transaction to which they are
local, so they do not influence other transactions. For the purpose of
clarity, but without loss of generality, we omit those operations here.
We also assume that transactions are executed in a single, fresh, dedicated
thread. We also omit the concepts of nested and recurrent transactions.

\paragraph{Start}
The initialization of a transaction is shown in $\ppprocedure{start}$ at
\pplineref{sva}{start}. When transaction $\tr_i$ starts it uses $\gv{\obj}$ to
assign itself a unique version $\pv{i}{\obj}$ for each variable $\obj$ in its
access set $\accesses{\tr_i}$. This must be done atomically and in isolation,
so these operations are guarded by locks---one lock $\glocks(\obj)$ for each
variable used.

\paragraph{Accesses}
Variables are accessed via procedure $\ppprocedure{access}$ at
\pplineref{sva}{access}. Before accessing a variable, transaction $\tr_i$ waits
for the access condition to be satisfied at \pplineref{sva}{call-access-cond}
(e.g. for the preceding transaction to release it). When this happens, $\tr_i$
makes a backup copy using $\ppprocedure{checkpoint}$
(\pplineref{sva}{checkpoint}). This procedure checks whether this is the first
access, and if so makes  a backup copy of $\obj$ to $\stored{i}{\obj}$ at
\pplineref{sva}{checkpoint-store} and sets the recovery counter to the current
version of $\obj$ at \pplineref{sva}{checkpoint-rv-set} to indicate which
version of $\obj$ the transaction first viewed and can revert to in case of
rollback (abort). Then, the access (either a write or read) is actually performed.

Afterward, transaction $\tr_i$ increments the access counter $\cc{i}{\obj}$ 
(\pplineref{sva}{call-inc-cc}) and proceeds to check whether this was the \last{} access on
$\obj$ by comparing $\cc{i}{\obj}$ to the appropriate
supremum $\supr{i}{\obj}$ (\pplineref{sva}{early-release-cond}).
If this is the
case, the variable is released early---i.e, $\lv{\obj}$ is set to the same value as
the transaction's private counter $\pv{i}{\obj}$
(\pplineref{sva}{early-release-lv}). At this point, some other transaction
$\tr_j$, whose $\pv{j}{\obj} - 1= \lv{\obj}$ can start accessing the variable
using procedure $\ppprocedure{access}$. Another provision made in \SVA{} is
that when variable $\obj$ is released by transaction $\tr_i$, $\cv{\obj}$ is
set to the transaction's $\pv{i}{\obj}$ version  (during early release at
\pplineref{sva}{early-release-cv}). This signifies both that there is a new
consistent version of $\obj$ and that $\tr_i$ modified $\obj$ the most
recently.

\paragraph{Commit}
Transaction $\tr_i$ can attempt to commit using procedure
$\ppprocedure{commit}$ at \pplineref{sva}{commit}. The variables in the
transaction's access set are committed independently (possibly in parallel).
First, transaction $\tr_i$ must pass the commit condition at
\pplineref{sva}{commit-dismiss-access-cond} for each variable in its access
set, so that a transaction is not allowed to commit before all the transactions
that accessed the same variables as $\tr_i$ before $\tr_i$ commit or abort.
The transaction then checks whether it has access to the variable and if this
is not the case, i.e., the variable was not accessed at all, waits until the
preceding transaction releases it at \pplineref{sva}{commit-ltv-set}. Then, 
the committing transaction executes procedure $\ppprocedure{dismiss}$ at
\pplineref{sva}{dismiss}. 
At this point, if the transaction did not access
the variable before commit, the current version counter is updated
(\pplineref{sva}{release-cv-set}).  Furthermore, if this transaction has not
previously released some variable $\obj$, $\lv{\obj}$ is set to $\pv{i}{\obj}$
to indicate the object is released (\pplineref{sva}{dismiss-release}).
Finally, the transaction erases backup copies from the store $\storedf$
(\pplineref{sva}{commit-clear}) and sets $\ltv{\obj}$ to its private version
counter's value (\pplineref{sva}{commit-ltv-set}) to indicate that a subsequent
transaction can now perform \ppprocedure{commit} or \ppprocedure{abort} on
$\obj$.

\paragraph{Abort}
Aborts are not performed by \SVA{} as part of its basic \emph{modus operandi},
but aborts can be triggered manually by the programmer. Furthermore, such
manually-triggered aborts can also cause other transactions to abort.
If the programmer decides to abort a transaction, or if an abort is forced,
procedure $\ppprocedure{abort}$ (\pplineref{sva}{rollback}) is used. 
As with commit, in on order for the abort to proceed, the transaction must pass
the commit condition at \pplineref{sva}{abort-dismiss-access-cond} for each
variable in its access set.
Then, previously unreleased variables are released using
$\ppprocedure{dismiss}$ (as described above). Afterward, the transaction
restores its variables using procedure $\ppprocedure{restore}$
(\pplineref{sva}{restore}). 
There, if $\tr_i$ accessed $\obj$ at least once and the version of $\obj$ that
it stored prior to accessing $\obj$ ($\rv{i}{\obj}$) is lower than the current
version of $\obj$ ($\cv{\obj}$), then $\tr_i$ is responsible for reverting
$\obj$ to the previous state (condition at \pplineref{sva}{restore-cond}).  In
that case, this procedure reverts $\obj$ to a copy from $\stored{i}{\obj}$
(\pplineref{sva}{restore-revert}). It also sets the current version to
$\rv{i}{\obj}$ (\pplineref{sva}{restore-cv-set}), indicating that variable
$\obj$ was restored to the state just from before $\tr_i$ modified it.  Note
that this means that $\cv{\obj}$ was set to a lower value than before.
Finally, the transaction cleans up $\stored{i}{\obj}$
(\pplineref{sva}{restore-clear}) and sets $\ltv{\obj}$ to its private version
counter's value (\pplineref{sva}{restore-ltv-set}). As with
$\ppprocedure{commit}$, this procedure may operate on each variable in
parallel.

If a manual abort is triggered by the user it may be the case that transaction
$\tr_i$ releases variable $\obj$ and aborts after some other transaction
$\tr_j$ already reads the value written to $\obj$ by $\tr_i$. In that case
$\tr_j$ must be forced to abort along with $\tr_i$.  
Hence, if subsequently the
situation arises that some other transaction $\tr_j$ tries to access $\obj$
after $\tr_i$ released it early and aborted, there is a condition at
\pplineref{sva}{call-rollback-cond} in procedure $\ppprocedure{access}$ which
will compare the values of $\cv{\obj}$ and $\rv{i}{\obj}$ for $\obj$. If
$\cv{\obj}$ and $\rv{j}{\obj}$ are equal, that implies $\tr_j$ currently has
access to a consistent state of $\obj$ and can access it.
However, if $\rv{j}{\obj}$ is greater than $\cv{\obj}$, this means that some
previous transaction (i.e., $\tr_i$) set $\cv{\obj}$ to a lower value than it was
before, which means it aborted and reverted $\obj$. 
This causes $\tr_j$ to be forced to abort (\pplineref{sva}{call-rollback})
instead of accessing $\obj$.
As an aside, if $\cv{\obj}$ is greater than $\rv{j}{\obj}$, then $\tr_j$ already
stopped modifying $\obj$ and released it, so there should not be any more
accesses to $\obj$. This means that $\tr_j$ violated its upper bound for $\obj$
and must also be forced to abort.

Similarly, if transaction $\tr_j$ tries to commit after $\tr_i$ aborted, then
it checks the condition at \pplineref{sva}{commit-rollback-cond} for each
variable, and if there is at least one variable $\obj$ for which $\rv{j}{\obj}$
is greater than $\cv{\obj}$, then again some previous $\tr_i$ must have aborted
and reverted $\obj$ and $\tr_j$ must then abort too.
Note that aborted transactions revert variables in the order imposed by 
the commit condition in wait (\pplineref{sva}{abort-dismiss-access-cond}), 
which ensures state consistency after rollback.

\subsection{Examples}

To illustrate further, the examples in
\rfig{fig:sva-sem-early}--\ref{fig:sva-sem-parallel} show some scenarios of
interacting \SVA{} transactions. 
In \rfig{fig:sva-sem-early}, transactions $\tr_i$ and $\tr_j$ both try to
increment variable $\obj$. Since $\tr_i$ starts before $\tr_j$, it has a lower
version for $\obj$ (e.g., $\mathtt{pv}_i(\obj)\!\gets\!1$) than $\tr_j$ (e.g.,
$\mathtt{pv}_j(\obj)\!\gets\!2$).  So, when accessing $\obj$, $\tr_i$ will pass
its access condition for $\obj$ sooner than $\tr_j$. Thus, $\tr_j$ has to wait,
while $\tr_i$ executes its operations on $\obj$. After its last operation on
$\obj$, $\tr_i$ releases it by setting $\obj$'s local counter to its own
version ($\mathtt{lv}(\obj)\!\gets\!1$). From this moment $\tr_1$ can no
longer access $\obj$. But since the local counter is now equal to $1$, which
is one lower than $\tr_j$'s version for $\obj$ of $2$, $\tr_j$ can now pass
the access condition for $\obj$ and start executing operations. In effect
$\tr_i$ and $\tr_j$ can execute in parallel in part. 
\rfig{fig:sva-sem-commit} shows a similar example of a transaction $\tr_i$
releasing $\objx$ early and operating on $\objy$ while a later transaction
$\tr_j$ operates on $\objx$ in parallel. The scenario differs from the last in
that $\tr_j$ finishes work before $\tr_i$, however it still waits with
committing until $\tr_i$ commits.
\rfig{fig:sva-sem-cascade} shows another similar scenario, except $\tr_i$
eventually aborts. Then, $\tr_j$ is also forced to abort (and retry) by \SVA{}.
Note that if meanwhile $\tr_j$ released $\obj$ early and some transaction
$\tr_k$ accessed it, $\tr_k$ would also be aborted. Thus, a cascading
abort including multiple transactions is a possibility.
Finally, \rfig{fig:sva-sem-parallel} shows two transactions operating on different
variables. Since \SVA{} performs synchronization per variable, if the access
sets of two transactions do not intersect, they can both execute in parallel,
without any synchronization.

\begin{figure}
\begin{center}
\begin{tikzpicture}
     \draw
           (0,2)        node[tid]       {$\tr_i$}
                        node[aop]       {$\init_i$} %
                        node[dot]       {} 
                        node[not]       {$\mathtt{pv}_i(x)\!\gets\!1$}

      -- ++(1.25,0)     node[aop]       {$\trop{i}{\objx}{0}$}
                        node[dot]       {}

      -- ++(1.25,0)     node[aop]       {$\twop{i}{\objx}{1}$}
                        node[dot] (wx)  {}
                        node[cir]       {}
                        node[not]       {$\mathtt{lv}(x)\!\gets\!1$}

      -- ++(1.25,0)     node[aop]       {$\twop{i}{\objy}{1}$}
                        node[dot] (wy)  {}

      -- ++(1.25,0)     node[aop]       {$\tryC_i\!\to\!\co_i$}
                        node[dot]       {}         
                        ;

     \draw
           (0.75,1)     node[tid]       {$\tr_{j}$}
                        node[aop]       {$\init_{j}$} %
                        node[dot]       {} 
                        node[not]       {$\mathtt{pv}_j(x)\!\gets\!2$}

      -- ++(1.25,0)     node[aop]       {$\trop{j}{\obj}{}$}
                        node[inv] (rx0) {}
                        node[not]       {$\mathtt{lv}(x)\!=\!1?$}
 
         ++(1.25,0)     node[aop]       {$\!\to\!1$}
                        node[res] (rx)  {}                       

      -- ++(1.25,0)     node[aop]       {$\twop{j}{\objx}{2}$}
                        node[dot]       {}

      -- ++(1.25,0)     node[aop]       {$\tryC_j\!\to\!\co_j$}
                        node[dot]       {}
                        ;
     
     \draw[hb] (wx) \squiggle (rx);
     \draw[wait] (rx0) -- (rx);
\end{tikzpicture}
\end{center}
\caption{\label{fig:sva-sem-early}\SVA{} early release example.
The symbol 
\protect\tikz{
    \protect\draw[] (0,0) -- ++(0.1,0) node[inv] (a) {} 
                             ++(0.4,0) node[res] (b) {}
                          -- ++(0.1,0);
    \protect\draw[wait] (a) -- (b);
}
denotes an operation execution split into the invocation event and the response
event to indicate waiting and 
\protect\tikz{
    \protect\draw[hb] (0,0.2) .. controls +(270:.25) and +(90:0.25) .. (0.5,0.0);
} 
indicates what caused the wait.
}
\end{figure}

\begin{figure}
\begin{center}
\begin{tikzpicture}
     \draw
           (0,2)        node[tid]       {$\tr_i$}
                        node[aop]       {$\init_i$} %
                        node[dot]       {} 
                        node[not]       {$\mathtt{pv}_i(x)\!\gets\!1$}

      -- ++(1.25,0)     node[aop]       {$\twop{i}{\objx}{1}$}
                        node[dot] (wx)  {}
                        node[cir]       {}

      -- ++(1.25,0)     node[aop]       {$\trop{i}{\objy}{0}$}
                        node[dot]       {}

      -- ++(1.25,0)     node[aop]       {$\twop{i}{\objy}{1}$}
                        node[dot] (wy)  {}
                       
      -- ++(1.25,0)     node[aop]       {$\tryC_i\!\to\!\co_i$}
                        node[dot] (ci)  {}         
                        node[not]       {$\mathtt{ltv}(x)\!\gets\!1$}
                        ;

     \draw
           (0.75,1)     node[tid]       {$\tr_{j}$}
                        node[aop]       {$\init_{j}$} %
                        node[dot]       {} 
                        node[not]       {$\mathtt{pv}_j(x)\!\gets\!2$}

      -- ++(1.25,0)     node[aop]       {$\trop{j}{\obj}{1}$}
                        node[dot] (rx)  {}

      -- ++(1.25,0)     node[aop]       {$\twop{j}{\objx}{2}$}
                        node[dot]       {}

      -- ++(1.25,0)     node[aop]       {$\tryC_j$}
                        node[inv] (cj0) {}
                        node[not]       {$\mathtt{ltv}(x)\!=\!1?$}
      
         ++(1.25,0)     node[aop]       {$\!\to\!\co_j$}
                        node[res] (cj)  {}
                        ;
     
     \draw[hb] (wx) \squiggle (rx);
     \draw[hb] (ci) \squiggle (cj);
     \draw[wait] (cj0) -- (cj);
\end{tikzpicture}
\end{center}
\caption{\label{fig:sva-sem-commit}\SVA{} wait on commit example.}
\end{figure}
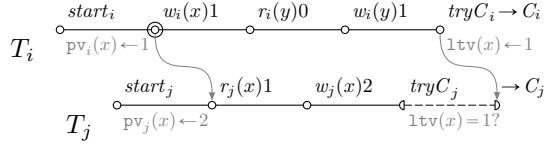

\begin{figure}
\begin{center}
\begin{tikzpicture}
     \draw
           (0,2)        node[tid]       {$\tr_i$}
                        node[aop]       {$\init_i$} %
                        node[dot]       {} 
                        node[not]       {$\mathtt{pv}_i(x)\!\gets\!1$}

      -- ++(1.25,0)     node[aop]       {$\twop{i}{\objx}{1}$}
                        node[dot] (wx)  {}
                        node[cir]       {}

      -- ++(1.25,0)     node[aop]       {$\trop{i}{\objy}{0}$}
                        node[dot]       {}

      -- ++(1.25,0)     node[aop]       {$\twop{i}{\objy}{1}$}
                        node[dot] (wy)  {}
                       
      -- ++(1.25,0)     node[aop]       {$\tryA_i\!\to\!\ab_i$}
                        node[dot] (ai)  {}         
                        node[not]       {$\mathtt{ltv}(x)\!\gets\!1$}
                        ;

     \draw
           (0.75,1)     node[tid]       {$\tr_{j}$}
                        node[aop]       {$\init_{j}$} %
                        node[dot]       {} 
                        node[not]       {$\mathtt{pv}_j(x)\!\gets\!2$}

      -- ++(1.25,0)     node[aop]       {$\trop{j}{\obj}{1}$}
                        node[dot] (rx)  {}

      -- ++(1.25,0)     node[aop]       {$\twop{j}{\objx}{2}$}
                        node[dot]       {}

      -- ++(1.25,0)     node[aop]       {$\tryC_j$}
                        node[inv] (aj0) {}
                        node[not]       {$\mathtt{ltv}(x)\!=\!1?$}
      
         ++(1.25,0)     node[aop]       {$\!\to\!\ab_j$}
                        node[res] (aj)  {}
                        ;
     
     \draw[hb] (wx) \squiggle (rx);
     \draw[hb] (ai) \squiggle (aj);
     \draw[wait] (aj0) -- (aj);
\end{tikzpicture}
\end{center}
\caption{\label{fig:sva-sem-cascade}\SVA{} cascading abort example.}
\end{figure}
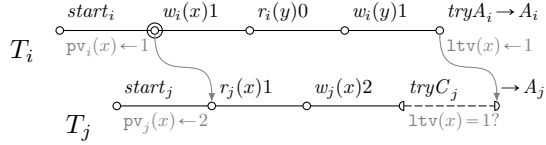

\begin{figure}
\begin{center}
\begin{tikzpicture}
     \draw
           (0,2)        node[tid]       {$\tr_i$}
                        node[aop]       {$\init_i$} %
                        node[dot]       {} 
                        node[not]       {$\mathtt{pv}_i(x)\!\gets\!1$}

      -- ++(1.25,0)     node[aop]       {$\trop{i}{\objx}{0}$}
                        node[dot] (rx)  {}
                        node[not]       {$\mathtt{lv}(x)\!=\!0?$}

      -- ++(1.25,0)     node[aop]       {$\twop{i}{\objx}{1}$}
                        node[dot] (wx)  {}
                       
      -- ++(1.25,0)     node[aop]       {$\tryC_i\!\to\!\co_i$}
                        node[dot]       {}         
                        ;

     \draw
           (0.2,1)      node[tid]       {$\tr_{j}$}
                        node[aop]       {$\init_{j}$} %
                        node[dot]       {} 
                        node[not]       {$\mathtt{pv}_i(y)\!\gets\!1$}

      -- ++(1.25,0)     node[aop]       {$\trop{j}{\objy}{0}$}
                        node[dot] (ry)  {}
                        node[not]       {$\mathtt{lv}(y)\!=\!0?$}

      -- ++(1.25,0)     node[aop]       {$\twop{j}{\objy}{1}$}
                        node[dot] (wy)  {}

      -- ++(1.25,0)     node[aop]       {$\tryC_j\!\to\!\co_j$}
                        node[dot]       {}   
                        ;

\end{tikzpicture}
\end{center}
\caption{\label{fig:sva-sem-parallel}\SVA{} parallel execution example.}
\end{figure}
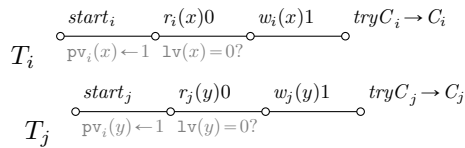

\subsection{Safety of \SVA{}}
\label{sec:lopaque-sva-safety-short}

We present a proof sketch showing that \SVA{} is last-use opaque. A complete
proof is in \rappx{sec:sva-proof}.
First, we make the following straightforward observations about \SVA{}.

\begin{observation}[Version Order]\label{e:version-order}
        Given the set $\transactions_\hist^\obj$ of all transactions that
        access $\obj$ in $\hist$ there is a total order called a version order
        $\prec_\obj$ on $\transactions_\hist^\obj$ s.t. for any $\tr_i, \tr_j
        \in \transactions_\hist^\obj$, $\tr_i \prec_\obj \tr_j$ if
        $\pv{i}{\obj} < \pv{j}{\obj}$.  
\end{observation}

\begin{observation}[Access Order]\label{e:access-order}
        If $\tr_i \prec_\obj \tr_j$ and $\tr_i$ performs operation
        $\op_i$ on $\obj$, and $\tr_j$ performs operation $\op_j$ on $\obj$,
        then $\op_i$ is completed in $\hist$ before $\op_j$.
\end{observation}

\begin{observation}[No Bufferring] \label{e:no-buffer} 
        Since transactions operate on variables rather than buffers, any
        read operation $\op = \frop{i}{\obj}{\val}$ in any transaction $\tr_i$
        is preceded in $\hist$ by some write operation $\fwop{j}{\obj}{\val}{\ok_j}$
        in some $\tr_j$ (possibly $i=j$).
\end{observation}

\begin{observation}[Read from Released]\label{e:read-released}
        If transaction $\tr_i$ executes a read operation or a write
        operation $\op$ on $\obj$ in $\hist$, then any transaction that
        previously executed a read or write operation on $\obj$ is either
        committed, aborted, or decided on $\obj$ before $\op$.
\end{observation}

\begin{observation}[Do not Read Aborted]\label{e:read-written}
        Assuming unique writes, if transaction $\tr_i$ executes
        $\fwop{i}{\obj}{\val}{\valu}$ and aborts in $\hist$, then $\obj$ will
        be reverted to a previous value. In consequence, no other transaction
        can read $\val$ from $\obj$.
\end{observation}

\begin{observation}[Commit Order]\label{e:commit-order}
        If transaction $\tr_i$ accesses $\obj$ in $\hist$ and commits or
        aborts in $\hist$, any transaction that
        previously executed a read or write operation on $\obj$ is either
        committed or aborted before $\tr_i$ commits or aborts.
\end{observation}

\begin{observation}[Forced Abort]\label{e:forced-abort}
        If transaction $\tr_i$ reads $\obj$ from $\tr_j$ and $\tr_j$
        subsequently aborts, then $\tr_i$ also aborts.
\end{observation}

Then, the main lemma follows, showing that \SVA{} produces final-state opaque
histories.
For convenience, we assume that the \SVA{} program always writes values to
variables that are unique and in the domain of the variable.

\begin{lemma} \label{lemma:sva-fs-lopacity-short}
    Any \SVA{} history $\hist$ is final-state last-use opaque.
\end{lemma}

\begin{proof}[sketch]
    Let $\hist_C = \compl{\hist}$ be a completion of $\hist$ if for
    every $\tr_i \in \hist$, if $\tr_i$ is live or commit-pending in $\hist$,
    then $\tr_i$ is aborted in $\hist_C$.
    Given $\hist_C$ we can construct $\seqH$, a sequential history
    s.t. $\seqH \equiv \hist_C$, where for any two transactions
    $\tr_i, \tr_j \in \hist_C$:
    \begin{enumerate}[a) ]
        \item if $\tr_i \prec_{\hist_C} \tr_j$, then $\tr_i \prec_{\seqH} \tr_j$,
        \item if $\tr_i \prec_{\obj}$ $\tr_j$ for any variable $\obj$, then
        $\tr_i \prec_{\seqH} \tr_j$.
    \end{enumerate}
    Note that if some transaction $\tr_i$ commits in $\hist$, then it commits
    in $\seqH$ (and \emph{vice versa}). Otherwise $\tr_i$ aborts in $\seqH$.

    Let $\tr_i$ be any transaction committed in $\hist$. Thus, $\tr_i$ also
    commits in $\seqH$. 
    From \robs{e:no-buffer}, any read operation execution $\op_i = \frop{i}{\obj}{\val}$ in
    $\hist|\tr_i$ is preceded in $\hist$  by $\op_j = \fwop{j}{\obj}{\val}{\ok_j}$. 
    If $\op_i$ is local, then $i=j$, so $\op_j$ is in a committed transaction.
    If $\op_i$ is not local, then $i\neq j$. In that case, from
    \robs{e:read-written}, $\tr_j$ cannot be aborted before $\op_i$ in $\hist$.
    Consequently, $\tr_j$ is either committed before $\op_i$ in $\hist$, live in $\hist$,
    or committed or aborted after $\op_i$. 
    In the former case $\tr_i$ reads from a committed transaction.
    In the latter case, since $\tr_i$ is committed, then from
    \robs{e:read-released} and \robs{e:commit-order} we know that $\tr_j$ commits
    or aborts in $\hist$ before $\tr_i$ commits.
    In addition, from \robs{e:forced-abort} we know that $\tr_j$ cannot abort in
    $\hist$, because it would have caused $\tr_i$ to also abort. 
    Thus, any committed $\tr_i$ reads only from committed transactions.

    From \robs{e:access-order}, if $\tr_i$ reads from the value written by an
    operation in $\tr_j$ then the write in $\tr_j$ completes before the read in
    $\tr_i$, which implies $\tr_j \prec_\obj \tr_i$. Hence, $\tr_j \prec_{\seqH}
    \tr_i$.
    Thus, if $\tr_i$ is committed in $\seqH$ and reads from some $\tr_j$, then 
    any such $\tr_j$ is committed and precedes $\tr_i$, so $\seqH|\tr_j
    \subseteq \vis{\seqH}{\tr_i}$. 
    Since all reads in committed transactions read from preceding committed
    transactions, then for each read in $\vis{\seqH}{\tr_i}$ reading $\val$
    from $\obj$ there will be a write operation execution writing $\val$ to
    $\obj$ in $\vis{\seqH}{\tr_i}$. 
    Since, from \robs{e:access-order} all accesses on $\obj$ operations
    follow $\prec_\obj$, then $\vis{\seqH}{\tr_i}$ is legal for any
    committed $\tr_i$.
    Thus, any $\tr_i$ that is committed in $\seqH$ is legal in $\seqH$. 

    Let $\tr_i$ be a transaction that is live or aborts in $\hist$, so it
    aborts in $\seqH$.
    From \robs{e:no-buffer} any read operation execution $\op_i =
    \frop{i}{\obj}{\val}$ in $\hist|\tr_i$ is preceded in $\hist$  by $\op_j =
    \fwop{j}{\obj}{\val}{\ok_j}$. 
    If $\op_i$ is local, then $i=j$, so $\op_j$ is always in
    $\vis{\seqH}{\tr_i}$ where $\op_j$ precedes $\op_i$.
    If $\op_i$ is not local, then $i\neq j$. In that case, from
    \robs{e:read-written}, $\tr_j$ cannot be aborted before $\op_i$ in $\hist$.
    Consequently, $\tr_j$ is either committed before $\op_i$ in $\hist$, live in $\hist$,
    or committed or aborted after $\op_i$. 
    In the former case $\tr_i$ reads from a committed transaction.
    In the latter case, from
    \robs{e:read-released} we know that either $\tr_j$ commits
    in $\hist$ or $\tr_j$ is decided on $\obj$ in $\hist$.
    Thus, any committed $\tr_i$ reads $\obj$ only from committed transactions
    or transactions that are decided on $\obj$.

    From \robs{e:access-order}, if $\tr_i$ reads from the value written by an
    operation in $\tr_j$ then the write in $\tr_j$ completes before the read in
    $\tr_i$, which implies $\tr_j \prec_\obj \tr_i$. Hence, $\tr_j \prec_{\seqH}
    \tr_i$.
    Thus, if $\tr_i$ is aborted in $\seqH$ and reads from some $\tr_j$, then
    any such $\tr_j$ is either committed and precedes $\tr_i$, so $\seqH|\tr_j
    \subseteq \luvis{\seqH}{\tr_i}$, or $\tr_j$ is decided on any $\obj$ if
    $\tr_i$ reads from $\obj$, so $\seqH\cpeC\tr_j \subseteq
    \luvis{\seqH}{\tr_i}$.
    Since all reads in aborted transactions read $\obj$ from preceding
    committed transactions or transactions decided on $\obj$, then for each
    read in $\luvis{\seqH}{\tr_i}$ reading $\val$ from $\obj$ there will be a
    write operation execution writing $\val$ to $\obj$. Since, from
    \robs{e:access-order} all accesses on $\obj$ operations follow
    $\prec_\obj$, then $\luvis{\seqH}{\tr_i}$ is legal for any aborted
    $\tr_i$.
    Thus, any $\tr_i$ that is aborted in $\seqH$ is last-use legal in $\seqH$. 

    Since any committed $\tr_i$ in $\seqH$ is legal in $\seqH$, and any
    aborted $\tr_i$ in $\seqH$ is last-use legal in $\seqH$, and since 
    $\seqH$ trivially follows the real time order of $\hist$, then from
    \rdef{def:fs-lopacity} $\hist$ is final-state last-use opaque.
\end{proof}

\begin{theorem} \label{thm:sva-lopacity-short}
    Any \SVA{} history $\hist$ is last-use opaque.
\end{theorem}

\begin{proof}
    Since by \rlemma{lemma:sva-fs-lopacity-short} any \SVA{} history $\hist$ is
    final-state last-use opaque, and any prefix $P$ of $\hist$ is also an
    \SVA{} history, then every prefix of $\hist$ is also final-state last-use
    opaque. Thus, by \rdef{def:lopacity}, $\hist$ is last-use opaque.
\end{proof}
         
\section{Related Work} \label{sec:rw}

Ever since opacity \cite{GK08,GK10} was introduced, it seems,
there were attempts to weaken its stringent requirements, while retaining some
guarantees over what serializability \cite{BSW79,Pap79} provides. We explore
the most pertinent examples in \rsec{sec:properties}: TMS1, TMS2 \cite{DGLM13}, elastic opacity \cite{FGG09}, live opacity \cite{DFK14},
and VWC \cite{IMR08}, as well as some apposite consistency criteria:
recoverability \cite{Had88}, ACA \cite{BHG87}, strictness \cite{BHG87}, and
rigorousness \cite{BGRS91}. 
Other attempts were more specialized and include virtual time opacity
\cite{IMR08}, where the real-time order condition is relaxed. 
Similarly, the $\diamond$ opacity family of properties \cite{KKW15} relax the
time ordering requirements of opacity to provide properties applicable to
deffered update replication.
Another example is view transactions \cite{AMT10}, where it is only required
that a transaction commits on any snapshot, that can be different than the one
the transaction viewed initially, provided that operating on either snapshot
produced externally indistinguishable results.  
While these properties have specific applications, none weaken the consistency
to allow variable access before commit.

Although algorithms  and systems are not the focus of this paper, some systems
research that explores relaxed consistency should be noted. 
We already mention
our own \SVA{} \cite{SW14-hlpp} in \rsec{sec:lopaque-sva}.
Dynamic STM \cite{HLMS03} is another system with early release, and it can be
credited with introducing the concept of early release in the TM context.
Dynamic STM allows transactions that only perform read operations on particular
variables to (manually) release them for use by other transactions. However, it
left the assurance of safety to the programmer, and, as the authors state, even
linearizability cannot be guaranteed by the system.
The authors of \cite{SK06} expanded on the work above and evaluated the concept
of early release with respect to read-only variables on several concurrent data
structures. The results showed that this form of early release does not provide
a significant advantage in most cases, although there are scenarios where it
would be advantageous if it were automated. We use a different approach in \SVA{}, where
early release is not limited to read-only variables.
Twilight STM \cite{BMT10} relaxes isolation to allow transactions to read
variables that are used by other transactions, and allow them to re-read their
values as they change in order to maintain the correctness of the computation. If
inconsistencies arise, a reconciliation is attempted on commit, and aborts are
induced if this is impossible. Since it allows operating on variables that were
released early, but potentially before \last{} write, Twilight STM will not
satisfy the consistency requirements of \lopacity{}, but it is likely to
guarantee serializability and recoverability.

DATM  \cite{RRHW09} is yet another noteworthy system with an early release mechanism.
DATM is an optimistic multicore-oriented TM based on TL2 \cite{DSS06}, augmented
with early-release support.
It allows a transaction $\tr_i$ to write to a variable that was accessed by
some uncommitted transaction $\tr_j$, as long as $\tr_j$ commits before
$\tr_i$.
DATM also allows transaction $\tr_i$ to read a speculative value, one written
by $\tr_j$ and accessed by $\tr_i$ before $\tr_j$ commits. 
DATM detects if $\tr_j$ overwrites the data
or aborts, in which case $\tr_i$ is forced to restart. 
DATM allows all schedules alowed by conflict-serializability. This means that
DATM allows overwritting, as well as cascading aborts. It also means that it
does not satisfy \lopacity{}.

\section{Conclusions} \label{sec:conclusions}

This paper explored the space of TM safety properties in terms of whether or
not, and to what extent, they allow a transaction to release a variable early,
or, in other words, for a transaction to read a value written by a live
transaction.
We showed that existing properties are either strong, but prevent early release
altogether (opacity, TMS1 and TMS2), or pose impractical restrictions on the
ability of transactions to abort (VWC and live opacity).  The remainder of the
properties are not strong enough for TM applications (serializability and
recoverability) since they allow a large range of inconsistent views, including
overwriting.
Hence, we presented a new TM safety property called \lopacity{} that strikes a
reasonable compromise. It allows early release without a requirement for
transactions that release early not to abort, but one that is nevertheless
strong enough to prevent most inconsistent views and make others
inconsequential. The resulting property may be a useful practical criterion for
reasoning about TMs with early release support.

We discussed the histories that are allowed by \lopacity{} and examined the
guarantees the property gives to the programmer. Last-use opacity always allows
for potential inconsistent views to occur due to cascading aborts. However, no
    other inconsistent views are allowed. The inconsistent views that can occur
    can be made harmless by taking away the programmer's ability to execute
    arbitrary aborts by either removing that operation completely or by
    removing it from the programmer's toolkit, but allowing it to be used by
    the TM system, e.g. for fault tolerance. Allowing the programmer to abort
    a transaction at will means that they will need to eliminate dangerous situations
    (possible division by zero, invalid memory accesses, etc.) on a
    case-by-case basis. Nevertheless, we predict that inconsistent views of
    this sort will be relatively rare in practical TM applications, and
    typically result from using the abort operation to program business logic.
Alternatively, a variant of last-use opacity called $\beta$--last-use opacity
can be used instead, which eliminates the inconsistent views by preventing
early release in transactions where a programmatic abort is possible.

Finally, we discussed \SVA{}, a pessimistic concurrency control TM algorithm with early
release, which we show satisfies \lopacity{}.

\paragraph{Acknowledgements}
The project was funded from National Science Centre funds granted by decision
No. DEC-2012/06/M/ST6/\\00463.

\bibliographystyle{abbrv}
\bibliography{top-arxiv}%

\clearpage
\appendix

\section{Last-use Opacity History Examples}

\def\so{\Longrightarrow}
\def\pref#1{(\ref{#1})}

By $A \sqsubseteq B$ we denote that sequence $A$ is a substring of $B$.

Note the following about $\sspec{\obj}$ for any $\obj, \tr_i, \tr_j$:

    \begin{align}
    & \label{eq:seq-r}
      [\frop{i}{\obj}{0}] \in \sspec{\obj} \\
    & \label{eq:seq-w}
      [\fwop{i}{\obj}{1}{\ok_i}] \in \sspec{\obj} \\
    & \label{eq:seq-wa}
      [\fwop{i}{\obj}{1}{\ab_i}] \in \sspec{\obj} \\
    & \label{eq:seq-wr}
      [\fwop{i}{\obj}{1}{\ok_i}, \frop{j}{\obj}{1}] \in \sspec{\obj} \\
    & \label{eq:seq-wra}
      [\fwop{i}{\obj}{1}{\ok_i}, \frop{j}{\obj}{\ab_i}] \in \sspec{\obj} \\
    & \label{eq:seq-wrw}
      [\fwop{i}{\obj}{1}{\ok_i}, \frop{j}{\obj}{1}, \fwop{j}{\obj}{2}{\ok_j}] \in \sspec{\obj} \\
    & \label{eq:seq-wrwa}
      [\fwop{i}{\obj}{1}{\ok_i}, \frop{j}{\obj}{1}, \fwop{j}{\obj}{2}{\ab_j}] \in \sspec{\obj} \\
    & \label{eq:seq-empty}
      \varnothing \in \sspec{\obj} \\ 
    & \label{eq:ill-seq-r}
      [\frop{i}{\obj}{1}] \not\in \sspec{\obj} \\ 
    & \label{eq:ill-seq-wwr}
      [\fwop{i}{\obj}{1}{\ok_i},
       \fwop{i}{\obj}{2}{\ok_i},\frop{i}{\obj}{1}] \not\in \sspec{\obj} \\
    & \label{eq:ill-seq-wrrw}
      [\fwop{i}{\objx}{1}{\ok_i},\frop{i}{\objy}{1},
       \frop{j}{\objx}{1},\fwop{j}{\objy}{1}{\ok_j}] \not\in \sspec{\obj}  
    \end{align}

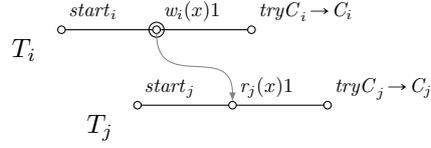
\begin{figure}
\begin{center}
\begin{tikzpicture}
     \draw
           (0,2)        node[tid]       {$\tr_i$}
                        node[aop]       {$\init_i$} %
                        node[dot]       {} 

      -- ++(1.25,0)     node[aop]       {$\twop{i}{\obj}{1}$}
                        node[dot] (wi)  {}
                        node[cir]       {}

      -- ++(1.25,0)     node[aop]       {$\tryC_i\!\to\!\co_i$}
                        node[dot]       {}         
                        ;

     \draw
           (1,1)        node[tid]       {$\tr_{j}$}
                        node[aop]       {$\init_{j}$} %
                        node[dot]       {} 

      -- ++(1.25,0)     node[aop]       {$\trop{j}{\obj}{1}$}
                        node[dot] (rj)  {}

      -- ++(1.25,0)     node[aop]       {$\tryC_j\!\to\!\co_j$}
                        node[dot]       {}
                        ;
     
     \draw[hb] (wi) \squiggle (rj);
\end{tikzpicture}
\end{center}
\caption{\label{fig:proof-h1} History $\hist_1$, last-use opaque.}
\end{figure}

\begin{lemma} \label{lemma:h1-fslop}
    $\hist_1$ is final-state last-use opaque.    
\end{lemma}

\begin{proof}
    \begin{align}
    & \text{let}~ C_1 = \compl{\hist_1} = \hist_1 \\    
    & \text{let}~ S_1 = C_1|\tr_i \cdot C_1|\tr_j \label{eq:let-s1} \\
    & %
      S_1 \equiv C_1 \label{eq:s1-equiv}\\
    & \text{real time order}~ \prec_{\hist_1} = \varnothing 
      \label{eq:h1-rt-h1} \\
    & \text{real time order}~ \prec_{S_1} = 
      \{ \tr_i \prec_{S_1} \tr_j \} \label{eq:h1-rt-s1} \\
    & \pref{eq:h1-rt-h1} \wedge \pref{eq:h1-rt-s1}
      \so \prec_{S_1} \subseteq \prec_{\hist_1} \label{eq:h1-rt-sub} \\
    & i = i \so S_1|\tr_i \subseteq \vis{S_1}{\tr_i} \label{eq:h1-vis-ti-ti} \\
    & \tr_i \prec_{S_1} \tr_j 
      \so S_1|\tr_j \nsubseteq \vis{S_1}{\tr_i} \label{eq:h1-vis-ti-tj} \\
    & \pref{eq:h1-vis-ti-ti} \wedge \pref{eq:h1-vis-ti-tj} 
      \so \vis{S_1}{\tr_i} = S_1|\tr_i \label{eq:s1-ti-vis} \\
    & \pref{eq:s1-ti-vis} \so \vis{S_1}{\tr_i}|\obj = [\fwop{i}{\obj}{1}{\ok_i}] 
      \label{eq:s1-ti-vis-x} \\
    & \pref{eq:s1-ti-vis-x} \wedge \pref{eq:seq-w} 
      \so \vis{S_1}{\tr_i} ~\text{is legal} \label{eq:h1-vis-ti-legal} \\
    & \pref{eq:h1-vis-ti-legal} \so  
      \tr_i ~\text{in}~ S_1 ~\text{is legal in}~ S_1 \label{eq:h1-ti-legal} \\
    & \tr_i \prec_{S_1} \tr_j \wedge \res{i}{}{\co_i} \in S_1|\tr_i 
      \so S_1|\tr_i \subseteq \vis{S_1}{\tr_i} \label{eq:h1-vis-tj-ti} \\
    & j = j \so  S_1|\tr_j \subseteq \vis{S_1}{\tr_i} \label{eq:h1-vis-tj-tj} \\
    & \pref{eq:h1-vis-tj-ti} \wedge \pref{eq:h1-vis-tj-tj} 
      \so \vis{S_1}{\tr_j} = S_1|\tr_i \cdot S_1|\tr_j \label{eq:s1-tj-vis} \\
    & \pref{eq:s1-tj-vis} 
      \so \vis{S_1}{\tr_j}|\obj = [\fwop{i}{\obj}{1}{\ok_i}, \frop{j}{\obj}{1}] 
      \label{eq:s1-tj-vis-x} \\
    & \pref{eq:s1-tj-vis-x} \wedge \pref{eq:seq-wr} 
      \so \vis{S_1}{\tr_j} ~\text{is legal} \label{eq:h1-vis-tj-legal} \\
    & \pref{eq:h1-vis-tj-legal} \so  
      \tr_j ~\text{in}~ S_1 ~\text{is legal in}~ S_1 \label{eq:h1-tj-legal} \\
    & \pref{eq:s1-equiv} \wedge \pref{eq:h1-rt-sub} 
      \wedge \pref{eq:h1-ti-legal} \wedge \pref{eq:h1-tj-legal} \so 
      \hist_1 ~\text{is final-state last-use opaque}
    \end{align}
\end{proof}

Let $P^1_1$ be a prefix s.t. $\hist_1 = P^1_1 \cdot [\res{j}{}{\co_j}]$.

\begin{lemma} \label{lemma:p11-fslop}
    $P^1_1$ is final-state last-use opaque.    
\end{lemma}

\begin{proof}
    \begin{align}
    & \text{let}~ C^1_1 = \compl{P^1_1} = P^1_1 \cdot [\res{j}{}{\co_j}] \\    
    & \hist_1 = C^1_1 \wedge ~\text{\rlemma{lemma:h1-fslop}}~ 
      \so P^1_1 ~\text{is final-state last-use opaque}
    \end{align}
\end{proof}

Let $P^1_2$ be a prefix s.t. $\hist_1 = P^1_2 \cdot [\tryC_j\to\co_j]$.

\begin{lemma} \label{lemma:p12-fslop}
    $P^1_2$ is final-state last-use opaque.    
\end{lemma}

\begin{proof}
    \begin{align}
    & \text{let}~ C^1_2 = \compl{P^1_2} = P^1_2 \cdot [\tryA_j\to\ab_j] \\    
    & \text{let}~ S^1_2 = C^1_2|\tr_i \cdot C^1_2|\tr_j \label{eq:let-s12} \\
    & %
      S^1_2 \equiv C^1_2 \label{eq:s12-equiv}\\
    & \text{real time order}~ \prec_{P^1_2} = \varnothing \label{eq:s12-rt-p12} \\
    & \text{real time order}~ \prec_{S^1_2} = 
      \{ \tr_i \prec_{S^1_2} \tr_j \} \label{eq:p12-rt-s12} \\
    & \pref{eq:s12-rt-p12} \wedge \pref{eq:p12-rt-s12}
      \so \prec_{S^1_2} \subseteq \prec_{P^1_2} \label{eq:p12-rt-sub} \\
    & i = i \so S^1_2|\tr_i \subseteq \vis{S^1_2}{\tr_i} \label{eq:p12-vis-ti-ti} \\
    & \tr_i \prec_{S^1_2} \tr_j 
      \so S^1_2|\tr_j \nsubseteq \vis{S^1_2}{\tr_i} \label{eq:p12-vis-ti-tj} \\
    & \pref{eq:p12-vis-ti-ti} \wedge \pref{eq:p12-vis-ti-tj} 
      \so \vis{S^1_2}{\tr_i} = S^1_2|\tr_i \label{eq:p12-ti-vis} \\
    & \pref{eq:p12-ti-vis} \so \vis{S^1_2}{\tr_i}|\obj = [\fwop{i}{\obj}{1}{\ok_i}] 
      \label{eq:p12-ti-vis-x} \\
    & \pref{eq:p12-ti-vis-x} \wedge \pref{eq:seq-w} 
      \so \vis{S^1_2}{\tr_i} ~\text{is legal} \label{eq:p12-vis-ti-legal} \\
    & \pref{eq:p12-vis-ti-legal} \so  
      \tr_i ~\text{in}~ S^1_2 ~\text{is legal in}~ S^1_2 \label{eq:p12-ti-legal} \\
    & \tr_i \prec_{S^1_2} \tr_j \wedge \res{i}{}{\co_i} \in S^1_2|\tr_i 
      \so S^1_2|\tr_i \subseteq \luvis{S^1_2}{\tr_j} \label{eq:p12-lvis-tj-ti} \\
    & j = j \so  S^1_2|\tr_j \subseteq \luvis{S^1_2}{\tr_j} \label{eq:p12-lvis-tj-tj} \\
    & \pref{eq:p12-lvis-tj-ti} \wedge \pref{eq:p12-lvis-tj-tj} 
      \so \luvis{S^1_2}{\tr_j} = S^1_2|\tr_i \cdot S^1_2|\tr_j \label{eq:p12-tj-lvis} \\
    & \pref{eq:p12-tj-lvis} 
      \so \luvis{S^1_2}{\tr_j}|\obj = [\fwop{i}{\obj}{1}{\ok_i}, \frop{j}{\obj}{1}] 
      \label{eq:p12-tj-lvis-x} \\
    & \pref{eq:p12-tj-lvis-x} \wedge \pref{eq:seq-wr} 
      \so \luvis{S^1_2}{\tr_j} ~\text{is legal} \label{eq:p12-lvis-tj-legal} \\
    & \pref{eq:p12-lvis-tj-legal} \so  
      \tr_j ~\text{in}~ S^1_2 ~\text{is last-use legal in}~ S^1_2 \label{eq:p12-tj-legal} \\
    & \pref{eq:s12-equiv} \wedge \pref{eq:p12-rt-sub} 
      \wedge \pref{eq:p12-ti-legal} \wedge \pref{eq:p12-tj-legal} \so 
      P^1_2 ~\text{is final-state last-use opaque}
    \end{align}
\end{proof}

Let $P^1_3$ be a prefix s.t. $\hist_1 = P^1_3 \cdot [\res{i}{}{\co_i}, \tryC_j\to\co_j]$.

\begin{lemma} \label{lemma:p13-fslop}
    $P^1_3$ is final-state last-use opaque.    
\end{lemma}

\begin{proof}
    \begin{align}
    & \text{let}~ C^1_3 = \compl{P^1_3} = P^1_3 \cdot [\res{i}{}{\co_i}, \tryC_j\to\co_j] \\
    & P^1_2 = C^1_3 \wedge ~\text{\rlemma{lemma:p12-fslop}}~ 
      \so P^1_3 ~\text{is final-state last-use opaque}
    \end{align}
\end{proof}

Let $P^1_4$ be a prefix s.t. $\hist_1 = P^1_4 \cdot [\tryC_i\to\co_i, \tryC_j\to\co_j]$.

\begin{lemma} \label{lemma:p14-fslop}
    $P^1_4$ is final-state last-use opaque.    
\end{lemma}

\begin{proof}
    \begin{align}
    & \text{let}~ C^1_4 = \compl{P^1_4} = P^1_4 \cdot [\tryA_i\to\ab_i, \tryA_j\to\ab_j] \\
    & \text{let}~ S^1_4 = C^1_4|\tr_i \cdot C^1_4|\tr_j \label{eq:let-s14} \\
    & %
      S^1_4 \equiv C^1_4 \label{eq:s14-equiv}\\
    & \text{real time order}~ \prec_{P^1_4} = \varnothing \label{eq:s14-rt-p14} \\
    & \text{real time order}~ \prec_{S^1_4} = 
      \{ \tr_i \prec_{S^1_4} \tr_j \} \label{eq:p14-rt-s14} \\
    & \pref{eq:s14-rt-p14} \wedge \pref{eq:p14-rt-s14}
      \so \prec_{S^1_4} \subseteq \prec_{P^1_4} \label{eq:p14-rt-sub} \\
    & i = i \so S^1_4|\tr_i \subseteq \luvis{S^1_4}{\tr_i} \label{eq:p14-lvis-ti-ti} \\
    & \tr_i \prec_{S^1_4} \tr_j 
      \so S^1_4|\tr_j \nsubseteq \luvis{S^1_4}{\tr_i} \label{eq:p14-lvis-ti-tj} \\
    & \pref{eq:p14-lvis-ti-ti} \wedge \pref{eq:p14-lvis-ti-tj} 
      \so \luvis{S^1_4}{\tr_i} = S^1_4|\tr_i \label{eq:p14-ti-lvis} \\
    & \pref{eq:p14-ti-lvis} \so \vis{S^1_4}{\tr_i}|\obj = [\fwop{i}{\obj}{1}{\ok_i}] 
      \label{eq:p14-ti-lvis-x} \\
    & \pref{eq:p14-ti-lvis-x} \wedge \pref{eq:seq-w} 
      \so \luvis{S^1_4}{\tr_i} ~\text{is legal} \label{eq:p14-lvis-ti-legal} \\
    & \pref{eq:p14-lvis-ti-legal} \so  
      \tr_i ~\text{in}~ S^1_4 ~\text{is last-use legal in}~ S^1_4 \label{eq:p14-ti-legal} \\
    & \fwop{i}{\obj}{1}{\ok_i} ~\text{is \last{} write on}~\obj~\text{in}~\tr_i
      \so \tr_i ~\text{is decided on}~\obj 
      S^1_4 \label{eq:p14-ti-decided} \\
    & \tr_i \prec_{S^1_4} \tr_j \wedge \pref{eq:p14-ti-decided}
      \so S^1_4\cpeC\tr_i \subseteq \luvis{S^1_4}{\tr_j} \label{eq:p14-lvis-tj-ti} \\
    & j = j \so  S^1_4|\tr_j \subseteq \luvis{S^1_4}{\tr_j} \label{eq:p14-lvis-tj-tj} \\
    & \pref{eq:p14-lvis-tj-ti} \wedge \pref{eq:p14-lvis-tj-tj} 
      \so \luvis{S^1_4}{\tr_j} = S^1_4\cpeC\tr_i \cdot S^1_4|\tr_j \label{eq:p14-tj-lvis} \\
    & \pref{eq:p14-tj-lvis} 
      \so \luvis{S^1_4}{\tr_j}|\obj = [\fwop{i}{\obj}{1}{\ok_i}, \frop{j}{\obj}{1}]
      \label{eq:p14-tj-lvis-x} \\
    & \pref{eq:p14-tj-lvis-x} \wedge \pref{eq:seq-wr} 
      \so \luvis{S^1_4}{\tr_j} ~\text{is legal} \label{eq:p14-lvis-tj-legal} \\
    & \pref{eq:p14-lvis-tj-legal} \so  
      \tr_j ~\text{in}~ S^1_4 ~\text{is last-use legal in}~ S^1_4 \label{eq:p14-tj-legal} \\
    & \pref{eq:s14-equiv} \wedge \pref{eq:p14-rt-sub} 
      \wedge \pref{eq:p14-ti-legal} \wedge \pref{eq:p14-tj-legal} \so 
      P^1_4 ~\text{is final-state last-use opaque}
    \end{align}
\end{proof}

Let $P^1_5$ be a prefix s.t. $\hist_1 = P^1_5 \cdot [\res{i}{}{1}, \tryC_i\to\co_i, \tryC_j\to\co_j]$.

\begin{lemma} \label{lemma:p15-fslop}
    $P^1_5$ is final-state last-use opaque.    
\end{lemma}

\begin{proof}
    \begin{align}
    & \text{let}~ C^1_5 = \compl{P^1_5} = P^1_5 \cdot [\res{i}{\ab_i}, \tryA_j\to\ab_j] \\
    & \text{let}~ S^1_5 = C^1_5|\tr_i \cdot C^1_5|\tr_j \label{eq:let-s15} \\
    & %
      S^1_5 \equiv C^1_5 \label{eq:s15-equiv}\\
    & \text{real time order}~ \prec_{P^1_5} = \varnothing \label{eq:s15-rt-p15} \\
    & \text{real time order}~ \prec_{S^1_5} = 
      \{ \tr_i \prec_{S^1_5} \tr_j \} \label{eq:p15-rt-s15} \\
    & \pref{eq:s15-rt-p15} \wedge \pref{eq:p15-rt-s15}
      \so \prec_{S^1_5} \subseteq \prec_{P^1_5} \label{eq:p15-rt-sub} \\
    & i = i \so S^1_5|\tr_i \subseteq \luvis{S^1_5}{\tr_i} \label{eq:p15-lvis-ti-ti} \\
    & \tr_i \prec_{S^1_5} \tr_j 
      \so S^1_5|\tr_j \nsubseteq \luvis{S^1_5}{\tr_i} \label{eq:p15-lvis-ti-tj} \\
    & \pref{eq:p15-lvis-ti-ti} \wedge \pref{eq:p15-lvis-ti-tj} 
      \so \luvis{S^1_5}{\tr_i} = S^1_5|\tr_i \label{eq:p15-ti-lvis} \\
    & \pref{eq:p15-ti-lvis} \so \vis{S^1_5}{\tr_i}|\obj = [\fwop{i}{\obj}{1}{\ok_i}] 
      \label{eq:p15-ti-lvis-x} \\
    & \pref{eq:p15-ti-lvis-x} \wedge \pref{eq:seq-w} 
      \so \luvis{S^1_5}{\tr_i} ~\text{is legal} \label{eq:p15-lvis-ti-legal} \\
    & \pref{eq:p15-lvis-ti-legal} \so  
      \tr_i ~\text{in}~ S^1_5 ~\text{is last-use legal in}~ S^1_5 \label{eq:p15-ti-legal} \\
    & \fwop{i}{\obj}{1}{\ok_i} ~\text{is \last{} write on}~\obj~\text{in}~\tr_i
      \so \tr_i ~\text{is decided on}~\obj \label{eq:p15-ti-decided} \\
    & \tr_i \prec_{S^1_5} \tr_j \wedge \pref{eq:p15-ti-decided}
      \so S^1_5\cpeC\tr_i \subseteq \luvis{S^1_5}{\tr_j} \label{eq:p15-lvis-tj-ti} \\
    & j = j \so  S^1_5|\tr_j \subseteq \luvis{S^1_5}{\tr_j} \label{eq:p15-lvis-tj-tj} \\
    & \pref{eq:p15-lvis-tj-ti} \wedge \pref{eq:p15-lvis-tj-tj} 
      \so \luvis{S^1_5}{\tr_j} = S^1_5\cpeC\tr_i \cdot S^1_5|\tr_j \label{eq:p15-tj-lvis} \\
    & \pref{eq:p15-tj-lvis} 
      \so \luvis{S^1_5}{\tr_j}|\obj = [\fwop{i}{\obj}{1}{\ok_i}, \frop{j}{\obj}{\ab_i}] 
      \label{eq:p15-tj-lvis-x} \\
    & \pref{eq:p15-tj-lvis-x} \wedge \pref{eq:seq-wra} 
      \so \luvis{S^1_5}{\tr_j} ~\text{is legal} \label{eq:p15-lvis-tj-legal} \\
    & \pref{eq:p15-lvis-tj-legal} \so  
      \tr_j ~\text{in}~ S^1_5 ~\text{is last-use legal in}~ S^1_5 \label{eq:p15-tj-legal} \\
    & \pref{eq:s15-equiv} \wedge \pref{eq:p15-rt-sub} 
      \wedge \pref{eq:p15-ti-legal} \wedge \pref{eq:p15-tj-legal} \so 
      P^1_5 ~\text{is final-state last-use opaque}
    \end{align}
\end{proof}

Let $P^1_6$ be a prefix s.t. $\hist_1 = P^1_6 \cdot 
    [\frop{1}{\obj}{1}, \tryC_i\to\co_i, \tryC_j\to\co_j]$.

\begin{lemma} \label{lemma:p16-fslop}
    $P^1_6$ is final-state last-use opaque.    
\end{lemma}

\begin{proof}
    \begin{align}
    & \text{let}~ C^1_6 = \compl{P^1_6} = P^1_6 \cdot [\tryA_i\to\ab_i, \tryA_j\to\ab_j] \\
    & \text{let}~ S^1_6 = C^1_6|\tr_i \cdot C^1_6|\tr_j \label{eq:let-s16} \\
    & %
      S^1_6 \equiv C^1_6 \label{eq:s16-equiv}\\
    & \text{real time order}~ \prec_{P^1_6} = \varnothing \label{eq:s16-rt-p16} \\
    & \text{real time order}~ \prec_{S^1_6} = 
      \{ \tr_i \prec_{S^1_6} \tr_j \} \label{eq:p16-rt-s16} \\
    & \pref{eq:s16-rt-p16} \wedge \pref{eq:p16-rt-s16}
      \so \prec_{S^1_6} \subseteq \prec_{P^1_6} \label{eq:p16-rt-sub} \\
    & i = i \so S^1_6|\tr_i \subseteq \luvis{S^1_6}{\tr_i} \label{eq:p16-lvis-ti-ti} \\
    & \tr_i \prec_{S^1_6} \tr_j 
      \so S^1_6|\tr_j \nsubseteq \luvis{S^1_6}{\tr_i} \label{eq:p16-lvis-ti-tj} \\
    & \pref{eq:p16-lvis-ti-ti} \wedge \pref{eq:p16-lvis-ti-tj} 
      \so \luvis{S^1_6}{\tr_i} = S^1_6|\tr_i \label{eq:p16-ti-lvis} \\
    & \pref{eq:p16-ti-lvis} \so \vis{S^1_6}{\tr_i}|\obj = [\fwop{i}{\obj}{1}{\ok_i}] 
      \label{eq:p16-ti-lvis-x} \\
    & \pref{eq:p16-ti-lvis-x} \wedge \pref{eq:seq-w} 
      \so \luvis{S^1_6}{\tr_i} ~\text{is legal} \label{eq:p16-lvis-ti-legal} \\
    & \pref{eq:p16-lvis-ti-legal} \so  
      \tr_i ~\text{in}~ S^1_6 ~\text{is last-use legal in}~ S^1_6 \label{eq:p16-ti-legal} \\
    & \fwop{i}{\obj}{1}{\ok_i} ~\text{is \last{} write on}~\obj~\text{in}~\tr_i
      \so \tr_i ~\text{is decided on}~\obj \label{eq:p16-ti-decided} \\
    & \tr_i \prec_{S^1_6} \tr_j \wedge \pref{eq:p16-ti-decided}
      \so S^1_6\cpeC\tr_i \subseteq \luvis{S^1_6}{\tr_j} \label{eq:p16-lvis-tj-ti} \\
    & j = j \so  S^1_6|\tr_j \subseteq \luvis{S^1_6}{\tr_j} \label{eq:p16-lvis-tj-tj} \\
    & \pref{eq:p16-lvis-tj-ti} \wedge \pref{eq:p16-lvis-tj-tj} 
      \so \luvis{S^1_6}{\tr_j} = S^1_6\cpeC\tr_i \cdot S^1_6|\tr_j \label{eq:p16-tj-lvis} \\
    & \pref{eq:p16-tj-lvis} 
      \so \luvis{S^1_6}{\tr_j}|\obj = [\fwop{i}{\obj}{1}{\ok_i} ] 
      \label{eq:p16-tj-lvis-x} \\
    & \pref{eq:p16-tj-lvis-x} \wedge \pref{eq:seq-w} 
      \so \luvis{S^1_6}{\tr_j} ~\text{is legal} \label{eq:p16-lvis-tj-legal} \\
    & \pref{eq:p16-lvis-tj-legal} \so  
      \tr_j ~\text{in}~ S^1_6 ~\text{is last-use legal in}~ S^1_6 \label{eq:p16-tj-legal} \\
    & \pref{eq:s16-equiv} \wedge \pref{eq:p16-rt-sub} 
      \wedge \pref{eq:p16-ti-legal} \wedge \pref{eq:p16-tj-legal} \so 
      P^1_6 ~\text{is final-state last-use opaque}
    \end{align}
\end{proof}

Let $P^1_7$ be a prefix s.t. $\hist_1 = P^1_7 \cdot 
    [\res{i}{}{\ok_i}, \frop{j}{\obj}{1}, \tryC_i\to\co_i, \tryC_j\to\co_j]$.

\begin{lemma} \label{lemma:p17-fslop}
    $P^1_7$ is final-state last-use opaque.    
\end{lemma}

\begin{proof}
    \begin{align}
    & \text{let}~ C^1_7 = \compl{P^1_7} = P^1_7 \cdot [\res{i}{}{\ab_i}, \tryA_j\to\ab_j] \\
    & \text{let}~ S^1_7 = C^1_7|\tr_i \cdot C^1_7|\tr_j \label{eq:let-s17} \\
    & %
      S^1_7 \equiv C^1_7 \label{eq:s17-equiv} \\
    & \text{real time order}~ \prec_{P^1_7} = \varnothing \label{eq:s17-rt-p17} \\
    & \text{real time order}~ \prec_{S^1_7} = 
      \{ \tr_i \prec_{S^1_7} \tr_j \} \label{eq:p17-rt-s17} \\
    & \pref{eq:s17-rt-p17} \wedge \pref{eq:p17-rt-s17}
      \so \prec_{S^1_7} \subseteq \prec_{P^1_7} \label{eq:p17-rt-sub} \\
    & i = i \so S^1_7|\tr_i \subseteq \luvis{S^1_7}{\tr_i} \label{eq:p17-lvis-ti-ti} \\
    & \tr_i \prec_{S^1_7} \tr_j 
      \so S^1_7|\tr_j \nsubseteq \luvis{S^1_7}{\tr_i} \label{eq:p17-lvis-ti-tj} \\
    & \pref{eq:p17-lvis-ti-ti} \wedge \pref{eq:p17-lvis-ti-tj} 
      \so \luvis{S^1_7}{\tr_i} = S^1_7|\tr_i \label{eq:p17-ti-lvis} \\
    & \pref{eq:p17-ti-lvis} \so \vis{S^1_7}{\tr_i}|\obj = [\fwop{i}{\obj}{1}{\ab_i}] 
      \label{eq:p17-ti-lvis-x} \\
    & \pref{eq:p17-ti-lvis-x} \wedge \pref{eq:seq-w} 
      \so \luvis{S^1_7}{\tr_i} ~\text{is legal} \label{eq:p17-lvis-ti-legal} \\
    & \pref{eq:p17-lvis-ti-legal} \so  
      \tr_i ~\text{in}~ S^1_7 ~\text{is last-use legal in}~ S^1_7 \label{eq:p17-ti-legal} \\
    & \fwop{i}{\obj}{1}{\ab_i} ~\text{is not \last{} write on}~\obj~\text{in}~\tr_i 
      \label{eq:p17-tj-not-last} \\
    & S^1_7|\tr_j|x \setminus \{ \fwop{i}{\obj}{1}{\ab_i} \} = \varnothing \label{eq:p17-tj-no-last} \\ 
    & \res{i}{}{\co_i} \not\in S^1_7|\tr_i \wedge \pref{eq:p17-tj-no-last} 
      \so S^1_7\cpeC\tr_i \nsubseteq \luvis{S^1_7}{\tr_j} \label{eq:p17-lvis-tj-ti} \\
    & j = j \so  S^1_7|\tr_j \subseteq \luvis{S^1_7}{\tr_j} \label{eq:p17-lvis-tj-tj} \\
    & \pref{eq:p17-lvis-tj-ti} \wedge \pref{eq:p17-lvis-tj-tj} 
      \so \luvis{S^1_7}{\tr_j} = S^1_7|\tr_j \label{eq:p17-tj-lvis} \\
    & \pref{eq:p17-tj-lvis} 
      \so \luvis{S^1_7}{\tr_j}|\obj = [ ] 
      \label{eq:p17-tj-lvis-x} \\
    & \pref{eq:p17-tj-lvis-x} \wedge \pref{eq:seq-empty} 
      \so \luvis{S^1_7}{\tr_j} ~\text{is legal} \label{eq:p17-lvis-tj-legal} \\
    & \pref{eq:p17-lvis-tj-legal} \so  
      \tr_j ~\text{in}~ S^1_7 ~\text{is last-use legal in}~ S^1_7 \label{eq:p17-tj-legal} \\
    & \pref{eq:s17-equiv} \wedge \pref{eq:p17-rt-sub} 
      \wedge \pref{eq:p17-ti-legal} \wedge \pref{eq:p17-tj-legal} \so 
      P^1_7 ~\text{is final-state last-use opaque}
    \end{align}
\end{proof}

Let $P^1_p$ be a any prefix s.t. $\hist_1 = P^1_p \cdot R \cdot
    [\fwop{i}{\obj}{1}{ok_i}, \frop{j}{\obj}{1}, \tryC_i\to\co_i, \tryC_j\to\co_j]$.

\begin{lemma} \label{lemma:p1p-fslop}
    Any $P^1_p$ is final-state last-use opaque.    
\end{lemma}

\begin{proof}
    Since $P^1_p$ does not contain any reads or writes, for any sequential
    history $S \equiv P^1_p$, transactions $\tr_i$ and $\tr_j$ are trivially
    both legal and last-use legal in $S$. Thus, $P^1_p$ is final-state last-use
    opaque.
\end{proof}

\begin{lemma} \label{lemma:h1-lop}
    $\hist_1$ is last-use opaque.    
\end{lemma}

\begin{proof}
    Since, from Lemmas \ref{lemma:p11-fslop}--\ref{lemma:p1p-fslop}, all
    prefixes of $\hist_1$ are final-state last-use opaque, then by
    \rdef{def:lopacity} $\hist_1$ is last-use opaque.
\end{proof}

\begin{corollary} \label{cor:h1-pref-lop}
    Any prefix of $\hist_1$ is last-use opaque.
\end{corollary}

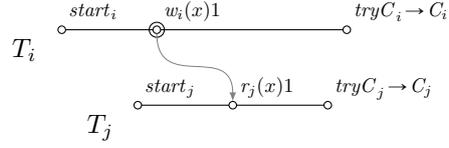
\begin{figure}
\begin{center}
\begin{tikzpicture}
     \draw
           (0,2)        node[tid]       {$\tr_i$}
                        node[aop]       {$\init_i$} %
                        node[dot]       {} 

      -- ++(1.25,0)     node[aop]       {$\twop{i}{\obj}{1}$}
                        node[dot] (wi)  {}
                        node[cir]       {}

      -- ++(2.5,0)      node[aop]       {$\tryC_i\!\to\!\co_i$}
                        node[dot]       {}         
                        ;

     \draw
           (1,1)        node[tid]       {$\tr_{j}$}
                        node[aop]       {$\init_{j}$} %
                        node[dot]       {} 

      -- ++(1.25,0)     node[aop]       {$\trop{j}{\obj}{1}$}
                        node[dot] (rj)  {}

      -- ++(1.25,0)     node[aop]       {$\tryC_j\!\to\!\co_j$}
                        node[dot]       {}
                        ;
     
     \draw[hb] (wi) \squiggle (rj);
\end{tikzpicture}
\end{center}
\caption{\label{fig:proof-h2} History $\hist_2$, not last-use opaque.}
\end{figure}

\begin{lemma} \label{lemma:h2-fslop}
    $\hist_2$ is final-state last-use opaque.    
\end{lemma}

\begin{proof}
    \begin{align}
    & \text{let}~ C_2 = \compl{\hist_2} = \hist_2 \\    
    & \text{let}~ S_2 = C_1|\tr_i \cdot C_2|\tr_j \label{eq:let-s2} \\
    & %
      S_2 \equiv C_2 \label{eq:s2-equiv}\\
    & S_2 = S_1  \wedge ~\text{\rlemma{lemma:h1-fslop}}~ \so 
      \hist_2 ~\text{is final-state last-use opaque}
    \end{align}
\end{proof}

Let $P^2_1$ be a prefix s.t. $\hist_2 = P^2_1 \cdot [\res{i}{}{\co_i}]$.

\begin{lemma} \label{lemma:p21-fslop}
    $P^2_1$ is final-state last-use opaque.    
\end{lemma}

\begin{proof}
    \begin{align}
    & \text{let}~ C^2_1 = \compl{P^2_1} = P^2_1 \cdot [\res{j}{}{\co_j}] \\    
    & \hist_2 = C^2_1 \wedge ~\text{\rlemma{lemma:h2-fslop}}~ 
      \so P^2_1 ~\text{is final-state last-use opaque}
    \end{align}
\end{proof}

Let $P^2_2$ be a prefix s.t. $\hist_2 = P^2_2 \cdot [\tryC_i\to\co_i]$.

\begin{lemma} \label{lemma:p22-fslop}
    $P^2_2$ is not final-state last-use opaque.    
\end{lemma}

\begin{proof}
    \begin{align}
    & \text{let}~ C^2_2 = \compl{P^2_2} = P^2_2 \cdot [\tryA_j\to\ab_j] \\    
    & \text{let}~ S^2_2 = C^2_2|\tr_i \cdot C^2_2|\tr_j \label{eq:let-s22} \\
    & %
      S^2_2 \equiv C^2_2 \label{eq:s22-equiv}\\
    & \text{real time order}~ \prec_{P^2_2} = \varnothing \label{eq:s22-rt-p22} \\
    & \text{real time order}~ \prec_{S^2_2} = 
      \{ \tr_i \prec_{S^2_2} \tr_j \} \label{eq:p22-rt-s22} \\
    & \pref{eq:s22-rt-p22} \wedge \pref{eq:p22-rt-s22}
      \so \prec_{S^2_2} \subseteq \prec_{P^2_2} \label{eq:p22-rt-sub} \\
    & \res{i}{}{\ab_i} \in S^2_2|\tr_i 
      \so S^2_2|\tr_i \nsubseteq \luvis{S^2_2}{\tr_j} \label{eq:p22-vis-tj-ti} \\
    & j = j \so  S^2_2|\tr_j \subseteq \luvis{S^2_2}{\tr_j} \label{eq:p22-vis-tj-tj} \\
    & \pref{eq:p22-vis-tj-ti} \wedge \pref{eq:p22-vis-tj-tj} 
      \so \luvis{S^2_2}{\tr_j} = S^2_2|\tr_j \label{eq:p22-tj-vis} \\
    & \pref{eq:p22-tj-vis} 
      \so \luvis{S^2_2}{\tr_j}|\obj = [\frop{j}{\obj}{1}] 
      \label{eq:p22-tj-lvis-x} \\
    & \pref{eq:p22-tj-lvis-x} \wedge \pref{eq:ill-seq-r} 
      \so \luvis{S^2_2}{\tr_j} ~\text{is not legal} \label{eq:p22-vis-tj-not-legal} \\
    & \pref{eq:p22-vis-tj-not-legal} \so  
      \tr_j ~\text{in}~ S^2_2 ~\text{is not legal in}~ S^2_2 \label{eq:p22-tj-not-legal} \\
    & \text{let}~ \dot{S}^2_2 = C^2_2|\tr_j \cdot C^2_2|\tr_i \label{eq:let-s22b} \\
    & \dot{S}^2_2 \equiv C^2_2 \label{eq:s22b-equiv}\\
    & \text{real time order}~ \prec_{P^2_2} = \varnothing \label{eq:s22b-rt-p22b} \\
    & \text{real time order}~ \prec_{\dot{S}^2_2} = 
      \{ \tr_i \prec_{\dot{S}^2_2} \tr_j \} \label{eq:p22b-rt-s22b} \\
    & \pref{eq:s22b-rt-p22b} \wedge \pref{eq:p22b-rt-s22b}
      \so \prec_{\dot{S}^2_2} \subseteq \prec_{P^2_2} \label{eq:p22b-rt-sub} \\
    & \tr_j \prec_{\dot{S}^2_2} \tr_i 
      \so \dot{S}^2_2|\tr_i \nsubseteq \vis{\dot{S}^2_2}{\tr_j} \label{eq:p22b-vis-tj-ti} \\
    & j = j \so  \dot{S}^2_2|\tr_j \subseteq \vis{\dot{S}^2_2}{\tr_i} \label{eq:p22b-vis-tj-tj} \\
    & \pref{eq:p22b-vis-tj-ti} \wedge \pref{eq:p22b-vis-tj-tj} 
      \so \vis{\dot{S}^2_2}{\tr_j} = \dot{S}^2_2|\tr_j \label{eq:p22b-tj-vis} \\
    & \pref{eq:p22b-tj-vis} 
      \so \vis{\dot{S}^2_2}{\tr_j}|\obj = [\frop{j}{\obj}{1}] 
      \label{eq:p22b-tj-vis-x} \\
    & \pref{eq:p22b-tj-vis-x} \wedge \pref{eq:ill-seq-r} 
      \so \luvis{\dot{S}^2_2}{\tr_j} ~\text{is not legal} \label{eq:p22b-vis-tj-not-legal} \\
    & \pref{eq:p22b-vis-tj-not-legal} \so  
      \tr_j ~\text{in}~ \dot{S}^2_2 ~\text{is not legal in}~ \dot{S}^2_2 
      \label{eq:p22b-tj-not-legal} \\
    & \pref{eq:p22-tj-not-legal} 
      \wedge \pref{eq:p22b-tj-not-legal} \so 
      P^2_2 ~\text{is not final-state last-use opaque}
    \end{align}
\end{proof}

\begin{lemma} \label{lemma:h2-lop}
    $\hist_2$ is not last-use opaque.    
\end{lemma}

\begin{proof}
    Even though, from \rlemma{lemma:h2-fslop}, $\hist_2$ is final-state
    last-use opaque, from \rlemma{lemma:p22-fslop}, prefix $P^2_2$ of $\hist_2$
    is not final-state last-use opaque, so, from \rdef{def:lopacity} $\hist_2$
    is not last-use opaque.
\end{proof}

\begin{figure}
\begin{center}
\begin{tikzpicture}
     \draw
           (0,2)        node[tid]       {$\tr_i$}
                        node[aop]       {$\init_i$} %
                        node[dot]       {} 

      -- ++(1.25,0)     node[aop]       {$\twop{i}{\obj}{1}$}
                        node[dot] (wi)  {}
                        node[cir]       {}

      -- ++(1.25,0)     node[aop]       {$\tryA_i\!\to\!\ab_i$}
                        node[dot]       {}         
                        ;

     \draw
           (1,1)        node[tid]       {$\tr_{j}$}
                        node[aop]       {$\init_{j}$} %
                        node[dot]       {} 

      -- ++(1.25,0)     node[aop]       {$\trop{j}{\obj}{1}$}
                        node[dot] (rj)  {}

      -- ++(1.25,0)     node[aop]       {$\tryC_j\!\to\!\ab_j$}
                        node[dot]       {}
                        ;
     
     \draw[hb] (wi) \squiggle (rj);
\end{tikzpicture}
\end{center}
\caption{\label{fig:proof-h3} History $\hist_3$, last-use opaque.}
\end{figure}
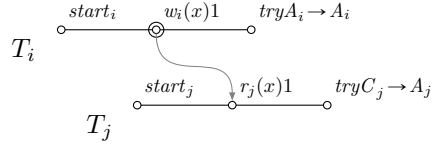

\begin{lemma} \label{lemma:h3-fslop}
    $\hist_3$ is final-state last-use opaque.    
\end{lemma}

\begin{proof}
    \begin{align}
    & \text{let}~ C_3 = \compl{\hist_3} = \hist_3 \\
    & \text{let}~ S_3 = C_3|\tr_i \cdot C_3|\tr_j \label{eq:let-s3} \\
    & %
      S_3 \equiv C_3 \label{eq:s3-equiv}\\
    & \text{real time order}~ \prec_{\hist_3} = \varnothing \label{eq:s3-rt-h3} \\
    & \text{real time order}~ \prec_{S_3} = 
      \{ \tr_i \prec_{S_3} \tr_j \} \label{eq:h3-rt-s3} \\
    & \pref{eq:s3-rt-h3} \wedge \pref{eq:h3-rt-s3}
      \so \prec_{S_3} \subseteq \prec_{\hist_3} \label{eq:h3-rt-sub} \\
    & i = i \so S_3|\tr_i \subseteq \luvis{S_3}{\tr_i} \label{eq:h3-lvis-ti-ti} \\
    & \tr_i \prec_{S_3} \tr_j 
      \so S_3|\tr_j \nsubseteq \luvis{S_3}{\tr_i} \label{eq:h3-lvis-ti-tj} \\
    & \pref{eq:h3-lvis-ti-ti} \wedge \pref{eq:h3-lvis-ti-tj} 
      \so \luvis{S_3}{\tr_i} = S_3|\tr_i \label{eq:h3-ti-lvis} \\
    & \pref{eq:h3-ti-lvis} \so \vis{S_3}{\tr_i}|\obj = [\fwop{i}{\obj}{1}{\ok_i}] 
      \label{eq:h3-ti-lvis-x} \\
    & \pref{eq:h3-ti-lvis-x} \wedge \pref{eq:seq-w} 
      \so \luvis{S_3}{\tr_i} ~\text{is legal} \label{eq:h3-lvis-ti-legal} \\
    & \pref{eq:h3-lvis-ti-legal} \so  
      \tr_i ~\text{in}~ S_3 ~\text{is last-use legal in}~ S_3 \label{eq:h3-ti-legal} \\
    & \fwop{i}{\obj}{1}{\ok_i} ~\text{is \last{} write on}~\obj~\text{in}~\tr_i
      \so \tr_i ~\text{is decided on}~\obj 
      S_3 \label{eq:h3-ti-decided} \\
    & \tr_i \prec_{S_3} \tr_j \wedge \pref{eq:h3-ti-decided}
      \so S_3\cpeC\tr_i \subseteq \luvis{S_3}{\tr_j} \label{eq:h3-lvis-tj-ti} \\
    & j = j \so  S_3|\tr_j \subseteq \luvis{S_3}{\tr_j} \label{eq:h3-lvis-tj-tj} \\
    & \pref{eq:h3-lvis-tj-ti} \wedge \pref{eq:h3-lvis-tj-tj} 
      \so \luvis{S_3}{\tr_j} = S_3\cpeC\tr_i \cdot S_3|\tr_j \label{eq:h3-tj-lvis} \\
    & \pref{eq:h3-tj-lvis} 
      \so \luvis{S_3}{\tr_j}|\obj = [\fwop{i}{\obj}{1}{\ok_i}, \frop{j}{\obj}{1}]
      \label{eq:h3-tj-lvis-x} \\
    & \pref{eq:h3-tj-lvis-x} \wedge \pref{eq:seq-wr} 
      \so \luvis{S_3}{\tr_j} ~\text{is legal} \label{eq:h3-lvis-tj-legal} \\
    & \pref{eq:h3-lvis-tj-legal} \so  
      \tr_j ~\text{in}~ S_3 ~\text{is last-use legal in}~ S_3 \label{eq:h3-tj-legal} \\
    & \pref{eq:s3-equiv} \wedge \pref{eq:h3-rt-sub} 
      \wedge \pref{eq:h3-ti-legal} \wedge \pref{eq:h3-tj-legal} \so 
      \hist_3 ~\text{is final-state last-use opaque}
    \end{align}
\end{proof}

Let $P^3_1$ be a prefix s.t. $\hist_3 = P^3_1 \cdot [\res{j}{}{\ab_j}]$.

\begin{lemma} \label{lemma:p31-fslop}
    $P^3_1$ is final-state last-use opaque.    
\end{lemma}

\begin{proof}
    \begin{align}
    & \text{let}~ C^3_1 = \compl{P^3_1} = P^3_1 \cdot [\res{j}{}{\ab_j}] \\    
    & \hist_3 = C^3_1 \wedge ~\text{\rlemma{lemma:h3-fslop}}~ 
      \so P^3_1 ~\text{is final-state last-use opaque}
    \end{align}
\end{proof}

Let $P^3_2$ be a prefix s.t. $\hist_3 = P^3_2 \cdot [\tryC_j\to\ab_j]$.

\begin{lemma} \label{lemma:p32-fslop}
    $P^3_2$ is final-state last-use opaque.    
\end{lemma}

\begin{proof}
    \begin{align}
    & \text{let}~ C^3_2 = \compl{P^3_2} = P^3_2 \cdot [\tryA_j\to\ab_j] \\
    & \text{let}~ S^3_2 = C^3_2|\tr_i \cdot C^3_2|\tr_j \label{eq:let-s32} \\
    & %
      S^3_2 \equiv C^3_2 \label{eq:s32-equiv}\\
    & \text{real time order}~ \prec_{P^3_2} = \varnothing \label{eq:s32-rt-p32} \\
    & \text{real time order}~ \prec_{S^3_2} = 
      \{ \tr_i \prec_{S^3_2} \tr_j \} \label{eq:p32-rt-s32} \\
    & \pref{eq:s32-rt-p32} \wedge \pref{eq:p32-rt-s32}
      \so \prec_{S^3_2} \subseteq \prec_{P^3_2} \label{eq:p32-rt-sub} \\
    & i = i \so S^3_2|\tr_i \subseteq \luvis{S^3_2}{\tr_i} \label{eq:p32-lvis-ti-ti} \\
    & \tr_i \prec_{S^3_2} \tr_j 
      \so S^3_2|\tr_j \nsubseteq \luvis{S^3_2}{\tr_i} \label{eq:p32-lvis-ti-tj} \\
    & \pref{eq:p32-lvis-ti-ti} \wedge \pref{eq:p32-lvis-ti-tj} 
      \so \luvis{S^3_2}{\tr_i} = S^3_2|\tr_i \label{eq:p32-ti-lvis} \\
    & \pref{eq:p32-ti-lvis} \so \vis{S^3_2}{\tr_i}|\obj = [\fwop{i}{\obj}{1}{\ok_i}] 
      \label{eq:p32-ti-lvis-x} \\
    & \pref{eq:p32-ti-lvis-x} \wedge \pref{eq:seq-w} 
      \so \luvis{S^3_2}{\tr_i} ~\text{is legal} \label{eq:p32-lvis-ti-legal} \\
    & \pref{eq:p32-lvis-ti-legal} \so  
      \tr_i ~\text{in}~ S^3_2 ~\text{is last-use legal in}~ S^3_2 \label{eq:p32-ti-legal} \\
    & \fwop{i}{\obj}{1}{\ok_i} ~\text{is \last{} write on}~\obj~\text{in}~\tr_i
      \so \tr_i ~\text{is decided on}~\obj 
      S^3_2 \label{eq:p32-ti-decided} \\
    & \tr_i \prec_{S^3_2} \tr_j \wedge \pref{eq:p32-ti-decided}
      \so S^3_2\cpeC\tr_i \subseteq \luvis{S^3_2}{\tr_j} \label{eq:p32-lvis-tj-ti} \\
    & j = j \so  S^3_2|\tr_j \subseteq \luvis{S^3_2}{\tr_j} \label{eq:p32-lvis-tj-tj} \\
    & \pref{eq:p32-lvis-tj-ti} \wedge \pref{eq:p32-lvis-tj-tj} 
      \so \luvis{S^3_2}{\tr_j} = S^3_2\cpeC\tr_i \cdot S^3_2|\tr_j \label{eq:p32-tj-lvis} \\
    & \pref{eq:p32-tj-lvis} 
      \so \luvis{S^3_2}{\tr_j}|\obj = [\fwop{i}{\obj}{1}{\ok_i}, \frop{j}{\obj}{1}]
      \label{eq:p32-tj-lvis-x} \\
    & \pref{eq:p32-tj-lvis-x} \wedge \pref{eq:seq-wr} 
      \so \luvis{S^3_2}{\tr_j} ~\text{is legal} \label{eq:p32-lvis-tj-legal} \\
    & \pref{eq:p32-lvis-tj-legal} \so  
      \tr_j ~\text{in}~ S^3_2 ~\text{is last-use legal in}~ S^3_2 \label{eq:p32-tj-legal} \\
    & \pref{eq:s32-equiv} \wedge \pref{eq:p32-rt-sub} 
      \wedge \pref{eq:p32-ti-legal} \wedge \pref{eq:p32-tj-legal} \so 
      P^3_2 ~\text{is final-state last-use opaque}
    \end{align}
\end{proof}

Let $P^3_3$ be a prefix s.t. $\hist_3 = P^3_3 \cdot [\tryC_j\to\ab_j]$.

\begin{lemma} \label{lemma:p33-fslop}
    $P^3_3$ is final-state last-use opaque.    
\end{lemma}

\begin{proof}
    \begin{align}
    & P^3_3 = P^1_4 \wedge ~\text{\rlemma{lemma:p14-fslop}}~ 
      \so P^3_3 ~\text{is final-state last-use opaque}
    \end{align}
\end{proof}

Let $P^3_p$ be a any prefix s.t. $P^3_3$.

\begin{lemma} \label{lemma:p3p-fslop}
    Any $P^3_p$ is final-state last-use opaque.    
\end{lemma}

\begin{proof}
    \begin{align}
    & \hist_1 = P^1_4 \cdot R  \wedge \rcor{cor:h3-pref-lop}
      \so P^1_4 ~\text{is last-use plague} \label{eq:p33-pref-lopaque} \\
    & P^1_4 = P^3_3 \wedge \pref{eq:p33-pref-lopaque} 
      \so P^3_3 ~\text{is last-use lopaque} \label{eq:p33-lop} \\
    & \pref{eq:p33-lop} 
      \so P^3_3 ~\text{is final-state last-use lopaque}   
    \end{align}
\end{proof}

\begin{lemma} \label{lemma:h3-lop}
    $\hist_3$ is last-use opaque.    
\end{lemma}

\begin{proof}
    Since, from Lemmas \ref{lemma:p31-fslop}--\ref{lemma:p3p-fslop}, all
    prefixes of $\hist_3$ are final-state last-use opaque, then by
    \rdef{def:lopacity} $\hist_3$ is last-use opaque.
\end{proof}

\begin{corollary} \label{cor:h3-pref-lop}
    Any prefix of $\hist_3$ is last-use opaque.
\end{corollary}

\begin{figure}
\begin{center}
\begin{tikzpicture}
     \draw
           (0,2)        node[tid]       {$\tr_i$}
                        node[aop]       {$\init_i$} %
                        node[dot]       {} 

      -- ++(1.25,0)     node[aop]       {$\twop{i}{\obj}{1}$}
                        node[dot] (wi)  {}
                        node[cir]       {}

      -- ++(1.25,0)     node[aop]       {$\tryA_i\!\to\!\ab_i$}
                        node[dot]       {}         
                        ;

     \draw
           (1,1)        node[tid]       {$\tr_{j}$}
                        node[aop]       {$\init_{j}$} %
                        node[dot]       {} 

      -- ++(1.25,0)     node[aop]       {$\trop{j}{\obj}{1}$}
                        node[dot] (rj)  {}

      -- ++(1.25,0)     node[aop]       {$\tryC_j\!\to\!\co_j$}
                        node[dot]       {}
                        ;
     
     \draw[hb] (wi) \squiggle (rj);
\end{tikzpicture}
\end{center}
\caption{\label{fig:proof-h4} History $\hist_4$, not last-use opaque.}
\end{figure}
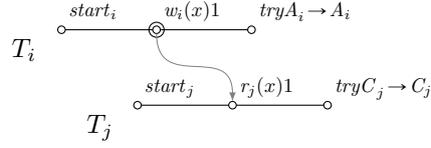

\begin{lemma} \label{lemma:h4-fslop}
    $\hist_4$ is not final-state last-use opaque.    
\end{lemma}

\begin{proof}
    \begin{align}
    & \text{let}~ C_4 = \compl{\hist_4} = \hist_4 \\    
    & \text{let}~ S_4 = C_4|\tr_i \cdot C_4|\tr_j \label{eq:let-s4} \\
    & %
      S_4 \equiv C_4 \label{eq:s4-equiv}\\
    & \text{real time order}~ \prec_{P_4} = \varnothing \label{eq:s4-rt-h4} \\
    & \text{real time order}~ \prec_{S_4} = 
      \{ \tr_i \prec_{S_4} \tr_j \} \label{eq:h4-rt-s4} \\
    & \pref{eq:s4-rt-h4} \wedge \pref{eq:h4-rt-s4}
      \so \prec_{S_4} \subseteq \prec_{P_4} \label{eq:h4-rt-sub} \\
    & \res{i}{}{\ab_i} \in S_4|\tr_i 
      \so S_4|\tr_i \nsubseteq \vis{S_4}{\tr_j} \label{eq:h4-vis-tj-ti} \\
    & j = j \so  S_4|\tr_j \subseteq \vis{S_4}{\tr_j} \label{eq:h4-vis-tj-tj} \\
    & \pref{eq:h4-vis-tj-ti} \wedge \pref{eq:h4-vis-tj-tj} 
      \so \vis{S_4}{\tr_j} = S_4|\tr_j \label{eq:h4-tj-vis} \\
    & \pref{eq:h4-tj-vis} 
      \so \vis{S_4}{\tr_j}|\obj = [\frop{j}{\obj}{1}] 
      \label{eq:h4-tj-vis-x} \\
    & \pref{eq:h4-tj-vis-x} \wedge \pref{eq:ill-seq-r} 
      \so \vis{S_4}{\tr_j} ~\text{is not legal} \label{eq:h4-vis-tj-not-legal} \\
    & \pref{eq:h4-vis-tj-not-legal} \so  
      \tr_j ~\text{in}~ S_4 ~\text{is not legal in}~ S_4 \label{eq:h4-tj-not-legal} \\
    & \text{let}~ \dot{S}_4 = C_4|\tr_j \cdot C_4|\tr_i \label{eq:let-s4b} \\
    & \dot{S}_4 \equiv C_4 \label{eq:s4b-equiv}\\
    & \text{real time order}~ \prec_{P_4} = \varnothing \label{eq:s4b-rt-h4b} \\
    & \text{real time order}~ \prec_{\dot{S}_4} = 
      \{ \tr_i \prec_{\dot{S}_4} \tr_j \} \label{eq:h4b-rt-s4b} \\
    & \pref{eq:s4b-rt-h4b} \wedge \pref{eq:h4b-rt-s4b}
      \so \prec_{\dot{S}_4} \subseteq \prec_{P_4} \label{eq:h4b-rt-sub} \\
    & \tr_j \prec_{\dot{S}_4} \tr_i 
      \so \dot{S}_4|\tr_i \nsubseteq \vis{\dot{S}_4}{\tr_j} \label{eq:h4b-vis-tj-ti} \\
    & j = j \so \dot{S}_4|\tr_j \subseteq \vis{\dot{S}_4}{\tr_i} \label{eq:h4b-vis-tj-tj} \\
    & \pref{eq:h4b-vis-tj-ti} \wedge \pref{eq:h4b-vis-tj-tj} 
      \so \vis{\dot{S}_4}{\tr_j} = \dot{S}_4|\tr_j \label{eq:h4b-tj-vis} \\
    & \pref{eq:h4b-tj-vis} 
      \so \vis{\dot{S}_4}{\tr_j}|\obj = [\frop{j}{\obj}{1}] 
      \label{eq:h4b-tj-vis-x} \\
    & \pref{eq:h4b-tj-vis-x} \wedge \pref{eq:ill-seq-r} 
      \so \vis{\dot{S}_4}{\tr_j} ~\text{is not legal} \label{eq:h4b-vis-tj-not-legal} \\
    & \pref{eq:h4b-vis-tj-not-legal} \so  
      \tr_j ~\text{in}~ \dot{S}_4 ~\text{is not legal in}~ \dot{S}_4 
      \label{eq:h4b-tj-not-legal} \\
    & \pref{eq:h4-tj-not-legal} 
      \wedge \pref{eq:h4b-tj-not-legal} \so 
      P_4 ~\text{is not final-state last-use opaque}
    \end{align}
\end{proof}

\begin{lemma} \label{lemma:h4-lop}
    $\hist_4$ is not last-use opaque.    
\end{lemma}

\begin{proof}
    From \rlemma{lemma:h4-fslop} is not final-state last-use opaque, then so,
    from \rdef{def:lopacity} $\hist_4$ is not last-use opaque.
\end{proof}

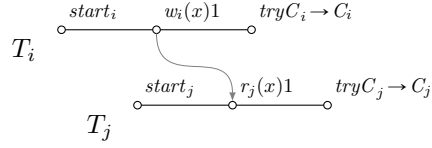
\begin{figure}
\begin{center}
\begin{tikzpicture}
     \draw
           (0,2)        node[tid]       {$\tr_i$}
                        node[aop]       {$\init_i$} %
                        node[dot]       {} 

      -- ++(1.25,0)     node[aop]       {$\twop{i}{\obj}{1}$}
                        node[dot] (wi)  {}

      -- ++(1.25,0)     node[aop]       {$\tryC_i\!\to\!\co_i$}
                        node[dot]       {}         
                        ;

     \draw
           (1,1)        node[tid]       {$\tr_{j}$}
                        node[aop]       {$\init_{j}$} %
                        node[dot]       {} 

      -- ++(1.25,0)     node[aop]       {$\trop{j}{\obj}{1}$}
                        node[dot] (rj)  {}

      -- ++(1.25,0)     node[aop]       {$\tryC_j\!\to\!\co_j$}
                        node[dot]       {}
                        ;
     
     \draw[hb] (wi) \squiggle (rj);
\end{tikzpicture}
\end{center}
\caption{\label{fig:proof-h5} History $\hist_5$, not last-use
opaque. Note that write in $\tr_i$ is not \last{} write.}
\end{figure}

\begin{lemma} \label{lemma:h5-fslop}
    $\hist_5$ is final-state last-use opaque.    
\end{lemma}

\begin{proof}
    \begin{align}
    & \text{let}~ C_5 = \compl{\hist_5} = \hist_5 \\    
    & \text{let}~ S_5 = C_5|\tr_i \cdot C_5|\tr_j \label{eq:let-s5} \\
    & %
      S_5 \equiv C_5 \label{eq:s5-equiv}\\
    & \text{real time order}~ \prec_{\hist_5} = \varnothing \label{eq:h5-rt-h5} \\
    & \text{real time order}~ \prec_{S_5} = 
      \{ \tr_i \prec_{S_5} \tr_j \} \label{eq:h5-rt-s5} \\
    & \pref{eq:h5-rt-h5} \wedge \pref{eq:h5-rt-s5}
      \so \prec_{S_5} \subseteq \prec_{\hist_5} \label{eq:h5-rt-sub} \\
    & S_5 = S_1 
     \!\wedge\! \pref{eq:s5-equiv} \!\wedge\! \pref{eq:h5-rt-sub} 
     \!\wedge\! \pref{eq:h1-ti-legal} \!\wedge\! \pref{eq:h1-tj-legal} 
      \so \hist_5 ~\text{is final-state last-use opaque}
    \end{align}
\end{proof}

Let $P^5_1$ be a prefix s.t. $\hist_5 = P^5_1 \cdot [\res{j}{}{\co_j}]$.

\begin{lemma} \label{lemma:p11-fslop}
    $P^5_1$ is final-state last-use opaque.    
\end{lemma}

\begin{proof}
    \begin{align}
    & \text{let}~ C^5_1 = \compl{P^5_1} = P^5_1 \cdot [\res{j}{}{\co_j}] \\    
    & \hist_5 = C^5_1 \wedge ~\text{\rlemma{lemma:h1-fslop}}~ 
      \so P^5_1 ~\text{is final-state last-use opaque}
    \end{align}
\end{proof}

Let $P^5_2$ be a prefix s.t. $\hist_5 = P^5_2 \cdot [\tryC_j\to\co_j]$.

\begin{lemma} \label{lemma:p52-fslop}
    $P^5_2$ is final-state last-use opaque.    
\end{lemma}

\begin{proof}
    \begin{align}
    & \text{let}~ C^5_2 = \compl{P^5_2} = P^5_2 \cdot [\tryA_j\to\ab_j] \\    
    & \text{let}~ S^5_2 = C^5_2|\tr_i \cdot C^5_2|\tr_j \label{eq:let-s52} \\
    & %
      S^5_2 \equiv C^5_2 \label{eq:s52-equiv}\\
    & \text{real time order}~ \prec_{P^5_2} = \varnothing \label{eq:s52-rt-p52} \\
    & \text{real time order}~ \prec_{S^5_2} = 
      \{ \tr_i \prec_{S^5_2} \tr_j \} \label{eq:p52-rt-s52} \\
    & \pref{eq:s52-rt-p52} \wedge \pref{eq:p52-rt-s52}
      \so \prec_{S^5_2} \subseteq \prec_{P^5_2} \label{eq:p52-rt-sub} \\
    & i = i \so S^5_2|\tr_i \subseteq \vis{S^5_2}{\tr_i} \label{eq:p52-vis-ti-ti} \\
    & \tr_i \prec_{S^5_2} \tr_j 
      \so S^5_2|\tr_j \nsubseteq \vis{S^5_2}{\tr_i} \label{eq:p52-vis-ti-tj} \\
    & \pref{eq:p52-vis-ti-ti} \wedge \pref{eq:p52-vis-ti-tj} 
      \so \vis{S^5_2}{\tr_i} = S^5_2|\tr_i \label{eq:p52-ti-vis} \\
    & \pref{eq:p52-ti-vis} \so \vis{S^5_2}{\tr_i}|\obj = [\fwop{i}{\obj}{1}{\ok_i}] 
      \label{eq:p52-ti-vis-x} \\
    & \pref{eq:p52-ti-vis-x} \wedge \pref{eq:seq-w} 
      \so \vis{S^5_2}{\tr_i} ~\text{is legal} \label{eq:p52-vis-ti-legal} \\
    & \pref{eq:p52-vis-ti-legal} \so  
      \tr_i ~\text{in}~ S^5_2 ~\text{is legal in}~ S^5_2 \label{eq:p52-ti-legal} \\
    & \tr_i \prec_{S^5_2} \tr_j \wedge \res{i}{}{\co_i} \in S^5_2|\tr_i 
      \so S^5_2|\tr_i \subseteq \luvis{S^5_2}{\tr_j} \label{eq:p52-lvis-tj-ti} \\
    & j = j \so  S^5_2|\tr_j \subseteq \luvis{S^5_2}{\tr_j} \label{eq:p52-lvis-tj-tj} \\
    & \pref{eq:p52-lvis-tj-ti} \wedge \pref{eq:p52-lvis-tj-tj} 
      \so \luvis{S^5_2}{\tr_j} = S^5_2|\tr_i \cdot S^5_2|\tr_j \label{eq:p52-tj-lvis} \\
    & \pref{eq:p52-tj-lvis} 
      \so \luvis{S^5_2}{\tr_j}|\obj = [\fwop{i}{\obj}{1}{\ok_i}, \frop{j}{\obj}{1}] 
      \label{eq:p52-tj-lvis-x} \\
    & \pref{eq:p52-tj-lvis-x} \wedge \pref{eq:seq-wr} 
      \so \luvis{S^5_2}{\tr_j} ~\text{is legal} \label{eq:p52-lvis-tj-legal} \\
    & \pref{eq:p52-lvis-tj-legal} \so  
      \tr_j ~\text{in}~ S^5_2 ~\text{is last-use legal in}~ S^5_2 \label{eq:p52-tj-legal} \\
    & \pref{eq:s52-equiv} \wedge \pref{eq:p52-rt-sub} 
      \wedge \pref{eq:p52-ti-legal} \wedge \pref{eq:p52-tj-legal} \so 
      P^5_2 ~\text{is final-state last-use opaque}
    \end{align}
\end{proof}

Let $P^5_3$ be a prefix s.t. $\hist_5 = P^5_3 \cdot [\res{i}{}{\co_i}, \tryC_j\to\co_j]$.

\begin{lemma} \label{lemma:p53-fslop}
    $P^5_3$ is final-state last-use opaque.    
\end{lemma}

\begin{proof}
    \begin{align}
    & \text{let}~ C^5_3 = \compl{P^5_3} = P^5_3 \cdot [\res{i}{}{\co_i}, \tryC_j\to\co_j] \\
    & P^5_2 = C^5_3 \wedge ~\text{\rlemma{lemma:p53-fslop}}~ 
      \so P^5_3 ~\text{is final-state last-use opaque}
    \end{align}
\end{proof}

Let $P^5_4$ be a prefix s.t. $\hist_5 = P^1_5 \cdot [\tryC_i\to\co_i,
\tryC_j\to\co_j]$.

\begin{lemma} \label{lemma:p54-fslop}
    $P^5_4$ is not final-state last-use opaque.    
\end{lemma}

\begin{proof}
    \begin{align}
    & \text{let}~ C^5_4 = \compl{P^5_4} = P^5_4 \cdot [\tryA_i\to\ab_i, \tryA_j\to\ab_j] \\
    & \text{let}~ S^5_4 = C^5_4|\tr_i \cdot C^5_4|\tr_j \label{eq:let-s54} \\
    & %
      S^5_4 \equiv C^5_4 \label{eq:s54-equiv}\\
    & \text{real time order}~ \prec_{P^5_4} = \varnothing \label{eq:s54-rt-p54} \\
    & \text{real time order}~ \prec_{S^5_4} = 
      \{ \tr_i \prec_{S^5_4} \tr_j \} \label{eq:p54-rt-s54} \\
    & \pref{eq:s54-rt-p54} \wedge \pref{eq:p54-rt-s54}
      \so \prec_{S^5_4} \subseteq \prec_{P^5_4} \label{eq:p54-rt-sub} \\
    & \fwop{i}{\obj}{1}{\ok_i} ~\text{is not \last{} write on}~\obj~\text{in}~\tr_i
      \so \tr_i ~\text{is not decided on}~\obj 
      \label{eq:p54-ti-decided} \\
    & \tr_i \prec_{S^5_4} \tr_j \wedge \pref{eq:p54-ti-decided}
      \so S^5_4\cpeC\tr_i \nsubseteq \luvis{S^5_4}{\tr_j} 
      \label{eq:p54-lvis-tj-ti} \\
    & j = j \so  S^5_4|\tr_j \subseteq \luvis{S^5_4}{\tr_j} 
      \label{eq:p54-lvis-tj-tj} \\
    & \pref{eq:p54-lvis-tj-ti} \wedge \pref{eq:p54-lvis-tj-tj} 
      \so \luvis{S^5_4}{\tr_j} = S^5_4\cpeC\tr_i \cdot S^5_4|\tr_j 
      \label{eq:p54-tj-lvis} \\
    & \pref{eq:p54-tj-lvis} 
      \so \luvis{S^5_4}{\tr_j}|\obj = [\frop{j}{\obj}{1}]
      \label{eq:p54-tj-lvis-x} \\
    & \pref{eq:p54-tj-lvis-x} \wedge \pref{eq:seq-wr} 
      \so \luvis{S^5_4}{\tr_j} ~\text{is not legal} \label{eq:p54-lvis-tj-not-legal} \\
    & \pref{eq:p54-lvis-tj-not-legal} \so  
      \tr_j ~\text{in}~ S^5_4 ~\text{is not last-use legal in}~ S^5_4 \label{eq:p54-tj-not-legal} \\
    & \text{let}~ \dot{S}^5_4 = C^5_4|\tr_j \cdot C^5_4|\tr_i \label{eq:let-s54b} \\
    & \dot{S}^5_4 \equiv C^5_4 \label{eq:s54b-equiv}\\
    & \text{real time order}~ \prec_{P^5_4} = \varnothing \label{eq:s54b-rt-p54b} \\
    & \text{real time order}~ \prec_{\dot{S}^5_4} = 
      \{ \tr_i \prec_{\dot{S}^5_4} \tr_j \} \label{eq:p54b-rt-s54b} \\
    & \pref{eq:s54b-rt-p54b} \wedge \pref{eq:p54b-rt-s54b}
      \so \prec_{\dot{S}^5_4} \subseteq \prec_{P^5_4} \label{eq:p54b-rt-sub} \\
    & \tr_j \prec_{\dot{S}^5_4} \tr_i 
      \so \dot{S}^5_4|\tr_i \nsubseteq \luvis{\dot{S}^5_4}{\tr_j} \label{eq:p54b-lvis-tj-ti} \\
    & j = j \so  \dot{S}^5_4|\tr_j \subseteq \luvis{\dot{S}^5_4}{\tr_i} \label{eq:p54b-lvis-tj-tj} \\
    & \pref{eq:p54b-lvis-tj-ti} \wedge \pref{eq:p54b-lvis-tj-tj} 
      \so \luvis{\dot{S}^5_4}{\tr_j} = \dot{S}^5_4|\tr_j \label{eq:p54b-tj-lvis} \\
    & \pref{eq:p54b-tj-lvis} 
      \so \luvis{\dot{S}^5_4}{\tr_j}|\obj = [\frop{j}{\obj}{1}] 
      \label{eq:p54b-tj-lvis-x} \\
    & \pref{eq:p54b-tj-lvis-x} \wedge \pref{eq:ill-seq-r} 
      \so \luvis{\dot{S}^5_4}{\tr_j} ~\text{is not legal} \label{eq:p54b-lvis-tj-not-legal} \\
    & \pref{eq:p54b-lvis-tj-not-legal} \so  
      \tr_j ~\text{in}~ \dot{S}^5_4 ~\text{is not legal in}~ \dot{S}^5_4 
      \label{eq:p54b-tj-not-legal} \\
    & \pref{eq:p54-tj-not-legal} 
      \wedge \pref{eq:p54b-tj-not-legal} \so 
      P^2_2 ~\text{is not final-state last-use opaque}
    \end{align}
\end{proof}

\begin{lemma} \label{lemma:h5-lop}
    $\hist_5$ is not last-use opaque.    
\end{lemma}

\begin{proof}
    Even though, from \rlemma{lemma:h5-fslop}, $\hist_5$ is final-state
    last-use opaque, from \rlemma{lemma:p54-fslop}, prefix $P^5_4$ of $\hist_5$
    is not final-state last-use opaque, so, from \rdef{def:lopacity} $\hist_5$
    is not last-use opaque.
\end{proof}

\begin{figure}
\begin{center}
\begin{tikzpicture}
     \draw
           (0,2)        node[tid]       {$\tr_i$}
                        node[aop]       {$\init_i$} %
                        node[dot]       {} 

      -- ++(1.25,0)     node[aop]       {$\twop{i}{\obj}{1}$}
                        node[dot] (wi)  {}

      -- ++(1.25,0)     node[aop]       {$\twop{i}{\obj}{2}$}
                        node[dot]       {}

      -- ++(1.25,0)     node[aop]       {$\tryC_i\!\to\!\co_i$}
                        node[dot]       {}         
                        ;

     \draw
           (1,1)        node[tid]       {$\tr_{j}$}
                        node[aop]       {$\init_{j}$} %
                        node[dot]       {} 

      -- ++(1.25,0)     node[aop]       {$\trop{j}{\obj}{1}$}
                        node[dot] (rj)  {}

      -- ++(1.25,0)     node[aop]       {$\tryA_j\!\to\!\ab_j$}
                        node[dot]       {}
                        ;
     
     \draw[hb] (wi) \squiggle (rj);
\end{tikzpicture}
\end{center}
\caption{\label{fig:proof-h6} History $\hist_6$, not last-use
opaque.}
\end{figure}
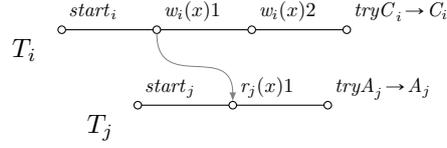

\begin{lemma} \label{lemma:h6-fslop}
    $\hist_6$ is not final-state last-use opaque.    
\end{lemma}

\begin{proof}
    \begin{align}
    & \text{let}~ C_6 = \compl{\hist_6} = \hist_6 \\
    & \text{let}~ S_6 = C_6|\tr_i \cdot C_6|\tr_j \label{eq:let-s6} \\
    & %
      S_6 \equiv C_6 \label{eq:s6-equiv}\\
    & \text{real time order}~ \prec_{\hist_6} = \varnothing \label{eq:s6-rt-h6} \\
    & \text{real time order}~ \prec_{S_6} = 
      \{ \tr_i \prec_{S_6} \tr_j \} \label{eq:h6-rt-s6} \\
    & \pref{eq:s6-rt-h6} \wedge \pref{eq:h6-rt-s6}
      \so \prec_{S_6} \subseteq \prec_{\hist_6} \label{eq:h6-rt-sub} \\
    & \tr_i \prec_{S_6} \tr_j \wedge \res{i}{}{\co_i} \in S_6|\tr_i 
      \so S_6|\tr_i \subseteq \luvis{S_6}{\tr_i} \label{eq:h6-lvis-tj-ti} \\
    & j = j \so  S_6|\tr_j \subseteq \luvis{S_6}{\tr_i} \label{eq:h6-lvis-tj-tj}\\
    & \pref{eq:h6-lvis-tj-ti} \wedge \pref{eq:h6-lvis-tj-tj} 
      \so \vis{S_6}{\tr_j} = S_6|\tr_i \cdot S_6|\tr_j \label{eq:s6-tj-lvis} \\
    & \pref{eq:s6-tj-lvis} 
      \so \vis{S_6}{\tr_j}|\obj = 
      [\fwop{i}{\obj}{1}{\ok_i},\fwop{i}{\obj}{2}{\ok_i}, \frop{j}{\obj}{1}] 
      \label{eq:s6-tj-lvis-x} \\
    & \pref{eq:s6-tj-lvis-x} \wedge \pref{eq:ill-seq-wwr} 
      \so \vis{S_6}{\tr_j} ~\text{is not legal} \label{eq:h6-lvis-tj-not-legal} \\
    & \pref{eq:h6-lvis-tj-not-legal} \so  
      \tr_j ~\text{in}~ S_6 ~\text{is not legal in}~ S_6 
      \label{eq:h6-tj-not-legal} \\
    & \text{let}~\dot{S}_6 = C_6|\tr_j \cdot C_6|\tr_i \label{eq:let-s6b}\\
    & \dot{S}_6 \equiv C_6 \label{eq:s6b-equiv}\\
    & \text{real time order}~ \prec_{\hist_6} = \varnothing 
      \label{eq:s6b-rt-h6b} \\
    & \text{real time order}~ \prec_{\dot{S}_6} = 
      \{ \tr_i \prec_{\dot{S}_6} \tr_j \} \label{eq:h6b-rt-s6b} \\
    & \pref{eq:s6b-rt-h6b} \wedge \pref{eq:h6b-rt-s6b}
      \so \prec_{\dot{S}_6} \subseteq \prec_{\hist_6} 
      \label{eq:h6b-rt-sub} \\
    & \tr_j \prec_{\dot{S}_6} \tr_i 
      \so \dot{S}_6|\tr_i \nsubseteq \luvis{\dot{S}_6}{\tr_j} 
      \label{eq:h6b-lvis-tj-ti} \\
    & j = j \so  \dot{S}_6|\tr_j \subseteq \luvis{\dot{S}_6}{\tr_i} 
      \label{eq:h6b-lvis-tj-tj} \\
    & \pref{eq:h6b-lvis-tj-ti} \wedge \pref{eq:h6b-lvis-tj-tj} 
      \so \luvis{\dot{S}_6}{\tr_j} = \dot{S}_6|\tr_j \label{eq:h6b-tj-lvis} \\
    & \pref{eq:h6b-tj-lvis} 
      \so \luvis{\dot{S}_6}{\tr_j}|\obj = [\frop{j}{\obj}{1}] 
      \label{eq:h6b-tj-lvis-x} \\
    & \pref{eq:h6b-tj-lvis-x} \wedge \pref{eq:ill-seq-r} 
      \so \luvis{\dot{S}_6}{\tr_j} ~\text{is not legal} 
      \label{eq:h6b-lvis-tj-not-legal} \\
    & \pref{eq:h6b-lvis-tj-not-legal} \so  
      \tr_j ~\text{in}~ \dot{S}_6 ~\text{is not legal in}~ \dot{S}_6 
      \label{eq:h6b-tj-not-legal} \\
    & \pref{eq:h6-tj-not-legal} 
      \wedge \pref{eq:h6b-tj-not-legal} \so 
      \hist_6 ~\text{is not final-state last-use opaque}
    \end{align}
\end{proof}

\begin{lemma} \label{lemma:h6-lop}
    $\hist_6$ is not last-use opaque.    
\end{lemma}

\begin{proof}
    From \rlemma{lemma:h6-fslop} is not final-state last-use opaque, then so,
    from \rdef{def:lopacity} $\hist_6$ is not last-use opaque.
\end{proof}

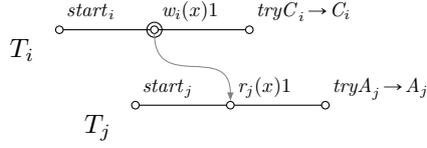
\begin{figure}
\begin{center}
\begin{tikzpicture}
     \draw
           (0,2)        node[tid]       {$\tr_i$}
                        node[aop]       {$\init_i$} %
                        node[dot]       {} 

      -- ++(1.25,0)     node[aop]       {$\twop{i}{\obj}{1}$}
                        node[dot] (wi)  {}
                        node[cir]       {}

      -- ++(1.25,0)     node[aop]       {$\tryC_i\!\to\!\co_i$}
                        node[dot]       {}         
                        ;

     \draw
           (1,1)        node[tid]       {$\tr_{j}$}
                        node[aop]       {$\init_{j}$} %
                        node[dot]       {} 

      -- ++(1.25,0)     node[aop]       {$\trop{j}{\obj}{1}$}
                        node[dot] (rj)  {}

      -- ++(1.25,0)     node[aop]       {$\tryA_j\!\to\!\ab_j$}
                        node[dot]       {}
                        ;
     
     \draw[hb] (wi) \squiggle (rj);
\end{tikzpicture}
\end{center}
\caption{\label{fig:proof-h7} History $\hist_7$, last-use
opaque.}
\end{figure}

\begin{lemma} \label{lemma:h7-fslop}
    $\hist_7$ is final-state last-use opaque.    
\end{lemma}

\begin{proof}
    \begin{align}
    & \text{let}~ C_7 = \compl{\hist_7} = \hist_7 \\
    & \text{let}~ S_7 = C_7|\tr_i \cdot C_7|\tr_j \label{eq:let-s7} \\
    & %
      S_7 \equiv C_7 \label{eq:s7-equiv}\\
    & \text{real time order}~\prec_{\hist_7}=\varnothing \label{eq:s7-rt-h7} \\
    & \text{real time order}~ \prec_{S_7} = 
      \{ \tr_i \prec_{S_7} \tr_j \} \label{eq:h7-rt-s7} \\
    & \pref{eq:s7-rt-h7} \wedge \pref{eq:h7-rt-s7}
      \so \prec_{S_7} \subseteq \prec_{\hist_7} \label{eq:h7-rt-sub} \\
    & i = i \so S_7|\tr_i \subseteq \vis{S_7}{\tr_i} \label{eq:h7-vis-ti-ti} \\
    & \tr_i \prec_{S_7} \tr_j 
      \so S_7|\tr_j \nsubseteq \vis{S_7}{\tr_i} \label{eq:h7-vis-ti-tj} \\
    & \pref{eq:h7-vis-ti-ti} \wedge \pref{eq:h7-vis-ti-tj} 
      \so \vis{S_7}{\tr_i} = S_7|\tr_i \label{eq:h7-ti-vis} \\
    & \pref{eq:h7-ti-vis} \so \vis{S_7}{\tr_i}|\obj = [\fwop{i}{\obj}{1}{\ok_i}] 
      \label{eq:h7-ti-vis-x} \\
    & \pref{eq:h7-ti-vis-x} \wedge \pref{eq:seq-w} 
      \so \luvis{S_7}{\tr_i} ~\text{is legal} \label{eq:h7-lvis-ti-legal} \\
    & \pref{eq:h7-lvis-ti-legal} \so  
      \tr_i ~\text{in}~ S_7 ~\text{is last-use legal in}~ S_7 
      \label{eq:h7-ti-legal} \\
    & \tr_i \prec_{S_7} \tr_j \wedge \res{i}{}{\co_i} \in S_7|\tr_i 
      \so S_7|\tr_i \subseteq \vis{S_7}{\tr_i} \label{eq:h7-lvis-tj-ti} \\
    & j = j \so  S_7|\tr_j \subseteq \vis{S_7}{\tr_j} \label{eq:h7-lvis-tj-tj} \\
    & \pref{eq:h7-lvis-tj-ti} \wedge \pref{eq:h7-lvis-tj-tj} 
      \so \vis{S_7}{\tr_j} = S_7|\tr_i \cdot S_7|\tr_j 
      \label{eq:h7-tj-lvis} \\
    & \pref{eq:h7-tj-lvis} 
      \so \luvis{S_7}{\tr_j}|\obj = [\fwop{i}{\obj}{1}{\ok_i}, \frop{j}{\obj}{1}]
      \label{eq:h7-tj-lvis-x} \\
    & \pref{eq:h7-tj-lvis-x} \wedge \pref{eq:seq-wr} 
      \so \luvis{S_7}{\tr_j} ~\text{is legal} \label{eq:h7-lvis-tj-legal} \\
    & \pref{eq:h7-lvis-tj-legal} \so  
      \tr_j ~\text{in}~ S_7 ~\text{is last-use legal in}~ S_7 
      \label{eq:h7-tj-legal} \\
    & \pref{eq:s7-equiv} \wedge \pref{eq:h7-rt-sub} 
      \wedge \pref{eq:h7-ti-legal} \wedge \pref{eq:h7-tj-legal} \so 
      \hist_7 ~\text{is final-state last-use opaque}
    \end{align}
\end{proof}

Let $P^7_1$ be a prefix s.t. $\hist_7 = P^7_1 \cdot [\res{i}{}{\co_i}]$.

\begin{lemma} \label{lemma:p71-fslop}
    $P^7_1$ is final-state last-use opaque.    
\end{lemma}

\begin{proof}
    \begin{align}
    & \text{let}~ C^7_1 = \compl{P^7_1} = P^7_1 \cdot [\res{i}{}{\co_i}] \\    
    & \hist_7 = C^7_1 \wedge ~\text{\rlemma{lemma:h7-fslop}}~ 
      \so P^7_1 ~\text{is final-state last-use opaque}
    \end{align}
\end{proof}

Let $P^7_2$ be a prefix s.t. $\hist_7 = P^7_2 \cdot [\tryC_i\to\co_i]$.

\begin{proof}
    \begin{align}
    & \text{let}~ C^7_2 = \compl{P^7_2} = P^7_2 \cdot [\tryA_i\to\ab_i] \\
    & \text{let}~ S^7_2 = C^7_2|\tr_i \cdot C^7_2|\tr_j \label{eq:let-s72} \\
    & %
      S^7_2 \equiv C^7_2 \label{eq:s72-equiv}\\
    & \text{real time order}~ \prec_{P^7_2} = \varnothing \label{eq:s72-rt-p72} \\
    & \text{real time order}~ \prec_{S^7_2} = 
      \{ \tr_i \prec_{S^7_2} \tr_j \} \label{eq:p72-rt-s72} \\
    & \pref{eq:s72-rt-p72} \wedge \pref{eq:p72-rt-s72}
      \so \prec_{S^7_2} \subseteq \prec_{P^7_2} \label{eq:p72-rt-sub} \\
    & i = i \so S^7_2|\tr_i \subseteq \luvis{S^7_2}{\tr_i} \label{eq:p72-lvis-ti-ti} \\
    & \tr_i \prec_{S^7_2} \tr_j 
      \so S^7_2|\tr_j \nsubseteq \luvis{S^7_2}{\tr_i} \label{eq:p72-lvis-ti-tj} \\
    & \pref{eq:p72-lvis-ti-ti} \wedge \pref{eq:p72-lvis-ti-tj} 
      \so \luvis{S^7_2}{\tr_i} = S^7_2|\tr_i \label{eq:p72-ti-lvis} \\
    & \pref{eq:p72-ti-lvis} \so \vis{S^7_2}{\tr_i}|\obj = [\fwop{i}{\obj}{1}{\ok_i}] 
      \label{eq:p72-ti-lvis-x} \\
    & \pref{eq:p72-ti-lvis-x} \wedge \pref{eq:seq-w} 
      \so \luvis{S^7_2}{\tr_i} ~\text{is legal} \label{eq:p72-lvis-ti-legal} \\
    & \pref{eq:p72-lvis-ti-legal} \so  
      \tr_i ~\text{in}~ S^7_2 ~\text{is last-use legal in}~ S^7_2 \label{eq:p72-ti-legal} \\
    & \fwop{i}{\obj}{1}{\ok_i} ~\text{is \last{} write on}~\obj~\text{in}~\tr_i
      \so \tr_i ~\text{is decided on}~\obj \label{eq:p72-ti-decided} \\
    & \tr_i \prec_{S^7_2} \tr_j \wedge \pref{eq:p72-ti-decided}
      \so S^7_2\cpeC\tr_i \subseteq \luvis{S^7_2}{\tr_j} \label{eq:p72-lvis-tj-ti} \\
    & j = j \so  S^7_2|\tr_j \subseteq \luvis{S^7_2}{\tr_j} \label{eq:p72-lvis-tj-tj} \\
    & \pref{eq:p72-lvis-tj-ti} \wedge \pref{eq:p72-lvis-tj-tj} 
      \so \luvis{S^7_2}{\tr_j} = S^7_2\cpeC\tr_i \cdot S^7_2|\tr_j \label{eq:p72-tj-lvis} \\
    & \pref{eq:p72-tj-lvis} 
      \so \luvis{S^7_2}{\tr_j}|\obj = [\fwop{i}{\obj}{1}{\ok_i} ] 
      \label{eq:p72-tj-lvis-x} \\
    & \pref{eq:p72-tj-lvis-x} \wedge \pref{eq:seq-w} 
      \so \luvis{S^7_2}{\tr_j} ~\text{is legal} \label{eq:p72-lvis-tj-legal} \\
    & \pref{eq:p72-lvis-tj-legal} \so  
      \tr_j ~\text{in}~ S^7_2 ~\text{is last-use legal in}~ S^7_2 \label{eq:p72-tj-legal} \\
    & \pref{eq:s72-equiv} \wedge \pref{eq:p72-rt-sub} 
      \wedge \pref{eq:p72-ti-legal} \wedge \pref{eq:p72-tj-legal} \so 
      P^7_2 ~\text{is final-state last-use opaque}
    \end{align}
\end{proof}

\begin{lemma} \label{lemma:p73-fslop}
    $P^7_3$ is final-state last-use opaque.    
\end{lemma}

\begin{proof}
    \begin{align}
    & \text{let}~ C^7_3 = \compl{P^7_3} = 
      P^7_3 \cdot [\res{j}{}{\ab_j},\tryC_i\to\co_i] \\    
    & P^7_3 = C^7_3 \wedge ~\text{\rlemma{lemma:p73-fslop}}~ 
      \so P^7_3 ~\text{is final-state last-use opaque}
    \end{align}
\end{proof}

Let $P^7_4$ be a prefix s.t. $\hist_3 = P^7_4 \cdot [\tryC_j\to\ab_j,\tryC_i\to\co_i]$.

\begin{lemma} \label{lemma:p74-fslop}
    $P^7_4$ is final-state last-use opaque.    
\end{lemma}

\begin{proof}
    \begin{align}
    & P^7_4 = P^5_4 \wedge ~\text{\rlemma{lemma:p54-fslop}}~ 
      \so P^7_4 ~\text{is final-state last-use opaque}
    \end{align}
\end{proof}

Let $P^7_p$ be a any prefix of $P^7_4$.

\begin{lemma} \label{lemma:p7p-fslop}
    Any $P^7_p$ is final-state last-use opaque.    
\end{lemma}

\begin{proof}
    \begin{align}
    & \hist_1 = P^1_4 \cdot R  \wedge \rcor{cor:h1-pref-lop}
      \so P^1_4 ~\text{is last-use lopaque} \label{eq:p74-pref-lopaque} \\
    & P^1_4 = P^7_4 \wedge \pref{eq:p74-pref-lopaque} 
      \so P^7_4 ~\text{is last-use lopaque} \label{eq:p74-lop} \\
    & \pref{eq:p74-lop} 
      \so P^7_p ~\text{is final-state last-use lopaque}   
    \end{align}
\end{proof}

\begin{lemma} \label{lemma:h3-lop}
    $\hist_7$ is last-use opaque.    
\end{lemma}

\begin{proof}
    Since, from Lemmas \ref{lemma:p71-fslop}--\ref{lemma:p7p-fslop}, all
    prefixes of $\hist_7$ are final-state last-use opaque, then by
    \rdef{def:lopacity} $\hist_7$ is last-use opaque.
\end{proof}

\begin{figure}
\begin{center}
\begin{tikzpicture}
     \draw
           (0,2)        node[tid]       {$\tr_i$}
                        node[aop]       {$\init_i$} %
                        node[dot]       {} 

      -- ++(1.25,0)     node[aop]       {$\twop{i}{\objx}{1}$}
                        node[dot] (wi)  {}
                        node[cir]       {}

      -- ++(2.25,0)     node[aop]       {$\trop{i}{\objy}{1}$}
                        node[dot] (ri)  {}

      -- ++(1.25,0)     node[]{} %
                        ;

     \draw
           (0,1)        node[tid]       {$\tr_{j}$}
                        node[aop]       {$\init_{j}$} %
                        node[dot]       {} 

      -- ++(1.75,0)      node[aop]       {$\trop{j}{\objx}{1}$}
                        node[dot] (rj)  {}

      -- ++(1.25,0)     node[aop]       {$\twop{j}{\objy}{1}$}
                        node[dot] (wj)  {}  
                        node[cir]       {}                     

     -- ++(1.75,0)      node[]{} %
                        ;
     
     \draw[hb] (wi) \squiggle (rj);
     \draw[hb] (wj) \squiggleup (ri);
\end{tikzpicture}
\end{center}
\caption{\label{fig:proof-h8}History $\hist_8$, not last-use opaque.}
\end{figure}
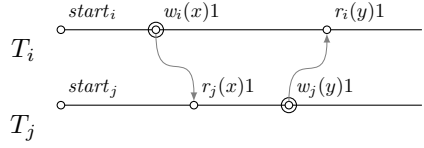

\begin{lemma} \label{lemma:h8-fslop}
    $\hist_8$ is not final-state last-use opaque.    
\end{lemma}

\begin{proof}
    \begin{align}
    & \text{let}~ C_8 = \compl{\hist_8} = \hist_8 \\
    & \text{let}~ S_8 = C_8|\tr_i \cdot C_8|\tr_j \label{eq:let-s8} \\
    & %
      S_8 \equiv C_8 \label{eq:s8-equiv}\\
    & \text{real time order}~ \prec_{\hist_8} = \varnothing \label{eq:s8-rt-h8} \\
    & \text{real time order}~ \prec_{S_8} = 
      \{ \tr_i \prec_{S_8} \tr_j \} \label{eq:h8-rt-s8} \\
    & \pref{eq:s8-rt-h8} \wedge \pref{eq:h8-rt-s8}
      \so \prec_{S_8} \subseteq \prec_{\hist_8} \label{eq:h8-rt-sub} \\
    & \tr_i \prec_{S_8} \tr_j \wedge \res{i}{}{\co_i} \in S_8|\tr_i 
      \so S_8|\tr_i \subseteq \luvis{S_8}{\tr_i} \label{eq:h8-lvis-tj-ti} \\
    & j = j \so  S_8|\tr_j \subseteq \luvis{S_8}{\tr_i} \label{eq:h8-lvis-tj-tj}\\
    & \pref{eq:h8-lvis-tj-ti} \wedge \pref{eq:h8-lvis-tj-tj} 
      \so \vis{S_8}{\tr_j} = S_8|\tr_i \cdot S_8|\tr_j \label{eq:s8-tj-lvis} \\
    & \pref{eq:s8-tj-lvis} 
      \so \vis{S_8}{\tr_j}|\obj = 
            [\fwop{i}{\objx}{1}{\ok_i},\frop{i}{\objy}{1},
             \frop{j}{\objx}{1},\fwop{j}{\objy}{1}{\ok_j}] 
      \label{eq:s8-tj-lvis-x} \\
    & \pref{eq:s8-tj-lvis-x} \wedge \pref{eq:ill-seq-wrrw} 
      \so \vis{S_8}{\tr_j} ~\text{is not legal} \label{eq:h8-lvis-tj-not-legal} \\
    & \pref{eq:h8-lvis-tj-not-legal} \so  
      \tr_j ~\text{in}~ S_8 ~\text{is not legal in}~ S_8 
      \label{eq:h8-tj-not-legal} \\
    & \text{let}~\dot{S}_8 = C_8|\tr_j \cdot C_8|\tr_i \label{eq:let-s8b}\\
    & \dot{S}_8 \equiv C_8 \label{eq:s8b-equiv}\\
    & \text{real time order}~ \prec_{\hist_8} = \varnothing 
      \label{eq:s8b-rt-h8b} \\
    & \text{real time order}~ \prec_{\dot{S}_8} = 
      \{ \tr_i \prec_{\dot{S}_8} \tr_j \} \label{eq:h8b-rt-s8b} \\
    & \pref{eq:s8b-rt-h8b} \wedge \pref{eq:h8b-rt-s8b}
      \so \prec_{\dot{S}_8} \subseteq \prec_{\hist_8} 
      \label{eq:h8b-rt-sub} \\
    & \tr_j \prec_{\dot{S}_8} \tr_i 
      \so \dot{S}_8|\tr_i \nsubseteq \luvis{\dot{S}_8}{\tr_j} 
      \label{eq:h8b-lvis-tj-ti} \\
    & j = j \so  \dot{S}_8|\tr_j \subseteq \luvis{\dot{S}_8}{\tr_i} 
      \label{eq:h8b-lvis-tj-tj} \\
    & \pref{eq:h8b-lvis-tj-ti} \wedge \pref{eq:h8b-lvis-tj-tj} 
      \so \luvis{\dot{S}_8}{\tr_j} = \dot{S}_8|\tr_j \label{eq:h8b-tj-lvis} \\
    & \pref{eq:h8b-tj-lvis} 
      \so \luvis{\dot{S}_8}{\tr_j}|\obj = 
            [\fwop{j}{\objy}{1}{\ok_i},\frop{j}{\objx}{1},
             \frop{i}{\objy}{1},\fwop{i}{\objx}{1}{\ok_i}] 
      \label{eq:h8b-tj-lvis-x} \\
    & \pref{eq:h8b-tj-lvis-x} \wedge \pref{eq:ill-seq-wrrw} 
      \so \luvis{\dot{S}_8}{\tr_j} ~\text{is not legal} 
      \label{eq:h8b-lvis-tj-not-legal} \\
    & \pref{eq:h8b-lvis-tj-not-legal} \so  
      \tr_j ~\text{in}~ \dot{S}_8 ~\text{is not legal in}~ \dot{S}_8 
      \label{eq:h8b-tj-not-legal} \\
    & \pref{eq:h8-tj-not-legal} 
      \wedge \pref{eq:h8b-tj-not-legal} \so 
      \hist_8 ~\text{is not final-state last-use opaque}
    \end{align}
\end{proof}

\begin{lemma} \label{lemma:h8-lop}
    $\hist_8$ is not last-use opaque.    
\end{lemma}

\begin{proof}
    From \rlemma{lemma:h8-fslop} is not final-state last-use opaque, then so,
    from \rdef{def:lopacity} $\hist_8$ is not last-use opaque.
\end{proof}

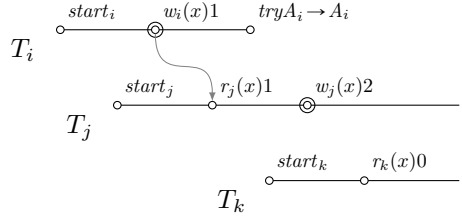
\begin{figure}
\begin{center}
\begin{tikzpicture}
     \draw
           (0,2)        node[tid]       {$\tr_i$}
                        node[aop]       {$\init_i$} %
                        node[dot]       {} 

      -- ++(1.25,0)     node[aop]       {$\twop{i}{\obj}{1}$}
                        node[dot] (wi)  {}
                        node[cir]       {}

      -- ++(1.25,0)      node[aop]       {$\tryA_i\!\to\!\ab_i$}
                        node[dot]       {}         
                        ;

     \draw
           (0.75,1)     node[tid]       {$\tr_{j}$}
                        node[aop]       {$\init_{j}$} %
                        node[dot]       {} 

      -- ++(1.25,0)     node[aop]       {$\trop{j}{\obj}{1}$}
                        node[dot] (rj)  {}

      -- ++(1.25,0)     node[aop]       {$\twop{j}{\obj}{2}$}
                        node[dot]       {}
                        node[cir]       {}                       

      -- ++(2,0)     node            {}       
                        ;

      \draw
           (2.75,0)   node[tid]       {$\tr_{k}$}
                        node[aop]       {$\init_{k}$} %
                        node[dot]       {} 

      -- ++(1.25,0)        node[aop]       {$\trop{k}{\obj}{0}$}
                        node[dot]       {}

      -- ++(1.25,0)     node            {}

                        ;                       
     
     \draw[hb] (wi) \squiggle (rj);
\end{tikzpicture}
\end{center}
\caption{\label{fig:h9-lop}
    $\hist_9$, last-use opaque.    
}
\end{figure}

\begin{lemma}
    $\hist_9$ is final-state last-use opaque.    
\end{lemma}

\begin{proof}
    \begin{align}
    & \text{let}~ C_9 = \compl{\hist_9} =  \hist^9 \cdot [\tryA_j\to\ab_j, \tryA_k\to{\ab_k}] \\
    & \text{let}~ S_9 = C_9|\tr_i \cdot C_9|\tr_j \cdot C_9|\tr_k \label{eq:let-s9} \\
    & %
      S_9 \equiv C_9 \label{eq:s9-equiv}\\
    & \text{real time order}~ \prec_{\hist_9} =  \{ \tr_i \prec_{\hist_9} \tr_k \}\label{eq:s9-rt-h9} \\
    & \text{real time order}~ \prec_{S_9} = 
      \{ \tr_i \prec_{S_9} \tr_j, \tr_i  \prec_{S_9} \tr_k, \tr_j  \prec_{S_9} \tr_k  \} \label{eq:h9-rt-s9} \\
    & \pref{eq:s9-rt-h9} \wedge \pref{eq:h9-rt-s9}
      \so \prec_{S_9} \subseteq \prec_{\hist_9} \label{eq:h9-rt-sub} \\
    & i = i \so S_9|\tr_i \subseteq \luvis{S_9}{\tr_i} \label{eq:h9-lvis-ti-ti} \\
    & \tr_i \prec_{S_9} \tr_j 
      \so S_9|\tr_j \nsubseteq \luvis{S_9}{\tr_i} \label{eq:h9-lvis-ti-tj} \\
    & \tr_i \prec_{S_9} \tr_k 
      \so S_9|\tr_k \nsubseteq \luvis{S_9}{\tr_i} \label{eq:h9-lvis-ti-tk} \\
    & \pref{eq:h9-lvis-ti-ti} \wedge \pref{eq:h9-lvis-ti-tj} \wedge \pref{eq:h9-lvis-ti-tk} 
      \so \luvis{S_9}{\tr_i} = S_9|\tr_i \label{eq:h9-ti-lvis} \\
    & \pref{eq:h9-ti-lvis} \so \luvis{S_9}{\tr_i}|\obj = [\fwop{i}{\obj}{1}{\ok_i}] 
      \label{eq:h9-ti-lvis-x} \\
    & \pref{eq:h9-ti-lvis-x} \wedge \pref{eq:seq-w} 
      \so \luvis{S_9}{\tr_i} ~\text{is legal} \label{eq:h9-lvis-ti-legal} \\
    & \pref{eq:h9-lvis-ti-legal} \so  
      \tr_i ~\text{in}~ S_9 ~\text{is last-use legal in}~ S_9 \label{eq:h9-ti-legal} \\
    & \fwop{i}{\obj}{1}{\ok_i} ~\text{is \last{} write on}~\obj~\text{in}~\tr_i
      \so \tr_i ~\text{is decided on}~\obj ~\text{in}~
      S_9 \label{eq:h9-ti-decided} \\
    & \tr_i \prec_{S_9} \tr_j \wedge \pref{eq:h9-ti-decided}
      \so S_9\cpeC\tr_i \subseteq \luvis{S_9}{\tr_j} \label{eq:h9-lvis-tj-ti} \\
    & j = j \so  S_9|\tr_j \subseteq \luvis{S_9}{\tr_j} \label{eq:h9-lvis-tj-tj} \\
    & \tr_j \prec_{S_9} \tr_k 
      \so S_9|\tr_k \nsubseteq \luvis{S_9}{\tr_j} \label{eq:h9-lvis-tj-tk} \\
    & \pref{eq:h9-lvis-tj-ti} \wedge \pref{eq:h9-lvis-tj-tj} \wedge \pref{eq:h9-lvis-tj-tk} 
      \so \luvis{S_9}{\tr_j} = S_9\cpeC\tr_i \cdot S_9|\tr_j \label{eq:h9-tj-lvis} \\
    & \pref{eq:h9-tj-lvis} 
      \so \luvis{S_9}{\tr_j}|\obj = [\fwop{i}{\obj}{1}{\ok_i}, \frop{j}{\obj}{1}, \fwop{j}{\obj}{2}{\ok_j}]
      \label{eq:h9-tj-lvis-x} \\
    & \pref{eq:h9-tj-lvis-x} \wedge \pref{eq:seq-wrw} 
      \so \luvis{S_9}{\tr_j} ~\text{is legal} \label{eq:h9-lvis-tj-legal} \\
    & \pref{eq:h9-lvis-tj-legal} \so  
      \tr_j ~\text{in}~ S_9 ~\text{is last-use legal in}~ S_9 \label{eq:h9-tj-legal} \\
    \end{align}
    \begin{align}
      & \tr_i \prec_{S_9} \tr_k \wedge \pref{eq:h9-ti-decided}  
      \so S_9\cpeC\tr_i \subseteq \luvis{S_9}{\tr_k}~\text{or}~
          S_9\cpeC\tr_i \nsubseteq \luvis{S_9}{\tr_k} \label{eq:h9-lvis-tk-ti-option} \\
      & \pref{eq:h9-lvis-tk-ti-option} 
      \so S_9\cpeC\tr_i \nsubseteq \luvis{S_9}{\tr_k} \label{eq:h9-lvis-tk-ti} \\
      & \fwop{j}{\obj}{2}{\ok_j} ~\text{is \last{} write on}~\obj~\text{in}~\tr_j
      \so \tr_j ~\text{is decided on}~\obj~\text{in}~
        S_9 \label{eq:h9-tj-decided} \\
      & \tr_j \prec_{S_9} \tr_k \wedge \pref{eq:h9-tj-decided}  
      \so S_9\cpeC\tr_j \subseteq \luvis{S_9}{\tr_k}~\text{or}~
          S_9\cpeC\tr_j \nsubseteq \luvis{S_9}{\tr_k} \label{eq:h9-lvis-tk-tj-option} \\
      & \pref{eq:h9-lvis-tk-tj-option} 
      \so S_9\cpeC\tr_j \nsubseteq \luvis{S_9}{\tr_k} \label{eq:h9-lvis-tk-tj} \\
    & k = k \so  S_9|\tr_k \subseteq \luvis{S_9}{\tr_k} \label{eq:h9-lvis-tk-tk} \\
    & \pref{eq:h9-lvis-tk-ti} \wedge \pref{eq:h9-lvis-tk-tj} \wedge \pref{eq:h9-lvis-tk-tk} 
      \so \luvis{S_9}{\tr_k} = S_9|\tr_k \label{eq:h9-tk-lvis} \\
    & \pref{eq:h9-tk-lvis} 
      \so \luvis{S_9}{\tr_k}|\obj = [\frop{j}{\obj}{0}]
      \label{eq:h9-tj-lvis-x} \\
    & \pref{eq:h9-tj-lvis-x} \wedge \pref{eq:seq-r} 
      \so \luvis{S_9}{\tr_k} ~\text{is legal} \label{eq:h9-lvis-tk-legal} \\
    & \pref{eq:h9-lvis-tk-legal} \so  
      \tr_k ~\text{in}~ S_9 ~\text{is last-use legal in}~ S_9 \label{eq:h9-tk-legal} \\
    & \pref{eq:s9-equiv} \wedge \pref{eq:h9-rt-sub} 
      \wedge \pref{eq:h9-ti-legal} \wedge \pref{eq:h9-tj-legal}\wedge \pref{eq:h9-tk-legal} \so 
      \hist_9 ~\text{is final-state last-use opaque}
    \end{align}
\end{proof}

Let $P^9_1$ be a prefix s.t. $\hist_9 = P^9_1 \cdot [\res{k}{}{0}]$.

\begin{lemma} \label{lemma:p91-fslop}
    $P^9_1$ is final-state last-use opaque.    
\end{lemma}

\begin{proof}
    \begin{align}
    & \text{let}~ C^9_1 = \compl{P^9_1} =  P^9_1 \cdot [\tryA_j\to\ab_j, \res{k}{}{\ab_k}] \\
    & \text{let}~ S^9_1 = C^9_1|\tr_i \cdot C^9_1|\tr_j \cdot C^9_1|\tr_k \label{eq:let-s91} \\
    & %
      S^9_1 \equiv C^9_1 \label{eq:s91-equiv}\\
    & \text{real time order}~ \prec_{P^9_1} =  \{ \tr_i \prec_{P^9_1} \tr_k \}\label{eq:s91-rt-p91} \\
    & \text{real time order}~ \prec_{S^9_1} = 
      \{ \tr_i \prec_{S^9_1} \tr_j, \tr_i  \prec_{S^9_1} \tr_k, \tr_j  \prec_{S^9_1} \tr_k  \} \label{eq:p91-rt-s91} \\
    & \pref{eq:s91-rt-p91} \wedge \pref{eq:p91-rt-s91}
      \so \prec_{S^9_1} \subseteq \prec_{P^9_1} \label{eq:p91-rt-sub} \\
    & i = i \so S^9_1|\tr_i \subseteq \luvis{S^9_1}{\tr_i} \label{eq:p91-lvis-ti-ti} \\
    & \tr_i \prec_{S^9_1} \tr_j 
      \so S^9_1|\tr_j \nsubseteq \luvis{S^9_1}{\tr_i} \label{eq:p91-lvis-ti-tj} \\
    & \tr_i \prec_{S^9_1} \tr_k 
      \so S^9_1|\tr_k \nsubseteq \luvis{S^9_1}{\tr_i} \label{eq:p91-lvis-ti-tk} \\
    & \pref{eq:p91-lvis-ti-ti} \wedge \pref{eq:p91-lvis-ti-tj} \wedge \pref{eq:p91-lvis-ti-tk} 
      \so \luvis{S^9_1}{\tr_i} = S^9_1|\tr_i \label{eq:p91-ti-lvis} \\
    & \pref{eq:p91-ti-lvis} \so \luvis{S^9_1}{\tr_i}|\obj = [\fwop{i}{\obj}{1}{\ok_i}] 
      \label{eq:p91-ti-lvis-x} \\
    & \pref{eq:p91-ti-lvis-x} \wedge \pref{eq:seq-w} 
      \so \luvis{S^9_1}{\tr_i} ~\text{is legal} \label{eq:p91-lvis-ti-legal} \\
    & \pref{eq:p91-lvis-ti-legal} \so  
      \tr_i ~\text{in}~ S^9_1 ~\text{is last-use legal in}~ S^9_1 \label{eq:p91-ti-legal} \\
    \end{align}    
    \begin{align}
    & \fwop{i}{\obj}{1}{\ok_i} ~\text{is \last{} write on}~\obj~\text{in}~\tr_i
      \so \tr_i ~\text{is decided on}~\obj ~\text{in}~
      S^9_1 \label{eq:p91-ti-decided} \\
    & \tr_i \prec_{S^9_1} \tr_j \wedge \pref{eq:p91-ti-decided}
      \so S^9_1\cpeC\tr_i \subseteq \luvis{S^9_1}{\tr_j} \label{eq:p91-lvis-tj-ti} \\
    & j = j \so  S^9_1|\tr_j \subseteq \luvis{S^9_1}{\tr_j} \label{eq:p91-lvis-tj-tj} \\
    & \tr_j \prec_{S^9_1} \tr_k 
      \so S^9_1|\tr_k \nsubseteq \luvis{S^9_1}{\tr_j} \label{eq:p91-lvis-tj-tk} \\
    & \pref{eq:p91-lvis-tj-ti} \wedge \pref{eq:p91-lvis-tj-tj} \wedge \pref{eq:p91-lvis-tj-tk} 
      \so \luvis{S^9_1}{\tr_j} = S^9_1\cpeC\tr_i \cdot S^9_1|\tr_j \label{eq:p91-tj-lvis} \\
    & \pref{eq:p91-tj-lvis} 
      \so \luvis{S^9_1}{\tr_j}|\obj = [\fwop{i}{\obj}{1}{\ok_i}, \frop{j}{\obj}{1}, \fwop{j}{\obj}{2}{\ok_j}]
      \label{eq:p91-tj-lvis-x} \\
    & \pref{eq:p91-tj-lvis-x} \wedge \pref{eq:seq-wrw} 
      \so \luvis{S^9_1}{\tr_j} ~\text{is legal} \label{eq:p91-lvis-tj-legal} \\
    & \pref{eq:p91-lvis-tj-legal} \so  
      \tr_j ~\text{in}~ S^9_1 ~\text{is last-use legal in}~ S^9_1 \label{eq:p91-tj-legal} \\
      & \tr_i \prec_{S^9_1} \tr_k \wedge \pref{eq:p91-ti-decided}  
      \so S^9_1\cpeC\tr_i \subseteq \luvis{S^9_1}{\tr_k}~\text{or}~
          S^9_1\cpeC\tr_i \nsubseteq \luvis{S^9_1}{\tr_k} \label{eq:p91-lvis-tk-ti-option} \\
      & \pref{eq:p91-lvis-tk-ti-option} 
      \so S^9_1\cpeC\tr_i \nsubseteq \luvis{S^9_1}{\tr_k} \label{eq:p91-lvis-tk-ti} \\
      & \fwop{j}{\obj}{2}{\ok_j} ~\text{is \last{} write on}~\obj~\text{in}~\tr_j
      \so \tr_j ~\text{is decided on}~\obj~\text{in}~
        S^9_1 \label{eq:p91-tj-decided} \\
      & \tr_j \prec_{S^9_1} \tr_k \wedge \pref{eq:p91-tj-decided}  
      \so S^9_1\cpeC\tr_j \subseteq \luvis{S^9_1}{\tr_k}~\text{or}~
          S^9_1\cpeC\tr_j \nsubseteq \luvis{S^9_1}{\tr_k} \label{eq:p91-lvis-tk-tj-option} \\
      & \pref{eq:p91-lvis-tk-tj-option} 
      \so S^9_1\cpeC\tr_j \nsubseteq \luvis{S^9_1}{\tr_k} \label{eq:p91-lvis-tk-tj} \\
    & k = k \so  S^9_1|\tr_k \subseteq \luvis{S^9_1}{\tr_k} \label{eq:p91-lvis-tk-tk} \\
    & \pref{eq:p91-lvis-tk-ti} \wedge \pref{eq:p91-lvis-tk-tj} \wedge \pref{eq:p91-lvis-tk-tk} 
      \so \luvis{S^9_1}{\tr_k} = S^9_1|\tr_k \label{eq:p91-tk-lvis} \\
    & \pref{eq:p91-tk-lvis} 
      \so \luvis{S^9_1}{\tr_k}|\obj = [\frop{k}{\obj}{\ab_k}]
      \label{eq:p91-tj-lvis-x} \\
    & \pref{eq:p91-tj-lvis-x} \wedge \pref{eq:seq-ra} 
      \so \luvis{S^9_1}{\tr_k} ~\text{is legal} \label{eq:p91-lvis-tk-legal} \\
    & \pref{eq:p91-lvis-tk-legal} \so  
      \tr_k ~\text{in}~ S^9_1 ~\text{is last-use legal in}~ S^9_1 \label{eq:p91-tk-legal} \\
    & \pref{eq:s91-equiv} \wedge \pref{eq:p91-rt-sub} 
      \wedge \pref{eq:p91-ti-legal} \wedge \pref{eq:p91-tj-legal}\wedge \pref{eq:p91-tk-legal} \so 
      P^9_1 ~\text{is final-state last-use opaque}
    \end{align}
\end{proof}

Let $P^9_2$ be a prefix s.t. $\hist_9 = P^9_2 \cdot [\frop{k}{\obj}{0}]$.

\begin{proof}
    \begin{align}
    & \text{let}~ C^9_2 = \compl{P^9_2} =  P^9_2 \cdot [\res{j}{}{\ab_j}, \res{k}{}{\ab_j}] \\
    & \text{let}~ S^9_2 = C^9_2|\tr_i \cdot C^9_2|\tr_j \cdot C^9_2|\tr_k \label{eq:let-s92} \\
    & %
      S^9_2 \equiv C^9_2 \label{eq:s92-equiv}\\
    & \text{real time order}~ \prec_{P^9_2} =  \{ \tr_i \prec_{P^9_2} \tr_k \}\label{eq:s92-rt-p92} \\
    & \text{real time order}~ \prec_{S^9_2} = 
      \{ \tr_i \prec_{S^9_2} \tr_j, \tr_i  \prec_{S^9_2} \tr_k, \tr_j  \prec_{S^9_2} \tr_k  \} \label{eq:p92-rt-s92} \\
    & \pref{eq:s92-rt-p92} \wedge \pref{eq:p92-rt-s92}
      \so \prec_{S^9_2} \subseteq \prec_{P^9_2} \label{eq:p92-rt-sub} \\
    & i = i \so S^9_2|\tr_i \subseteq \luvis{S^9_2}{\tr_i} \label{eq:p92-lvis-ti-ti} \\
    & \tr_i \prec_{S^9_2} \tr_j 
      \so S^9_2|\tr_j \nsubseteq \luvis{S^9_2}{\tr_i} \label{eq:p92-lvis-ti-tj} \\
    & \tr_i \prec_{S^9_2} \tr_k 
      \so S^9_2|\tr_k \nsubseteq \luvis{S^9_2}{\tr_i} \label{eq:p92-lvis-ti-tk} \\
    & \pref{eq:p92-lvis-ti-ti} \wedge \pref{eq:p92-lvis-ti-tj} \wedge \pref{eq:p92-lvis-ti-tk} 
      \so \luvis{S^9_2}{\tr_i} = S^9_2|\tr_i \label{eq:p92-ti-lvis} \\
    & \pref{eq:p92-ti-lvis} \so \luvis{S^9_2}{\tr_i}|\obj = [\fwop{i}{\obj}{1}{\ok_i}] 
      \label{eq:p92-ti-lvis-x} \\
    & \pref{eq:p92-ti-lvis-x} \wedge \pref{eq:seq-w} 
      \so \luvis{S^9_2}{\tr_i} ~\text{is legal} \label{eq:p92-lvis-ti-legal} \\
    & \pref{eq:p92-lvis-ti-legal} \so  
      \tr_i ~\text{in}~ S^9_2 ~\text{is last-use legal in}~ S^9_2 \label{eq:p92-ti-legal} \\
    & \fwop{i}{\obj}{1}{\ok_i} ~\text{is \last{} write on}~\obj~\text{in}~\tr_i
      \so \tr_i ~\text{is decided on}~\obj ~\text{in}~
      S^9_2 \label{eq:p92-ti-decided} \\
    & \tr_i \prec_{S^9_2} \tr_j \wedge \pref{eq:p92-ti-decided}
      \so S^9_2\cpeC\tr_i \subseteq \luvis{S^9_2}{\tr_j} \label{eq:p92-lvis-tj-ti} \\
    & j = j \so  S^9_2|\tr_j \subseteq \luvis{S^9_2}{\tr_j} \label{eq:p92-lvis-tj-tj} \\
    & \tr_j \prec_{S^9_2} \tr_k 
      \so S^9_2|\tr_k \nsubseteq \luvis{S^9_2}{\tr_j} \label{eq:p92-lvis-tj-tk} \\
    & \pref{eq:p92-lvis-tj-ti} \wedge \pref{eq:p92-lvis-tj-tj} \wedge \pref{eq:p92-lvis-tj-tk} 
      \so \luvis{S^9_2}{\tr_j} = S^9_2\cpeC\tr_i \cdot S^9_2|\tr_j \label{eq:p92-tj-lvis} \\
    & \pref{eq:p92-tj-lvis} 
      \so \luvis{S^9_2}{\tr_j}|\obj = [\fwop{i}{\obj}{1}{\ok_i}, \frop{j}{\obj}{1}, \fwop{j}{\obj}{2}{\ok_j}]
      \label{eq:p92-tj-lvis-x} \\
    & \pref{eq:p92-tj-lvis-x} \wedge \pref{eq:seq-wrw} 
      \so \luvis{S^9_2}{\tr_j} ~\text{is legal} \label{eq:p92-lvis-tj-legal} \\
    & \pref{eq:p92-lvis-tj-legal} \so  
      \tr_j ~\text{in}~ S^9_2 ~\text{is last-use legal in}~ S^9_2 \label{eq:p92-tj-legal} \\
      & \tr_i \prec_{S^9_2} \tr_k \wedge \pref{eq:p92-ti-decided}  
      \so S^9_2\cpeC\tr_i \subseteq \luvis{S^9_2}{\tr_k}~\text{or}~
          S^9_2\cpeC\tr_i \nsubseteq \luvis{S^9_2}{\tr_k} \label{eq:p92-lvis-tk-ti-option} \\
      & \pref{eq:p92-lvis-tk-ti-option} 
      \so S^9_2\cpeC\tr_i \nsubseteq \luvis{S^9_2}{\tr_k} \label{eq:p92-lvis-tk-ti} \\
      & \fwop{j}{\obj}{2}{\ok_j} ~\text{is \last{} write on}~\obj~\text{in}~\tr_j
      \so \tr_j ~\text{is decided on}~\obj~\text{in}~
        S^9_2 \label{eq:p92-tj-decided} \\
      & \tr_j \prec_{S^9_2} \tr_k \wedge \pref{eq:p92-tj-decided}  
      \so S^9_2\cpeC\tr_j \subseteq \luvis{S^9_2}{\tr_k}~\text{or}~
          S^9_2\cpeC\tr_j \nsubseteq \luvis{S^9_2}{\tr_k} \label{eq:p92-lvis-tk-tj-option} \\
      & \pref{eq:p92-lvis-tk-tj-option} 
      \so S^9_2\cpeC\tr_j \nsubseteq \luvis{S^9_2}{\tr_k} \label{eq:p92-lvis-tk-tj} \\
    & k = k \so  S^9_2|\tr_k \subseteq \luvis{S^9_2}{\tr_k} \label{eq:p92-lvis-tk-tk} \\
    & \pref{eq:p92-lvis-tk-ti} \wedge \pref{eq:p92-lvis-tk-tj} \wedge \pref{eq:p92-lvis-tk-tk} 
      \so \luvis{S^9_2}{\tr_k} = S^9_2|\tr_k \label{eq:p92-tk-lvis} \\
    & \pref{eq:p92-tk-lvis} 
      \so \luvis{S^9_2}{\tr_k}|\obj = \varnothing
      \label{eq:p92-tj-lvis-x} \\
    & \pref{eq:p92-tj-lvis-x} \wedge \pref{eq:seq-empty} 
      \so \luvis{S^9_2}{\tr_k} ~\text{is legal} \label{eq:p92-lvis-tk-legal} \\
    & \pref{eq:p92-lvis-tk-legal} \so  
      \tr_k ~\text{in}~ S^9_2 ~\text{is last-use legal in}~ S^9_2 \label{eq:p92-tk-legal} \\
    & \pref{eq:s92-equiv} \wedge \pref{eq:p92-rt-sub} 
      \wedge \pref{eq:p92-ti-legal} \wedge \pref{eq:p92-tj-legal}\wedge \pref{eq:p92-tk-legal} \so 
      P^9_2 ~\text{is final-state last-use opaque} \\
    \end{align}
\end{proof}

Let $P^9_2$ be a prefix s.t. $\hist_9 = P^9_2 \cdot [\res{j}{}{\ok_j},\frop{k}{\obj}{0}]$.

\begin{proof}
    \begin{align}
    & \text{let}~ C^9_3 = \compl{P^9_3} =  P^9_3 \cdot [\res{j}{}{\ab_j}, \res{k}{}{\ab_j}] \\
    & \text{let}~ S^9_3 = C^9_3|\tr_i \cdot C^9_3|\tr_j \cdot C^9_3|\tr_k \label{eq:let-s93} \\
    & %
      S^9_3 \equiv C^9_3 \label{eq:s93-equiv}\\
    & \text{real time order}~ \prec_{P^9_3} =  \{ \tr_i \prec_{P^9_3} \tr_k \}\label{eq:s93-rt-p93} \\
    & \text{real time order}~ \prec_{S^9_3} = 
      \{ \tr_i \prec_{S^9_3} \tr_j, \tr_i  \prec_{S^9_3} \tr_k, \tr_j  \prec_{S^9_3} \tr_k  \} \label{eq:p93-rt-s93} \\
    & \pref{eq:s93-rt-p93} \wedge \pref{eq:p93-rt-s93}
      \so \prec_{S^9_3} \subseteq \prec_{P^9_3} \label{eq:p93-rt-sub} \\
    & i = i \so S^9_3|\tr_i \subseteq \luvis{S^9_3}{\tr_i} \label{eq:p93-lvis-ti-ti} \\
    & \tr_i \prec_{S^9_3} \tr_j 
      \so S^9_3|\tr_j \nsubseteq \luvis{S^9_3}{\tr_i} \label{eq:p93-lvis-ti-tj} \\
    & \tr_i \prec_{S^9_3} \tr_k 
      \so S^9_3|\tr_k \nsubseteq \luvis{S^9_3}{\tr_i} \label{eq:p93-lvis-ti-tk} \\
    & \pref{eq:p93-lvis-ti-ti} \wedge \pref{eq:p93-lvis-ti-tj} \wedge \pref{eq:p93-lvis-ti-tk} 
      \so \luvis{S^9_3}{\tr_i} = S^9_3|\tr_i \label{eq:p93-ti-lvis} \\
    & \pref{eq:p93-ti-lvis} \so \luvis{S^9_3}{\tr_i}|\obj = [\fwop{i}{\obj}{1}{\ok_i}] 
      \label{eq:p93-ti-lvis-x} \\
    & \pref{eq:p93-ti-lvis-x} \wedge \pref{eq:seq-w} 
      \so \luvis{S^9_3}{\tr_i} ~\text{is legal} \label{eq:p93-lvis-ti-legal} \\
    & \pref{eq:p93-lvis-ti-legal} \so  
      \tr_i ~\text{in}~ S^9_3 ~\text{is last-use legal in}~ S^9_3 \label{eq:p93-ti-legal} \\
    & \fwop{i}{\obj}{1}{\ok_i} ~\text{is \last{} write on}~\obj~\text{in}~\tr_i
      \so \tr_i ~\text{is decided on}~\obj ~\text{in}~
      S^9_3 \label{eq:p93-ti-decided} \\
    & \tr_i \prec_{S^9_3} \tr_j \wedge \pref{eq:p93-ti-decided}
      \so S^9_3\cpeC\tr_i \subseteq \luvis{S^9_3}{\tr_j} \label{eq:p93-lvis-tj-ti} \\
    & j = j \so  S^9_3|\tr_j \subseteq \luvis{S^9_3}{\tr_j} \label{eq:p93-lvis-tj-tj} \\
    & \tr_j \prec_{S^9_3} \tr_k 
      \so S^9_3|\tr_k \nsubseteq \luvis{S^9_3}{\tr_j} \label{eq:p93-lvis-tj-tk} \\
    & \pref{eq:p93-lvis-tj-ti} \wedge \pref{eq:p93-lvis-tj-tj} \wedge \pref{eq:p93-lvis-tj-tk} 
      \so \luvis{S^9_3}{\tr_j} = S^9_3\cpeC\tr_i \cdot S^9_3|\tr_j \label{eq:p93-tj-lvis} \\
    & \pref{eq:p93-tj-lvis} 
      \so \luvis{S^9_3}{\tr_j}|\obj = [\fwop{i}{\obj}{1}{\ok_i}, \frop{j}{\obj}{1}, \fwop{j}{\obj}{2}{\ab_j}]
      \label{eq:p93-tj-lvis-x} \\
    & \pref{eq:p93-tj-lvis-x} \wedge \pref{eq:seq-wrwa} 
      \so \luvis{S^9_3}{\tr_j} ~\text{is legal} \label{eq:p93-lvis-tj-legal} \\
    & \pref{eq:p93-lvis-tj-legal} \so  
      \tr_j ~\text{in}~ S^9_3 ~\text{is last-use legal in}~ S^9_3 \label{eq:p93-tj-legal} \\
      & \tr_i \prec_{S^9_3} \tr_k \wedge \pref{eq:p93-ti-decided}  
      \so S^9_3\cpeC\tr_i \subseteq \luvis{S^9_3}{\tr_k}~\text{or}~
          S^9_3\cpeC\tr_i \nsubseteq \luvis{S^9_3}{\tr_k} \label{eq:p93-lvis-tk-ti-option} \\
      & \pref{eq:p93-lvis-tk-ti-option} 
      \so S^9_3\cpeC\tr_i \nsubseteq \luvis{S^9_3}{\tr_k} \label{eq:p93-lvis-tk-ti} \\
      & \fwop{j}{\obj}{2}{\ok_j} ~\text{is \last{} write on}~\obj~\text{in}~\tr_j
      \so \tr_j ~\text{is decided on}~\obj~\text{in}~
        S^9_3 \label{eq:p93-tj-decided} \\
      & \tr_j \prec_{S^9_3} \tr_k \wedge \pref{eq:p93-tj-decided}  
      \so S^9_3\cpeC\tr_j \subseteq \luvis{S^9_3}{\tr_k}~\text{or}~
          S^9_3\cpeC\tr_j \nsubseteq \luvis{S^9_3}{\tr_k} \label{eq:p93-lvis-tk-tj-option} \\
      & \pref{eq:p93-lvis-tk-tj-option} 
      \so S^9_3\cpeC\tr_j \nsubseteq \luvis{S^9_3}{\tr_k} \label{eq:p93-lvis-tk-tj} \\
    & k = k \so  S^9_3|\tr_k \subseteq \luvis{S^9_3}{\tr_k} \label{eq:p93-lvis-tk-tk} \\
    & \pref{eq:p93-lvis-tk-ti} \wedge \pref{eq:p93-lvis-tk-tj} \wedge \pref{eq:p93-lvis-tk-tk} 
      \so \luvis{S^9_3}{\tr_k} = S^9_3|\tr_k \label{eq:p93-tk-lvis} \\
    & \pref{eq:p93-tk-lvis} 
      \so \luvis{S^9_3}{\tr_k}|\obj = \varnothing
      \label{eq:p93-tj-lvis-x} \\
    & \pref{eq:p93-tj-lvis-x} \wedge \pref{eq:seq-empty} 
      \so \luvis{S^9_3}{\tr_k} ~\text{is legal} \label{eq:p93-lvis-tk-legal} \\
    & \pref{eq:p93-lvis-tk-legal} \so  
      \tr_k ~\text{in}~ S^9_3 ~\text{is last-use legal in}~ S^9_3 \label{eq:p93-tk-legal} \\
    & \pref{eq:s93-equiv} \wedge \pref{eq:p93-rt-sub} 
      \wedge \pref{eq:p93-ti-legal} \wedge \pref{eq:p93-tj-legal}\wedge \pref{eq:p93-tk-legal} \so 
      P^9_3 ~\text{is final-state last-use opaque}
    \end{align}
\end{proof}

Let $P^9_4$ be a prefix s.t. $\hist_9 = P^9_4 \cdot [\fwop{j}{\obj}{2}{\ok_j}, \frop{k}{\obj}{0}]$.

\begin{lemma} \label{lemma:p94-fslop}
    $P^9_4$ is final-state last-use opaque.    
\end{lemma}

\begin{proof}
    \begin{align}
    & \text{let}~ C^9_4 = \compl{P^9_4} =  P^9_4 \cdot [\res{j}{}{\ab_j}, \res{k}{}{\ab_j}] \\
    & \text{let}~ S^9_4 = C^9_4|\tr_i \cdot C^9_4|\tr_j \cdot C^9_4|\tr_k \label{eq:let-p94} \\
    & %
      S^9_4 \equiv C^9_4 \label{eq:p94-equiv}\\
    & \text{real time order}~ \prec_{P^9_4} =  \{ \tr_i \prec_{P^9_4} \tr_k \}\label{eq:p94-rt-p94} \\
    & \text{real time order}~ \prec_{S^9_4} = 
      \{ \tr_i \prec_{S^9_4} \tr_j, \tr_i  \prec_{S^9_4} \tr_k, \tr_j  \prec_{S^9_4} \tr_k  \} \label{eq:p94-rt-p94} \\
    & \pref{eq:p94-rt-p94} \wedge \pref{eq:p94-rt-p94}
      \so \prec_{S^9_4} \subseteq \prec_{P^9_4} \label{eq:p94-rt-sub} \\
    & i = i \so S^9_4|\tr_i \subseteq \luvis{S^9_4}{\tr_i} \label{eq:p94-lvis-ti-ti} \\
    & \tr_i \prec_{S^9_4} \tr_j 
      \so S^9_4|\tr_j \nsubseteq \luvis{S^9_4}{\tr_i} \label{eq:p94-lvis-ti-tj} \\
    & \tr_i \prec_{S^9_4} \tr_k 
      \so S^9_4|\tr_k \nsubseteq \luvis{S^9_4}{\tr_i} \label{eq:p94-lvis-ti-tk} \\
    & \pref{eq:p94-lvis-ti-ti} \wedge \pref{eq:p94-lvis-ti-tj} \wedge \pref{eq:p94-lvis-ti-tk} 
      \so \luvis{S^9_4}{\tr_i} = S^9_4|\tr_i \label{eq:p94-ti-lvis} \\
    & \pref{eq:p94-ti-lvis} \so \luvis{S^9_4}{\tr_i}|\obj = [\fwop{i}{\obj}{1}{\ok_i}] 
      \label{eq:p94-ti-lvis-x} \\
    & \pref{eq:p94-ti-lvis-x} \wedge \pref{eq:seq-w} 
      \so \luvis{S^9_4}{\tr_i} ~\text{is legal} \label{eq:p94-lvis-ti-legal} \\
    & \pref{eq:p94-lvis-ti-legal} \so  
      \tr_i ~\text{in}~ S^9_4 ~\text{is last-use legal in}~ S^9_4 \label{eq:p94-ti-legal} \\
    & \fwop{i}{\obj}{1}{\ok_i} ~\text{is \last{} write on}~\obj~\text{in}~\tr_i
      \so \tr_i ~\text{is decided on}~\obj ~\text{in}~
      S^9_4 \label{eq:p94-ti-decided} \\
    & \tr_i \prec_{S^9_4} \tr_j \wedge \pref{eq:p94-ti-decided}
      \so S^9_4\cpeC\tr_i \subseteq \luvis{S^9_4}{\tr_j} \label{eq:p94-lvis-tj-ti} \\
    & j = j \so  S^9_4|\tr_j \subseteq \luvis{S^9_4}{\tr_j} \label{eq:p94-lvis-tj-tj} \\
    & \tr_j \prec_{S^9_4} \tr_k 
      \so S^9_4|\tr_k \nsubseteq \luvis{S^9_4}{\tr_j} \label{eq:p94-lvis-tj-tk} \\
    & \pref{eq:p94-lvis-tj-ti} \wedge \pref{eq:p94-lvis-tj-tj} \wedge \pref{eq:p94-lvis-tj-tk} 
      \so \luvis{S^9_4}{\tr_j} = S^9_4\cpeC\tr_i \cdot S^9_4|\tr_j \label{eq:p94-tj-lvis} \\
    & \pref{eq:p94-tj-lvis} 
      \so \luvis{S^9_4}{\tr_j}|\obj = [\fwop{i}{\obj}{1}{\ok_i}, \frop{j}{\obj}{1}]
      \label{eq:p94-tj-lvis-x} \\
    & \pref{eq:p94-tj-lvis-x} \wedge \pref{eq:seq-wr} 
      \so \luvis{S^9_4}{\tr_j} ~\text{is legal} \label{eq:p94-lvis-tj-legal} \\
    & \pref{eq:p94-lvis-tj-legal} \so  
      \tr_j ~\text{in}~ S^9_4 ~\text{is last-use legal in}~ S^9_4 \label{eq:p94-tj-legal} \\
    \end{align}    
    \begin{align}
      & \tr_i \prec_{S^9_4} \tr_k \wedge \pref{eq:p94-ti-decided}  
      \so S^9_4\cpeC\tr_i \subseteq \luvis{S^9_4}{\tr_k}~\text{or}~
          S^9_4\cpeC\tr_i \nsubseteq \luvis{S^9_4}{\tr_k} \label{eq:p94-lvis-tk-ti-option} \\
      & \pref{eq:p94-lvis-tk-ti-option} 
      \so S^9_4\cpeC\tr_i \nsubseteq \luvis{S^9_4}{\tr_k} \label{eq:p94-lvis-tk-ti} \\
      & \fwop{j}{\obj}{2}{\ok_j} ~\text{is \last{} write on}~\obj~\text{in}~\tr_j
      \so \tr_j ~\text{is decided on}~\obj~\text{in}~
        S^9_4 \label{eq:p94-tj-decided} \\
      & \tr_j \prec_{S^9_4} \tr_k \wedge \pref{eq:p94-tj-decided}  
      \so S^9_4\cpeC\tr_j \subseteq \luvis{S^9_4}{\tr_k}~\text{or}~
          S^9_4\cpeC\tr_j \nsubseteq \luvis{S^9_4}{\tr_k} \label{eq:p94-lvis-tk-tj-option} \\
      & \pref{eq:p94-lvis-tk-tj-option} 
      \so S^9_4\cpeC\tr_j \nsubseteq \luvis{S^9_4}{\tr_k} \label{eq:p94-lvis-tk-tj} \\
    & k = k \so  S^9_4|\tr_k \subseteq \luvis{S^9_4}{\tr_k} \label{eq:p94-lvis-tk-tk} \\
    & \pref{eq:p94-lvis-tk-ti} \wedge \pref{eq:p94-lvis-tk-tj} \wedge \pref{eq:p94-lvis-tk-tk} 
      \so \luvis{S^9_4}{\tr_k} = S^9_4|\tr_k \label{eq:p94-tk-lvis} \\
    & \pref{eq:p94-tk-lvis} 
      \so \luvis{S^9_4}{\tr_k}|\obj = \varnothing
      \label{eq:p94-tj-lvis-x} \\
    & \pref{eq:p94-tj-lvis-x} \wedge \pref{eq:seq-empty} 
      \so \luvis{S^9_4}{\tr_k} ~\text{is legal} \label{eq:p94-lvis-tk-legal} \\
    & \pref{eq:p94-lvis-tk-legal} \so  
      \tr_k ~\text{in}~ S^9_4 ~\text{is last-use legal in}~ S^9_4 \label{eq:p94-tk-legal} \\
    \end{align}    
    \begin{align}
    & \pref{eq:p94-equiv} \wedge \pref{eq:p94-rt-sub} 
      \wedge \pref{eq:p94-ti-legal} \wedge \pref{eq:p94-tj-legal}\wedge \pref{eq:p94-tk-legal} \so 
      P^9_4 ~\text{is final-state last-use opaque}
    \end{align}
\end{proof}

Let $P^9_5$ be a prefix s.t. $\hist_9 = P^9_5 \cdot [\init_k, \fwop{j}{\obj}{2}{\ok_j}, \frop{k}{\obj}{0}]$.

\begin{lemma} \label{lemma:p95-fslop}
    $P^9_4$ is final-state last-use opaque.    
\end{lemma}

\begin{proof}
    \begin{align}
    & P^9_5 = P^3_2 \wedge \rlemma{lemma:p32-fslop}
    \end{align}
\end{proof}

Let $P^9_p$ be any prefix of $P^9_5$.

\begin{lemma} \label{lemma:p9p-fslop}
    $P^9_p$ is final-state last-use opaque.    
\end{lemma}

\begin{proof}
    \begin{align}
    & \hist_1 = P^1_4 \cdot R  \wedge \rcor{cor:h3-pref-lop}
      \so P^1_4 ~\text{is last-use plague} \label{eq:p9p-pref-lopaque} \\
    & P^1_4 = P^3_3 \wedge \pref{eq:p9p-pref-lopaque} 
      \so P^9_p ~\text{is last-use lopaque} \label{eq:p9p-lop} \\
    & \pref{eq:p9p-lop} 
      \so P^9_p ~\text{is final-state last-use lopaque}   
    \end{align}
\end{proof}

\begin{lemma} \label{lemma:h3-lop}
    $\hist_9$ is last-use opaque.    
\end{lemma}

\begin{proof}
    Since, from Lemmas \ref{lemma:p91-fslop}--\ref{lemma:p9p-fslop}, all
    prefixes of $\hist_9$ are final-state last-use opaque, then by
    \rdef{def:lopacity} $\hist_9$ is last-use opaque.
\end{proof}

\begin{corollary} \label{cor:h9-pref-lop}
    Any prefix of $\hist_9$ is last-use opaque.
\end{corollary}
           
\section{Property Comparison} \label{sec:property-comparison}

VWC is incomparable to \lopacity{}.

\begin{lemma}
    There exists a \lopaque{} history $H$ that is not virtual world consistent.
\end{lemma}

\begin{proof}[sketch]
    Since, as an extension of by \rlemma{lemma:h1-lop}, \lopacity{} supports early
    release, then by \rdef{def:release-support} and
    (by \rdef{def:aborting-support}) from  \rlemma{lemma:h3-lop} there exists some \lopaque{} history where
    some transaction reads from a live transaction and aborts. Since, by
    \rthm{thm:vwc-no-abort} VWC, does not support aborting
    releasing transactions, then, by the same definitions, such a history is
    not VWC. Hence a history with a transaction releasing early may be
    \lopaque{} but not VWC.
\end{proof}

\begin{theorem}
    There exists a virtual world consistent history $H$ that is not \lopaque{}.
\end{theorem}

\begin{proof}[sketch]
    Since each transaction in a VWC history can be explained by a different
    causal past from other transactions, it is possible that in a correct VWC
    history transactions do not agree on the order of operations in the
    sequential witness history.
    However, in order for $H$ to to be \lopaque{} the legality of transactions
    needs to be established using a single sequential history with a single
    order of operations.
    Thus, it is possible for a VWC history not to be \lopaque{}.
\end{proof}

Since TMS1 is incomparable to \lopacity{}. 

\begin{theorem} \label{thm:lopacity-not-tms1}
    There exists a \lopaque{} history $H$ that is not TMS1.
\end{theorem}

\begin{proof}[sketch]
    Since, as an extension of by \rlemma{lemma:h1-lop}, \lopacity{} supports early release,
    then by \rdef{def:release-support} and \ref{def:early-release} there exist
    histories that are \lopaque{} where some
    transaction reads from a live transaction. Since, by
    \rthm{thm:tms1-early-release} TMS1, does not support early release, then, by
    the same definitions, histories containing early release are not TMS1. Hence a
    history with a transaction releasing early may be \lopaque{} but not TMS1.
\end{proof}

 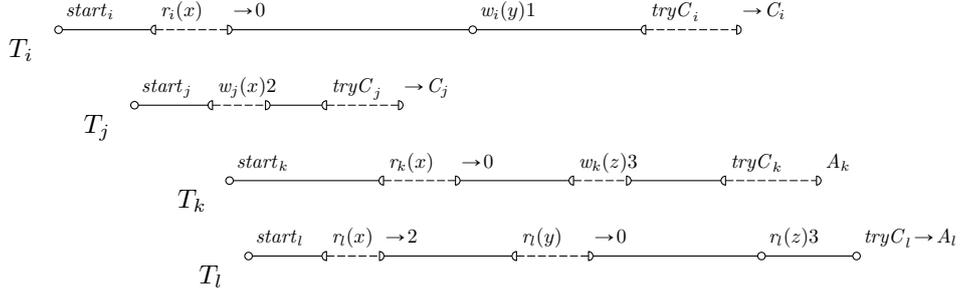
\begin{figure}
\begin{center}
\begin{tikzpicture}
     \draw
           (0,3)        node[tid]       {$\tr_i$}
                        node[aop]       {$\init_i$} %
                        node[dot]       {} 

      -- ++(1.25,0)     node[aop]       {$\trop{i}{\obj}{}$}
                        node[inv] (ii)  {}

         ++(1,0)        node[aop]       {$\!\to\!{0}$}
                        node[res] (ri)  {}

      -- ++(3.2,0)      node[aop]       {$\twop{i}{\objy}{1}$}
                        node[dot]       {}

      -- ++(2.25,0)     node[aop]       {$\tryC_i$}
                        node[inv] (iii) {}        

         ++(1.25,0)     node[aop]       {$\!\to\!\co_i$}
                        node[res] (rii) {}        
                        ;

     \draw[wait] (ii) -- (ri);
     \draw[wait] (iii) -- (rii);

     \draw
           (1,2)        node[tid]       {$\tr_{j}$}
                        node[aop]       {$\init_{j}$} %
                        node[dot]       {} 

      -- ++(1,0)        node[aop]       {$\twop{j}{\obj}{2}$}
                        node[inv] (ij)  {}

         ++(0.75,0)     node[aop]       {}
                        node[res] (rj)  {}

      -- ++(0.75,0)     node[aop]       {$\tryC_j$}
                        node[inv] (ijj) {}

         ++(1,0)        node[aop]       {$\!\to\!\co_j$}
                        node[res] (rjj) {}
                        ;

       \draw[wait] (ij) -- (rj);
       \draw[wait] (ijj) -- (rjj);
  
       \draw
           (2.25,1)     node[tid]       {$\tr_k$}
                        node[aop]       {$\init_k$} %
                        node[dot]       {} 

      -- ++(2,0)        node[aop]       {$\trop{k}{\obj}{}$}
                        node[inv] (ik)  {}

         ++(1,0)        node[aop]       {$\!\to\!{0}$}
                        node[res] (rk)  {}

      -- ++(1.5,0)      node[aop]       {$\twop{k}{\objz}{3}$}
                        node[inv] (ikk) {}
 
         ++(.75,0)      node[aop]       {}
                        node[res] (rkk) {}                       

      -- ++(1.25,0)     node[aop]       {$\tryC_k$}
                        node[inv] (ikkk){}     
                         
         ++(1.25,0)     node[aop]       {$\ab_k$}
                        node[res] (rkkk){}       
                        ;   

       \draw[wait] (ik) -- (rk);
       \draw[wait] (ikk) -- (rkk);
       \draw[wait] (ikkk) -- (rkkk);

       \draw
           (2.5,0)      node[tid]       {$\tr_l$}
                        node[aop]       {$\init_l$} %
                        node[dot]       {} 

      -- ++(1,0)        node[aop]       {$\trop{l}{\obj}{}$}
                        node[inv] (il)  {}

         ++(0.75,0)     node[aop]       {$\!\to\!{2}$}
                        node[res] (rl)  {}

      -- ++(1.75,0)     node[aop]       {$\trop{l}{\objy}{}$}
                        node[inv] (ill) {}

         ++(1.00,0)     node[aop]       {$\!\to\!{0}$}
                        node[res] (rll) {}

      -- ++(2.25,0)     node[aop]       {$\trop{l}{\objz}{3}$}
                        node[dot]       {}

      -- ++(1.25,0)     node[aop]       {$\tryC_l\!\to\!\ab_l$}
                        node[dot]       {}         
                        ;  
                        
       \draw[wait] (il) -- (rl);
       \draw[wait] (ill) -- (rll);

\end{tikzpicture}
\end{center}
\caption{\label{fig:tms1-history} TMS1 history example \cite{DGLM13}.}
\end{figure}

\begin{theorem} \label{thm:tms1-not-lopaque}
    There exists a TMS1 history $\hist$ that is not \lopaque{}.
\end{theorem}

\begin{proof}[sketch]
    Let history $\hist$ be the history presented in \rfig{fig:tms1-history}. In
    \cite{DGLM13} (Fig. 6 therein) the authors show that the history satisfies TMS1.
    The same history is not last-use opaque. Note that if $\vis{S}{\tr_i}$ is
    to be legal, in any $S$ equivalent to $\hist$, $\tr_i \prec_S \tr_j$,
    because $\tr_i$ reads $0$ from $\obj$ and $\tr_j$ writes $2$ to $\obj$ (and
    commits). In addition, $\tr_j \prec_S \tr_l$, because $\tr_l$ reads $2$
    from $\obj$ and $\tr_k \prec_S \tr_l$, because $\tr_l$ reads $\objz$ from
    $\tr_k$. Then, by extension $\tr_i \prec_S \tr_j \prec_S \tr_l$.
    However, note that in any $S$ it must be that $\tr_l \prec_S \tr_i$,
    because $\tr_l$ reads $\objy$ from $\tr_i$, which is a contradiction. Thus,
    $\hist$ is not last-use opaque.
\end{proof}

TMS2 is strictly stronger than \lopacity{}.

\begin{proposition} \label{prop:tms2-is-lopaque}
    All TMS2 histories are \lopaque{}.
\end{proposition}

\begin{proof}[sketch]
The authors of \cite{DGLM13} believe (but do not demonstrate) that 
all opaque histories satisfy TMS2. If this
is the case, then, since all opaque histories are \lopaque{}
(\rlemma{lemma:all-opaque-are-lopaque}), then it is true that all \lopaque{}
histories satisfy TMS2. Thus, we believe the proposition is true, pending a
demonstration that all opaque histories satisfy TMS2.
\end{proof}

\Lopacity{} and elastic opacity are incomparable.

\begin{lemma} \label{lemma:elastic-opacity-not-lopaque}
    There exists an elastic opaque history $\hist$ that is not \lopaque{}.
\end{lemma}

\begin{proof}[sketch]
    Since 
    elastic opaque histories may not be serializable \cite{FGG09}, and since, as
    all \lopaque{} histories trivially require serializability
    then some elastic opaque histories are not \lopaque{}.
\end{proof}

\begin{lemma} \label{thm:lopacity-not-elastic-opaque}
    There exists a \lopaque{} history $\hist$ that is not elastic opaque.
\end{lemma}

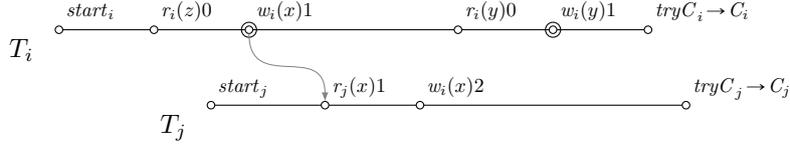
\begin{figure} 
\begin{center}
\begin{tikzpicture}
     \draw
           (0,2)        node[tid]       {$\tr_i$}
                        node[aop]       {$\init_i$} %
                        node[dot]       {} 

      -- ++(1.25,0)     node[aop]       {$\trop{i}{\objz}{0}$}
                        node[dot]       {}  

      -- ++(1.25,0)     node[aop]       {$\twop{i}{\obj}{1}$}
                        node[dot] (wi)  {}
                        node[cir]       {}

      -- ++(2.75,0)     node[aop]       {$\trop{i}{\objy}{0}$}
                        node[dot]       {}  

      -- ++(1.25,0)     node[aop]       {$\twop{i}{\objy}{1}$}
                        node[dot]       {}
                        node[cir]       {}

      -- ++(1.25,0)     node[aop]       {$\tryC_i\!\to\!\co_i$}
                        node[dot]       {}         
                        ;

     \draw
           (2.,1)       node[tid]       {$\tr_{j}$}
                        node[aop]       {$\init_{j}$} %
                        node[dot]       {} 

      -- ++(1.5,0)      node[aop]       {$\trop{j}{\obj}{1}$}
                        node[dot] (rj)  {}

      -- ++(1.25,0)     node[aop]       {$\twop{i}{\obj}{2}$}
                        node[dot]       {}

      -- ++(3.5,0)      node[aop]       {$\tryC_j\!\to\!\co_j$}
                        node[dot]       {}
                        ;
     
     \draw[hb] (wi) \squiggle (rj);
\end{tikzpicture}
\end{center}
\captionof{figure}{\label{fig:luopaque-history-not-elastic} \Lopaque{} history that does not satisfy elastic opacity.}
\end{figure}

\begin{proof}[sketch]
    Let history $\hist$ be the history presented in
    \rfig{fig:luopaque-history-not-elastic}.
    It should be straightforward to see that $\hist$ is \lopaque{} for an
    equivalent sequential history $S=\hist|\tr_i \cdot \hist|\tr_j$. Operations
    on $\objz$ are always justified in any sequential equivalent history since
    they are all within $\tr_i$ and their effects are not visible in $\tr_j$.
    The read operation on $\objy$ is expected to read $0$ since it is not
    preceded in $S$ by any write, and it does read $0$. Thus operations on
    $\objy$ and $\objz$ will not break legality of either $\tr_i$ or $\tr_j$.
    With that in mind, the history can be shown to be last-use opaque by
    analogy to \rlemma{lemma:h1-lop}.

    On the other hand, let $\tr_i$ be an elastic transaction.
    The only possible well-formed cut of $\hist|\tr_i$ is $C_i = \{ [
    \rop{z}0, \wop{x}1, \rop{y}0,$ $\wop{y}1 ] \}$. (In particular, the
    following cut is not well-formed, since $\wop{x}1$ and $\wop{y}0$ are in two
    different subhistories of the cut: $C_i' = \{ [ \rop{z}0, \wop{x}1
    ], [ \rop{y}0,$ $ \wop{y}1 ] \}$). Let $f_C(\hist)$ be a
    cutting function that applies cut $C$. Then, since the cut contains only one
    subhistory, it should be straightforward to see that $f_C(\hist) = \hist$.
    Then, we note that $\hist$ contains an operation in $\hist|\tr_j$ that reads
    the value of $\obj$ from $\hist|\tr_i$ and $\tr_i$ is live. 
    That means that in the prefix $P$ of $H$ s.t. $H = P \cdot
    [\tryC_i\!\to\!\co_i,\tryC_j\!\to\!\co_j]$ both transactions will be
    aborted in any completion of $P$, so for any sequential equivalent history
    $S$ $\vis{S}{\tr_i}$ will not contain $S|\tr_j$, since either $\tr_j$ is
    aborted in any $S$. Therefore $\vis{S}{\tr_i}$ will not justify reading $1$
    from $\obj$ and will not be legal, causing $P$ not to be final state
    opaque (\rdef{def:fs-opacity}), which in turn means that $\hist$ is not
    opaque (\rdef{def:opacity}).
\end{proof}

Last-use opacity is strictly stronger than recoverability.

\begin{lemma} \label{lem:lop-recoverability-apx}
    Any \lopaque{} history $H$ is recoverable.
\end{lemma}

\begin{proof}[sketch]
    Let us assume that $H$ is not recoverable. Then there must be some
    transactions $\tr_i$  and $\tr_j$ s.t. $\tr_j$ reads from $\tr_i$ and then
    $\tr_j$ commits before $\tr_i$. 
    By analogy to \rlemma{lemma:h2-lop}, such a history will contain a prefix
    $P$ where any completion will contain an aborted $\tr_i$ and a committed
    $\tr_j$, so for any equivalent sequential history $S$ $\vis{S}{\tr_j}$ will
    not contain $S|\tr_i$. Since $\tr_i$ reads from $\tr_i$ then such
    $\vis{S}{\tr_j}$ will not be legal, so by \rdef{def:fs-lopacity} $P$ is not
    last-use opaque and thus, by \rdef{def:lopacity}, $\hist$ is not last-use
    opaque, which is a contradiction.
\end{proof}

ACA is incomparable to \lopacity{}.

\begin{lemma} \label{lemma:lop-not-aca}
    There exists a \lopaque{} history $H$ that does not avoid cascading aborts.
\end{lemma}

\begin{proof}[sketch]
    \rlemma{lemma:h1-lop} demonstrates that $H_1$ is \lopaque{}. However, since
    $\tr_i$ reads from $\tr_j$ in $\hist_1$ and $\frop{j}{\obj}{\val}
    \prec_{\hist_1} \tryC_i\!\to\!\co_i$ the history is not ACA, since it
    contradicts \rdef{def:aca-short}.
\end{proof}

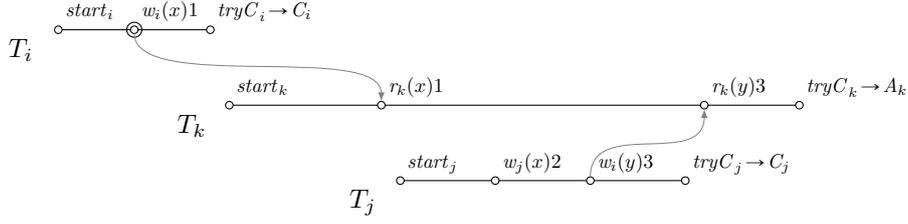
\begin{figure} 
\begin{center}
\begin{tikzpicture}
     \draw
           (0,3)        node[tid]       {$\tr_i$}
                        node[aop]       {$\init_i$} %
                        node[dot]       {} 

      -- ++(1,0)        node[aop]       {$\twop{i}{\objx}{1}$}
                        node[dot] (wix) {}
                        node[cir]       {}

      -- ++(1,0)        node[aop]       {$\tryC_i\!\to\!\co_i$}
                        node[dot]       {}         
                        ;

     \draw
           (2.25,2)        node[tid]       {$\tr_k$}
                        node[aop]       {$\init_k$} %
                        node[dot]       {} 

      -- ++(2,0)        node[aop]       {$\trop{k}{\objx}{1}$}
                        node[dot] (rkx) {}  

      -- ++(4.25,0)     node[aop]       {$\trop{k}{\objy}{3}$}
                        node[dot] (rky) {}

      -- ++(1.25,0)     node[aop]       {$\tryC_k\!\to\!\ab_k$}
                        node[dot]       {}         
                        ;

     \draw
            (4.5,1)    node[tid]       {$\tr_{j}$}
                        node[aop]       {$\init_{j}$} %
                        node[dot]       {} 

      -- ++(1.25,0)     node[aop]       {$\twop{j}{\obj}{2}$}
                        node[dot] (wjx) {}

      -- ++(1.25,0)     node[aop]       {$\twop{i}{\objy}{3}$}
                        node[dot] (wjy) {}

      -- ++(1.25,0)     node[aop]       {$\tryC_j\!\to\!\co_j$}
                        node[dot]       {}
                        ;
     
     \draw[hb] (wix) \squiggle (rkx);
     \draw[hb] (wjy) \squiggleup (rky);
\end{tikzpicture}
\end{center}
\captionof{figure}{\label{fig:aca-history-not-lopaque} ACA history that does not satisfy elastic \lopacity{}.}
\end{figure}

\begin{lemma} \label{lemma:aca-not-lop}
    There exists an ACA history $H$ that is not \lopaque{}.
\end{lemma}

\begin{proof}[sketch]
    The history in \rfig{fig:aca-history-not-lopaque} is shown to be ACA in \cite{AH14}.
    However, note, that $\compl{H} = H$, and given any sequential $S \equiv
    \compl{\hist}$ $\tr_k$ $\tr_k$ must follow both $\tr_i$ and $\tr_k$ in $S$
    because $\tr_k$ reads form both transactions.
    Since $\tr_i \prec^{rt}_H \tr_j$ and $\tr_i \prec^{rt}_H \tr_k$, then
    $\tr_i$ must precede both other transactions in $S$.
    Hence, $S = H|\tr_i \cdot H|\tr_j \cdot H|\tr_k$, so $\vis{S}{\tr_k} = S$
    and therefore $\vis{S}{\tr_k}$ is illegal because $\frop{k}{\obj}{1}$ is
    preceded in $\vis{S}{\tr_k}|\obj$ by $\frop{k}{\obj}{2}$.
\end{proof}

Strictness and \lopacity{} are also incomparable.

\begin{theorem}
    There exists a \lopaque{} history $H$ that is not strict.
\end{theorem}
 
\begin{proof}[sketch]
    Since any strict history is also ACA \cite{AH14}, and since
    \rlemma{lemma:lop-not-aca} shows that not all \lopaque{} histories are ACA,
    then not all \lopaque{} histories are strict.
\end{proof}
 
\begin{theorem}
    There exists a strict history $H$ that is not \lopaque{}.
\end{theorem}
 
\begin{proof}[sketch]
    The history in \rfig{fig:aca-history-not-lopaque} is shown to be strict in
    \cite{AH14}.
    However, as we show in \rlemma{lemma:aca-not-lop}, this history is not
    \lopaque{}.
\end{proof}

Rigorousness is strictly stronger than \lopacity{}.

\begin{lemma}
    Any rigorous history $H$ is \lopaque{}.
\end{lemma}

\begin{proof}[sketch]
Since \cite{AH14} demonstrates that rigorous histories are opaque, and since we
show in \rlemma{lemma:all-opaque-are-lopaque} that opaque histories are also
\lopaque{}, then all rigorous histories are \lopaque{}.
\end{proof}

Live opacity is stronger than \lopacity{}.

\begin{lemma}
    Any live opaque history $\hist$ is \lopaque{}.
\end{lemma}

\begin{proof}[sketch]
    Since $\hist$ is live opaque there exists a sequential history $S$ that
    justifies serializability of $\hist$ and an extension $S'$ of $S$ where if
    transaction  $\tr_i$  is not in $S$ then it is replaced in $S'$ by
    $\tr_i^{gr}$ containing only non-local reads. 
    $S'$ is legal and preserves the real-time order of $\hist$ (accounting for
    replaced transactions).
    In addition, from \rthm{thm:live-opacity-aborting}, no transaction in
    $\hist$ reads from a live transaction (in any prefix of $\hist$). 
    Therefore, since $S'$ is legal, any read operation $\op_i =
    \frop{i}{\obj}{\val}$ in $\hist$ that is preceded
    $\fwop{j}{\obj}{\val}{\valu}$ in $\hist$,  $\tr_j$ is committed in $S$ and
    is included in $S'$ in full.

    Let $S''$ be a sequential history constructed by replacing the operations
    removed to create $S'$ where if $\tr_i \in \hist$ and $\tr_i \not\int S$
    then $\tr_i$ is aborted in $S''$. $S''$ preserves the real time order of
    $\hist$ and $S'' \equiv \hist$.
    Note that, since $S'$ is legal, if some write $\op^w$ is in $S''$ and not
    in $S'$, then there is no non-local read operation $\op^r$ reading the value
    written by $\op^w$. Hence any operation reading the value written by
    $\op^w$ is local, and since all local reads in transactions that are
    replaced in $S'$ read legal values (by \rdef{def:live-opacity}), then all
    reads reading from any $\op_w$ read legal values in $S''$.
    Since $S'$ is legal, then all reads reading from transactions that are in
    $S$ read legal values in $S'$. Since $S'' \equiv \hist$, then these read
    and write operations also read legal values in $S''$.
    Because of this, and since no transaction reads from another live
    transaction, $\vis{S''}{\tr_i}$ will be legal for any transaction in $S''$. 
    In addition, $\luvis{S''}{\tr_i}$ will be legal for any aborted transaction
    in $S''$.
    Therefore any live opaque $\hist$ will be final state last-use opaque. 
    Since any prefix of $\hist$ is also live opaque, then any prefix will also
    be final-state last-use opaque, hence $\hist$ is opaque.
\end{proof}

\section{$\beta$--last-use opacity}

\begin{definition}[\Blast{} Write Invocation] \label{def:blast-write-inv} Given
a program $\prog$, a set of processes $\processes$ executing $\prog$ and a
history $\hist$ s.t. $\hist \models \exec{\prog}{\processes}$, i.e.  $\hist \in
\evalhist{\prog}{\processes}$, an invocation $\inv{i}{}{\wop{\obj}{\val}}$ is
the \emph{\last{} write invocation} on some variable $\obj$ by transaction
$\tr_i$ in $\hist$, if for any history $\hist' \in
\evalhist{\prog}{\processes}$ for which $\hist$ is a prefix (i.e., $\hist' =
\hist \cdot R$) there is no operation invocation $\mathit{inv}'$ s.t.
$\inv{i}{}{\wop{\obj}{\val}}$ precedes $\mathit{inv}'$ in
$\hist'|\tr_i$ where 
\begin{inparaenum}[(a) ]
    \item $\mathit{inv}' = \inv{i}{}{\wop{\obj}{\valu}}$ or
    \item $\mathit{inv}' = \inv{i}{}{\tryA}$.
\end{inparaenum}
\end{definition}

\begin{definition}[\Blast{} Write] \label{def:blast-write-op}
Given program $\prog$, a set of processes $\processes$ executing $\prog$ and
history $\hist$ s.t. $\hist \models \exec{\prog}{\processes}$, an operation
execution is the \emph{\last{} write} on some variable $\obj$ by transaction
$\tr_i$ in $\hist$ if it comprises of an invocation and a response other than
$\ab_i$, and the invocation is the \emph{\blast{} write invocation} on $\obj$ by
$\tr_i$ in $\hist$.
\end{definition}

\begin{definition} [Transaction $\beta-$Decided on $\obj$] \label{def:bdecided} 
Given a program $\prog$, a set of processes $\processes$ and a history $\hist$
s.t. $\hist \models \exec{\prog}{\processes}$, we say transaction $\tr_i \in
\hist$ \emph{$\beta$-decided on} variable $\obj$ in $\hist$ iff $\hist|\tr_i$ contains
a complete write operation execution $\fwop{i}{\obj}{\val}{\ok_i}$ that is the
\blast{} write on $\obj$.
\end{definition}

Given some history $\hist$, let $\cpetrans{\hist}_\beta$ be a set of transactions
s.t. $\tr_i \in \cpetrans{\hist}_\beta$ iff there is some variable $\obj$
s.t. $\tr_i$ $\beta$-decided on $\obj$ in $\hist$.

Given any $\tr_i \in \hist$, $\hist\cpe^\beta\tr_i$ is the longest subsequence of
$\hist|\tr_i$ s.t.:
\begin{enumerate}[a) ]
    \item $\hist\cpe^\beta\tr_i$ contains $\init_i\to\valu$, and
    \item for any variable $\obj$, if $\tr_i$ $\beta$-decided on $\obj$ in $\hist$, then
    $\hist\cpe^\beta\tr_i$ contains $(\hist|\tr_i)|\obj$.
\end{enumerate}
In addition, $\hist\cpeC^\beta\tr_i$ is a sequence s.t. $\hist\cpeC^\beta\tr_i =
\hist\cpe^\beta\tr_i \cdot [\tryC_i\to\co_i]$.

Given a sequential history $S$ s.t. $S\equiv\hist$,
$\bluvis{S}{\tr_i}$ is the longest subhistory of $S$, s.t. for each $\tr_j
\in S$:
\begin{enumerate}[a) ]
    \item $S|\tr_j \subseteq \bluvis{S}{\tr_i}$ if $i=j$ or $\tr_j$ is committed
    in $S$, or 
    \item $S\cpeC^\beta\tr_j \subseteq \bluvis{S}{\tr_i}$ if  $\tr_j$ is not
    committed in $S$ but $\tr_j \in \cpetrans{\hist}_\beta.$
\end{enumerate}

Given a sequential history $S$ and a transaction $\tr_i \in S$, we then say
that transaction $\tr_i$ is \emph{$\beta$--last-use legal in} $S$ if $\bluvis{S}{\tr_i}$
is legal.

\begin{definition} [Final-state $\beta$--Last-use Opacity] \label{def:final-state-blu-opacity} \label{def:fs-blu-opacity}
    A finite history $\hist$ is final-state $\beta$--last-use opaque if, and only if,
    there exists a sequential history $S$ equivalent to any completion of
    $\hist$ s.t., 
    \begin{enumerate}[a) ] 
        \item $S$ preserves the real-time order of $\hist$,
        \item every transaction in $S$ that is committed in $S$ is
        legal in $S$,
        \item every transaction in $S$ that is not committed in $S$ is
        $\beta$--last-use legal in $S$. 
    \end{enumerate}
\end{definition}

\begin{definition} [$\beta$--Last-use Opacity] \label{def:blu-opacity}
    A history $\hist$ is $\beta$--last-use opaque if, and only if, every finite prefix of
    $\hist$ is final-state $\beta$--last-use opaque.
\end{definition}

\section{\SVA{} Correctness Proof}
\label{sec:sva-proof}

Since the values used within writes are under the control of the program
(rather than \SVA{}) we simply assume that they are within the domain of the
variables.

\begin{assumption} [Values Within Domain] \label{obs:within-domain}
    For any transaction $\tr_i$ in any \SVA{} history $\hist$ given variable
    $\obj$ with the domain of $\mathbb{D}$, if $\fwop{i}{\obj}{\val}{\valu} \in
    \hist|\tr_i$, then $\val \in \mathbb{D}$.
\end{assumption}

\begin{definition} [Direct Precedence]
    For operations $\op_i,\op_j \in \hist$, $\op_j \dprec_{\hist} \op_i$ iff
    $\op_j \prec_{\hist} \op_i$ and $\nexists{\op_k\in\hist}$ s.t. $\op_j
    \prec_{\hist} \op_k \prec_{\hist} \op_i$.
\end{definition}

\begin{definition} [Operation Execution Conditional]
    Given predicate $P$ and operation $\op$
    $P \hto \op$ denotes that $P$ is true only if $\op$ executes.
\end{definition}

\begin{definition} [Operation Execution Converse]
    Given predicate $P$ and operation $\op$
    $P \hgets \op$ denotes that $\op$ executes only if $P$ is true.
\end{definition}

Let there be any $\prog, \processes, \hist \models \exec{\prog}{\processes}$,
$\op_i \in \hist|\tr_i$.

\def\lacc#1#2{\ddot{\op}_{#1}^{#2}}

\begin{definition} \label{d:last-access} \label{d1}
    $\op_i$ is \last{} access on $\obj$ in $\tr_i$, denoted $\op_i = \lacc{i}{\obj}$ if both:
    \begin{enumerate}[a) ]
        \item $\op_i$ is \last{} read on $\obj$ in $\tr_i$ or 
              $\op_i$ is \last{} write on $\obj$ in $\tr_i$, and
        \item $\nexists{\op_i' \in \hist}$ s.t. $\op_i \prec_\hist \op_i'$ and
                $\op_i'$ is \last{} read on $\obj$ in $\tr_i$ or 
                $\op_i'$ is \last{} write on $\obj$ in $\tr_i$.
    \end{enumerate}
\end{definition}

Let there be any $\prog, \processes, \hist \models \exec{\prog}{\processes}$,
$\op_i \in \hist|\tr_i$, $\op_i = \begin{cases} \frop{i}{\obj}{\val}, \\
\fwop{i}{\obj}{\val}{\ok_i}. \end{cases}$

\begin{lemma} [Access Condition] \label{l:access-cond} \label{o5}
    $\lv{\obj} = \pv{i}{\obj} - 1 \hgets \op_i$.
\end{lemma}

\begin{proof}
    Condition at \pplineref{sva}{call-access-cond} dominates access
    at \pplineref{sva}{call-access}.
\end{proof}

\begin{lemma} [Abort Condition] \label{l:abort-cond} \label{o7}
    $\ltv{\obj} = \pv{i}{\obj} - 1 \hgets \res{i}{}{\ab_i}$.
\end{lemma}

\begin{proof}
    Access condition at \pplineref{sva}{abort-dismiss-access-cond} dominates
    dismiss at \pplineref{sva}{abort-dismiss} in procedure abort %
    for each
    variable. Hence, all variables must pass
    \pplineref{sva}{abort-dismiss-access-cond} before abort %
    concludes.
\end{proof}

\begin{lemma} [Commit Condition] \label{l:commit-cond} \label{o8}
    $\ltv{\obj} = \pv{i}{\obj} - 1 \hgets \res{i}{}{\co_i}$.
\end{lemma}

\begin{proof}
    By analogy to \rlemma{l:abort-cond}.
\end{proof}

\begin{lemma} [Early Release] \label{l:early-release} \label{o9}
    If $\op_i = \lacc{i}{\obj}$ then $\lv{\obj} = \pv{i}{\obj} \hto \op_i$.
\end{lemma}

\begin{proof}
    $\lv{\obj}$ can be set by $\tr_i$ at \pplineref{sva}{early-release-lv} and
    at \pplineref{sva}{dismiss-release}. The former is set during the last
    access on some $\obj$ in $\tr_j$ (\pplineref{sva}{early-release-cond}
    dominates \pplineref{sva}{early-release-lv}). The latter is set during
    commit, which means that if any \last{} access was present, it was executed
    prior to commit. 
\end{proof}

Let $r_i = \begin{cases} 
                    \res{i}{}{\ab_i}, \\
                    \res{i}{}{\co_i}. 
               \end{cases}$

\begin{lemma} [Release] \label{l:release} \label{o10}
    If $\nexists \op_i'\in\hist|\tr_i$ s.t. $\op_i'= \lacc{i}{\obj}$ 
    and $\obj \in \aset{\tr_i}$
    then $\lv{\obj} = \pv{i}{\obj} \hto r_i$.
\end{lemma}

\begin{proof}
    If $\op_i'$ is not \last{} access then \pplineref{sva}{early-release-cond}
    will not be passed, so only assignment of $\lv{\obj}$ is in
    \pplineref{sva}{dismiss-release} which execute only during commit or
    abort.
\end{proof}

\begin{lemma} [Terminal Release] \label{l:trelease} \label{o11}
    If $\obj \in \aset{\tr_i}$
    then $\lv{\obj} = \pv{i}{\obj} \hto r_i$.
\end{lemma}

\begin{proof}
    $\ltv{\obj}$ can be set only in \pplineref{sva}{restore-ltv-set} or
    \pplineref{sva}{commit-ltv-set}, which are part of abort and commit,
    respectively.
\end{proof}

Let there be any $\hist$, $\tr_i \in \hist$, $\tr_j \in \hist$, $\op_i \in
\hist|\tr_i$, $\op_i = \frop{i}{\obj}{\valu}$, $\op_j \in \hist|\tr_j$, $\op_j =
\fwop{j}{\obj}{\val}{\ok_j}$.

\begin{lemma} [No Buffering] \label{l:no-buffer} \label{o1}
    If $\op_j \dprec_{\hist|\obj} \op_i$ and not $\op_j \prec_{\hist} \res{j}{}{\ab_j}
    \prec_{\hist} \op_i$ then $u=v$.
\end{lemma}

\begin{lemma} [Revert On Abort] \label{o2} \label{l:revert-on-abort}
    If $\op_j \dprec_{\hist|\obj} \op_i$ and $\op_j \prec_{\hist} \res{j}{}{\ab_j}
    \prec_{\hist} \op_i$ then $u\neq v$. 
\end{lemma}

\begin{proof}
    If abort %
    is executed then the restore procedure is executed for all
    $\obj \in \aset{\tr_i}$. Thus, \pplineref{sva}{restore-revert} restores
    $\obj$ to value $\val'$ which is acquired before the any operation on
    $\obj$ is executed by $\tr_i$, hence $\val'\neq\val$, so $\valu\neq\val$.
\end{proof}

Let $\hist|\init$ be a subhistory of $\hist$ that for each $\tr_j \in \hist$
contains only the operation $\init_j$.

\begin{lemma} [Consecutive Versions] \label{l:consecutive-versions} \label{o3}
    If $\obj \in \aset{\tr_i} \cap \aset{\tr_j}$
    and $\inv{i}{}{\init_i} \dprec_{\hist|\init} \inv{i}{}{\init_j}$ 
    then $\pv{i}{\obj} -1 = \pv{j}{\obj}$.
\end{lemma}

\begin{proof}
    If $\tr_i$ returns at \pplineref{sva}{lock-x} for $\obj$ then no $\tr_j$
    s.t. $\obj \in \aset{\tr_j}$ returns at \pplineref{sva}{lock-x} until
    $\tr_i$ executes \pplineref{sva}{unlock-x} for $\obj$.
    Hence, $\tr_i$  alone increments $\gv{\obj}$ at
    \pplineref{sva}{start-set-gv} and sets $\pv{i}{\obj}$ to the new value of
    $\gv{\obj}$. 
    If $\init_i \dprec_{\hist|\init} \init_j$ then $\tr_i$ will return at
    \pplineref{sva}{lock-x}  and $\tr_j$ will return next. No other transaction
    will return  at \pplineref{sva}{lock-x} between $\tr_i$ and $\tr_j$.
\end{proof}

\begin{lemma} [Unique Versions] \label{l:unique-versions} \label{o4}
    If $\obj \in \aset{\tr_i} \cap \aset{\tr_j}$
    then $\pv{i}{\obj} \neq \pv{j}{\obj}$.
\end{lemma}

\begin{proof}
    From \rlemma{l:consecutive-versions}.
\end{proof}

\begin{lemma} [Monotonic Versions] \label{l:monotonic-versions} \label{o13}
    If $\obj \in \aset{\tr_i} \cap \aset{\tr_j}$
    and $\inv{i}{}{\init_i} \prec_{\hist|\init} \inv{i}{}{\init_j}$ 
    then $\pv{i}{\obj} < \pv{j}{\obj}$.
\end{lemma}

\begin{proof}
    From \rlemma{l:consecutive-versions}, \rlemma{l:unique-versions}.    
\end{proof}

\begin{definition} [Version Order] \label{obs:order-variable}
    Let $\prec_\obj$ be an order s.t. $\tr_i \prec_\obj \tr_j$ iff
    $\pv{i}{\obj} < \pv{j}{\obj}$.
\end{definition}

\begin{lemma} [Forced Abort Condition] \label{l:forced-abort-cond} \label{o14}
    $\rv{i}{\obj} < \cv{\obj} \hto \res{i}{}{\ab_i}$.
\end{lemma}

\begin{proof}
    Condition at \pplineref{sva}{call-rollback-cond} dominates abort at
    \pplineref{sva}{call-rollback}.
    Condition at \pplineref{sva}{commit-rollback-cond} dominates abort at
    \pplineref{sva}{commit-rollback}.
\end{proof}

Let there be any $\prog, \processes, \hist \models \exec{\prog}{\processes}$,
$\op_i \in \hist|\tr_i$, $\op_i = \begin{cases} \frop{i}{\obj}{\val}, \\
\fwop{i}{\obj}{\val}{\ok_i}. \end{cases}$

\begin{lemma} \label{o15}
    $\cv{\obj} < \rv{i}{\obj} \hgets \op_i$.
\end{lemma}

\begin{proof}
    Condition at \pplineref{sva}{call-rollback-cond} dominates abort at
    \pplineref{sva}{call-rollback}.
\end{proof}

\begin{lemma} [Current Version Early Release] \label{l:current-version-early-release} \label{o16}
    If $\op_i = \lacc{i}{\obj}$ then $\cv{\obj} = \rv{i}{\obj} \hto \op_i$.
\end{lemma}

\begin{proof}
    By analogy to \rlemma{o9}.
\end{proof}

\begin{lemma} [Current Version Release] \label{l:current-version-release} \label{o17}
    If $\nexists \op_i\in\hist|\tr_i$ s.t. $\op_i= \lacc{i}{\obj}$ 
    and $\obj \in \aset{\tr_i}$
    then $\cv{\obj} = \rv{i}{\obj} \hto r_i$.
\end{lemma}

\begin{proof}
    By analogy to \rlemma{o10}.
\end{proof}

\begin{lemma} \label{o18}
    $\cv{\obj} = \rv{i}{\obj} \hgets \res{i}{}{\ab_i}$.
\end{lemma}

\begin{proof}
    From \rlemma{o10} and \rlemma{o17}.
\end{proof}

Let there be any $\prog, \processes, \hist \models \exec{\prog}{\processes}$,
$\tr_i \in \hist$, $\tr_j \in \hist$, 
$\op_i \in \hist|\tr_i$, $\op_j \in \hist|\tr_j$, 
$\op_i = \begin{cases} 
         \frop{i}{\obj}{\val}, \\ 
         \fwop{i}{\obj}{\val}{\ok_i},
         \end{cases}$
$\op_j = \begin{cases} 
         \frop{j}{\obj}{\val}, \\ 
         \fwop{j}{\obj}{\val}{\ok_j}.
         \end{cases}$

\begin{lemma} [Access Order] \label{l:access-order} \label{o19}
    $\pv{i}{\obj} < \pv{j}{\obj} \Leftrightarrow \op_i \prec_\hist \op_j$.
\end{lemma}

\begin{proof}
    From \rlemma{o5} and \rlemma{o13}.
\end{proof}

Let there be any $\hist$, $\tr_i \in \hist$, $\tr_j \in \hist$, $\op_i \in
\hist|\tr_i$, $\op_i = \frop{i}{\obj}{\valu}$, $\op_j \in \hist|\tr_j$, $\op_j =
\fwop{j}{\obj}{\val}{\ok_j}$.

Let there be any $\prog, \processes, \hist \models \exec{\prog}{\processes}$.

\begin{lemma} [Access Prefix] \label{o20}
    If $\lv{\obj} = \pv{j}{\obj}$
    then $\forall{\tr_k \in \hist}$ s.t. $\pv{k}{\obj} < \pv{i}{\obj}$
    either $\res{k}{}{\co_k} \in \hist|\tr_k$,
           $\res{k}{}{\ab_k} \in \hist|\tr_k$, or
           $\lacc{k}{\obj} \in \hist|\tr_k$.
\end{lemma}

\begin{proof}
    \begin{align}
        & \forall{\tr_l,\tr_k \in \hist}~\text{s.t.}~\pv{l}{\obj} =
          \pv{k}{\obj}-1: \label{e:o20:def-ver} \\
        & \text{\rlemma{o5}} \so \lv{\obj} = \pv{k}{\obj} - 1 \hgets \op_k 
          \label{e:o20:o5} \\
        & \pref{e:o20:def-ver}\wedge\pref{e:o20:o5} \so 
          \lv{\obj} = \pv{l}{\obj} \hgets \op_k \\
        & \text{\rlemma{o9}} \so \lv{\obj} = \pv{l}{\obj} \hto \lacc{l}{\obj} 
          \label{e:o20:dec} \\
        & \text{\rlemma{o10}} \so \lv{\obj} = \pv{l}{\obj} \hto r~\text{where}~
          r=\res{i}{}{\ab_i}~\text{or}~r=\res{i}{}{\co_i} 
          \label{e:o20:finish}\\
        & \pref{e:o20:dec} \wedge \pref{e:o20:finish} \so 
          \tr_l~\text{is committed, aborted or decided on}~\obj 
    \end{align}
    Trivially extends for any $\tr_l,\tr_k$ s.t. $\pv{l}{\obj} < \pv{k}{\obj}$. 
\end{proof}

\begin{lemma} [C] \label{o21}
    If $\ltv{\obj} = \pv{j}{\obj}$
    then $\forall{\tr_k \in \hist}$ s.t. $\pv{k}{\obj} < \pv{i}{\obj}$
    either $\res{k}{}{\co_k} \in \hist|\tr_k$, or
           $\res{k}{}{\ab_k} \in \hist|\tr_k$.
\end{lemma}

\begin{proof}
    \begin{align}
        & \forall{\tr_l,\tr_k \in \hist}~\text{s.t.}~\pv{l}{\obj} =
          \pv{k}{\obj}-1: \label{e:o21:def-ver} \\
        & \text{\rlemma{o8}} \so \ltv{\obj} = \pv{k}{\obj} - 1 \hgets \res{k}{}{\co_i}
          \label{e:o21:o8} \\
        & \pref{e:o21:def-ver}\wedge\pref{e:o21:o8} \so 
          \ltv{\obj} = \pv{l}{\obj} \hgets \op_k \\
        & \text{\rlemma{o11}} \so \ltv{\obj} = \pv{l}{\obj} \hto r~\text{where}~
          r=\res{i}{}{\ab_i}~\text{or}~r=\res{i}{}{\co_i} 
          \label{e:o21:finish}\\
        & \pref{e:o21:finish} \so 
          \tr_l~\text{is committed or aborted} 
    \end{align}
    Trivially extends for any $\tr_l,\tr_k$ s.t. $\pv{l}{\obj} < \pv{k}{\obj}$. 
\end{proof}

Let there be any $\hist$, $\tr_i \in \hist$, $\tr_j \in \hist$, $\op_i \in
\hist|\tr_i$, $\op_i = \frop{i}{\obj}{\valu}$, $\op_j \in \hist|\tr_j$, $\op_j =
\fwop{j}{\obj}{\val}{\ok_j}$.

\begin{lemma} [Forced Abort] \label{l:forced-abort} \label{o12}
    If $\obj \in \aset{\tr_i} \cap \aset{\tr_j}$ 
    and $\res{j}{}{\ab_j} \in \hist|\tr_j$
    and $\op_i \prec_\hist \res{j}{}{\ab_j}$
    then $\res{i}{}{\ab_i} \in \hist|\tr_i$.
\end{lemma}

\begin{proof}
    \begin{align}
        & \res{j}{}{\ab_j} \in \hist|\tr_j \wedge \rlemma{o18} 
          \so \cv{\obj} = pv{j}{\obj} \hgets \res{j}{}{\ab_j} \label{e:o12:cv} \\
        & \rlemma{o13} \so \pv{j}{\obj} < \pv{i}{\obj} \label{e:o12:o13} \\
        & \pref{e:o12:cv} \wedge \pref{e:o12:o13} 
          \so \cv{\obj} < pv{i}{\obj} \hgets \res{j}{}{\ab_j} \label{e:o12:o13:cv} \\
        & \rlemma{o15} \so \cv{\obj}= \rv{i}{\obj} \hgets \op_i
          \so \rv{i}{\obj} = \pv{j}{\obj} \label{e:o12:rv} \\
        & \pref{e:o12:rv} \wedge \pref{e:o12:o13} \so \rv{i}{\obj} > \pv{i}{x} 
          \label{e:o12:rv-gt} \\
        & \pref{e:o12:rv-gt} \wedge \pref{e:o12:o13:cv} 
          \so \rv{i}{\obj} < \cv{x} \label{e:o12:rv-lt} \\
        & \pref{e:o12:rv-lt} \so \rv{i}{\obj} < \cv{\obj} \hto \res{i}{}{\ab_i}          
          \so \res{i}{}{\ab_i} \in \hist|\tr_i
    \end{align}
\end{proof}

Let there be any $\prog, \processes, \hist \models \exec{\prog}{\processes}$,
$\tr_i \in \hist$, $\tr_j \in \hist$, 
$\op_i \in \hist|\tr_i$, $\op_j \in \hist|\tr_j$, 
$\op_i = \begin{cases} 
         \frop{i}{\obj}{\val}, \\ 
         \fwop{i}{\obj}{\val}{\ok_i},
         \end{cases}$
$\op_j = \begin{cases} 
         \frop{j}{\obj}{\val}, \\ 
         \fwop{j}{\obj}{\val}{\ok_j}.
         \end{cases}$

\begin{definition} [Completion Construction] \label{d3} \label{d:completion-construction}
    $\hist_C = \compl{\hist}$ s.t. $\forall{\tr_k \in \hist},$
    $\res{k}{}{\co_k} \not\in \hist|\tr_k \Leftrightarrow \res{k}{}{\ab_k} \in
    \hist_C|\tr_k$
\end{definition}

\begin{definition} [Sequential History Construction] \label{d4} \label{d:seqh-construction}
    $\seqH$ is a sequential history s.t. $\seqH \equiv \hist_C$ and $\tr_i
    \prec_{\hist_C} \tr_j \so \tr_i \prec_{\seqH} \tr_j$ and $\tr_i
    \prec_{\obj} \tr_j \so \tr_i \prec_{\seqH} \tr_j$.
\end{definition}

Let there be any $\hist$, $\tr_i \in \hist$, $\tr_j \in \hist$, $\op_i \in
\hist|\tr_i$, $\op_i = \frop{i}{\obj}{\valu}$, $\op_j \in \hist|\tr_j$, $\op_j =
\fwop{j}{\obj}{\val}{\ok_j}$.

\begin{lemma} \label{l1}
    If $\tr_i$ reads $\obj$ from $\tr_j$ then $\tr_j$ is committed in $\hist$
    or $\tr_j$ is decided on $\obj$ in $\hist$.
\end{lemma} 

\begin{proof}
    \begin{align}
        & \tr_i~\text{reads}~\obj~\text{from}~\tr_j 
          \so \op_i = \frop{i}{\obj}{\val} \wedge
          \op_j=\fwop{j}{\obj}{\val}{\ok_i} \wedge \op_j \prec_\hist \op_i \\
        & \rlemma{o5} \so \lv{\obj} = \pv{\obj} - 1 \hgets \op_i \\
        & \rlemma{o19} \wedge \op_j \prec_\hist \op_i 
          \so \pv{j}{\obj} < \pv{i}{\obj} \label{e:l1:pv} \\
        & \pref{e:l1:pv} \wedge \rlemma{o20} \so \tr_j~\text{is committed,
        aborted, or decided on}~\obj
    \end{align}
    
    Let us assume that $\tr_j$ is aborted:
    
    \begin{align}
          \op_i \prec \res{j}{}{\ab_j}:\;
        & \rlemma{o3} \so v \neq v \label{e:l1:vv}
          \so ~\text{contradiction} \\
          \res{j}{}{\ab_j} \prec_\hist \op_i:\;
        & \rlemma{o9} \so \lv{\obj} = \pv{\obj} \hto \op_i 
          \wedge \op_i = \lacc{i}{\obj}
    \end{align}

    Thus, $\tr_i$ is committed or decided on $\obj$.
\end{proof}

\begin{corollary} \label{c:l1}
    If $P$ is any prefix of $\hist$, then if $\tr_i$ reads $\obj$ from $\tr_j$
    in $P$ then $\tr_j$ is committed in $P$ or $\tr_j$ is decided on $\obj$ in
    $P$.
\end{corollary} 

\begin{lemma} \label{l2}
    If $\tr_i$ reads $\obj$ from $\tr_j$ and $\tr_j$ is committed in $\hist$
    then $\tr_j$ is committed in $\hist$.
\end{lemma} 

\begin{proof}
    \begin{align}
        & \tr_i~\text{reads}~\obj~\text{from}~\tr_j 
          \so \op_i = \frop{i}{\obj}{\val} \wedge
          \op_j=\fwop{j}{\obj}{\val}{\ok_i} \wedge \op_j \prec_\hist \op_i \\
        & \text{\rlemma{o8}} \so \ltv{\obj} = \pv{k}{\obj} - 1 \hgets \res{i}{}{\ok_i}
          \label{e:l2:o8} \\
        & \text{\rlemma{o19}}\wedge \op_j \prec_\hist \op_i 
          \so \pv{j}{\obj} < \pv{i}{\obj} \label{e:l2:pv} \\
        & \text{\rlemma{o21}}\wedge \pref{e:l2:o8} \wedge \pref{e:l2:pv} 
          \so r \in \hist|\tr_j~\text{where}~
          r=\res{j}{}{\ab_j}~\text{or}~r=\res{j}{}{\co_j} \label{e:l2:r} \\
        & \text{\rlemma{o12}} \so \text{if}~\res{j}{}{\ab_j} \in \hist|\tr_j 
          ~\text{then}~ \res{j}{}{\ab_j} \in \hist|\tr_j \so ~\text{contradition} 
          \label{e:l2:cont} \\
        & \pref{e:l2:cont} \so \res{i}{}{\ab_i} \in \hist|\tr_i       
    \end{align}    
\end{proof}

Let there be any $\prog, \processes, \hist \models \exec{\prog}{\processes}$,
$\tr_i \in \hist$, $\op_i = \frop{i}{\obj}{\val}$, $\op_i \in \hist|\tr_i$.

\begin{lemma} \label{lemma:transactions-in-vis} \label{l3}
    If $\res{i}{}{\co_i} \in \seqH|\tr_i$ then
    $\exists\op_j = \fwop{j}{\obj}{\val}{\ok_j} \in \vis{\seqH}{\tr_i}.$
\end{lemma}

\begin{proof}
    If $i=j$ then trivially $\op_j \in \vis{\seqH}{\tr_i}$.
    Otherwise:
    \begin{align}
        & i\neq j \wedge \text{\rlemma{l:no-buffer}}  
          \so \exists \tr_j \wedge \op_j \in \hist_C|\tr_j 
          \label{e:l3:rf} \\
        & \pref{e:l3:rf} \wedge \text{\rlemma{l2}} \wedge 
          \res{i}{}{\co_i} \in \hist_C|\tr_i \so  
          \exists \res{j}{}{\co_j} \in \hist_C|\tr_j 
          \label{e:l3:co} \\
        & \text{\rdef{d:seqh-construction}} \wedge \pref{e:l3:co} 
          \so \res{j}{}{\co_j} \in \seqH|\tr_j \wedge \tr_j \prec_{\seqH} \tr_i
          \label{e:l3:co-co} \\
        & \pref{e:l3:co-co} \so \seqH|\tr_j \subseteq \vis{\seqH}{\tr_i} \so
          \op_j \in \vis{\seqH}{\tr_i}  
    \end{align}
\end{proof}

\begin{lemma} \label{lemma:transactions-in-lvis} \label{l4}
    $\exists\op_j = \fwop{j}{\obj}{\val}{\ok_j} \in \luvis{\seqH}{\tr_i}.$
\end{lemma}

\begin{proof}
    If $i=j$ then trivially $\op_j \in \luvis{\seqH}{\tr_i}$.
    Otherwise: 
    \begin{align}
        & i\neq j \wedge \text{\rlemma{l:no-buffer}}  
          \so \exists \tr_j \wedge \op_j \in \hist_C|\tr_j 
          \label{e:l4:rf} \\
        & \pref{e:l4:rf} \wedge \text{\rlemma{l1}} \wedge 
          \res{i}{}{\co_i} \in \hist_C|\tr_i \so  
               \text{either}~ 
               \exists \res{j}{}{\co_j} \in \hist_C|\tr_j ~\text{or}~
               \exists \lacc{j}{\obj} \in \hist_C|\tr_j
          \label{e:l4:co} \\
        & \text{\rdef{d:seqh-construction}} \wedge \pref{e:l4:co} 
          \so \res{j}{}{\co_j} \in \seqH|\tr_j \wedge \tr_j \prec_{\seqH} \tr_i
          \label{e:l4:co-co} \\
        & \pref{e:l4:co-co} \so 
          \seqH|\tr_j \subseteq \luvis{\seqH}{\tr_i} \so
          \op_j \in \luvis{\seqH}{\tr_i} \label{e:l4:commit} \\
        & \pref{e:l4:co} \so \lacc{j}{\obj} \in \seqH|\tr_j \wedge \tr_j
          \prec_{\seqH} \tr_i \label{e:l4:cpe} \\
        & \pref{e:l4:cpe} \wedge \seqH\cpe\tr_j \subseteq \luvis{\seqH}{\tr_i}
          \so \op_j \in \luvis{\seqH}{\tr_i} \label{e:l4:last-use}
    \end{align}
\end{proof}

\begin{lemma} \label{lemma:operations-in-vis}
    Given $\seqH$ and any two transactions $\tr_i, \tr_j \in \seqH$ s.t. there
    is an operation execution $\fwop{j}{\obj}{\val}{\ok_j} \in \seqH|\tr_j$ and
    $\frop{i}{\obj}{\val} \in \seqH|\tr_i$ then there is no operation
    $\fwop{k}{\obj}{\valu}{\ok_k}$ (executed by some $\tr_k \in \seqH$) in
    $\vis{\seqH}{\tr_i}$ s.t. $\fwop{k}{\obj}{\valu}{\ok_k}$ precedes
    $\frop{i}{\obj}{\val}$ in $\vis{\seqH}{\tr_i}$ and follows
    $\fwop{j}{\obj}{\val}{\ok_j}$ in $\vis{\seqH}{\tr_i}$.
\end{lemma}

\begin{proof}
    For the sake of contradiction, assume that $\op_k$ exists as specified.
     
    If $k=i$, then $\op_k \prec_{\hist|\tr_i} \op_i$, which contradicts
    \rlemma{l:no-buffer} (assuming unique writes).
    
    If $k=j$, then from \rlemma{l1} $\tr_j$ is either committed or decided on
    $\obj$ in $\seqH$. 
    If $\tr_i$ commits, then $\op_i$ reading $\val$ contradicts \rlemma{l:no-buffer}.
    If $\tr_i$ does not commit in $P$, then this contradicts \rlemma{l2}.

    Otherwise, $\exists \tr_k \in \hist$ s.t. $\op_k \in \hist|\tr_k$ 
    from \rlemma{l1} $\tr_j$ is either committed or decided on
    $\obj$ in $\seqH$ and from \rlemma{l2} $\tr_k$ is committed
    in $\hist$.
    Since $\tr_k$ commits, this contradicts \rlemma{l:no-buffer}.
\end{proof}

\begin{lemma} \label{lemma:operations-in-lvis}
    Given $\seqH$ and any two transaction $\tr_i, \tr_j \in \seqH$ s.t. there
    is an operation execution $\fwop{j}{\obj}{\val}{\ok_j} \in \seqH|\tr_j$ and
    $\frop{i}{\obj}{\val} \in \seqH|\tr_i$ 
    then there is no operation
    $\fwop{k}{\obj}{\valu}{\ok_k}$ (executed by some $\tr_k \in \seqH$) in
    $\luvis{\seqH}{\tr_i}$ s.t. $\fwop{k}{\obj}{\valu}{\ok_k}$ precedes
    $\frop{i}{\obj}{\val}$ in $\vis{\seqH}{\tr_i}$ and follows
    $\fwop{j}{\obj}{\val}{\ok_j}$ in $\vis{\seqH}{\tr_i}$.
\end{lemma}

\begin{proof}
    By analogy to \rlemma{lemma:operations-in-vis}.
\end{proof}

\begin{lemma} \label{thm:sva-fs-lopacity} \label{lemma:sva-fs-lopacity}
    Any \SVA{} history $\hist$ is final-state last-use opaque.
\end{lemma}

\begin{proof} 
    Given $\seqH$, let $\tr_i \in \seqH$ be any transaction that is
    committed in $\seqH$.
    In that case, from \rlemma{lemma:transactions-in-vis} and
    \rlemma{lemma:operations-in-vis}, every read operation execution
    $\frop{i}{\obj}{\val}$ in $\vis{\seqH}{\tr_i}$ is preceded in
    $\vis{\seqH}{\tr_i}$ by a write operation execution $\fwop{j}{\obj}{\val}{\ok_j}$
    (for some $\tr_j$).
    In addition, from \rass{obs:within-domain}, every write operation execution
    $\fwop{i}{\obj}{\val}{\ok_i}$ in $\vis{\seqH}{\tr_i}$ trivially writes $\val \in
    D$.
    Therefore, for every variable $\obj$, $\vis{\seqH}{\tr_i}|\obj \in
    \seq{\obj}$, so $\vis{\seqH}{\tr_i}$ is legal. 
    Consequently $\tr_i$ in $\seqH$ is legal in $\seqH$.

    Given the same $\seqH$, let $\tr_i \in \seqH$ be any
    transaction that is not committed in $\seqH$ (so it is aborted in $\seqH$).
    From \rlemma{lemma:transactions-in-lvis} and
    \rlemma{lemma:operations-in-lvis}, every read operation execution
    $\frop{i}{\obj}{\val}$ in $\luvis{\seqH}{\tr_i}$ is preceded in
    $\luvis{\seqH}{\tr_i}$ by a write operation execution
    $\fwop{j}{\obj}{\val}{\ok_j}$ (for some $\tr_j$).
    In addition, from \rass{obs:within-domain}, every write operation execution
    $\fwop{i}{\obj}{\val}{\ok_i}$ in $\luvis{\seqH}{\tr_i}$ trivially writes $\val \in
    D$.
    Therefore, for every variable $\obj$, $\luvis{\seqH}{\tr_i}|\obj \in
    \seq{\obj}$, so $\luvis{\seqH}{\tr_i}$ is legal. 
    Thus, $\tr_i$ in $\seqH$ is last-use legal in $\seqH$.

    Since all committed transactions in $\seqH$ are legal in $\seqH$ and since
    all aborted transactions in $\seqH$ are last-use legal in $\seqH$, then, by
    \rdef{def:fs-lopacity} $\hist$ is final-state last use opaque.
\end{proof}

\begin{theorem} \label{thm:sva-lopacity}
    Any \SVA{} history $\hist$ is last-use opaque.
\end{theorem}

\begin{proof}
    Since by \rlemma{lemma:sva-fs-lopacity} any \SVA{} history $\hist$ is
    final-state last-use opaque, and any prefix $P$ of $\hist$ is also an
    \SVA{} history, then every prefix of $\hist$ is also final-state last-use
    opaque. Thus, by \rdef{def:lopacity}, $\hist$ is last-use opaque.
\end{proof}

\end{document}